\documentclass[12pt]{article}

\usepackage{fullpage}
\usepackage{amsmath} 
\usepackage{amsthm}
\usepackage{amssymb}
\usepackage{bbding}
\usepackage[inline]{enumitem}
\PassOptionsToPackage{dvipsnames,svgnames}{xcolor}
\usepackage[numbers,sort&compress]{natbib}
\usepackage[colorlinks,citecolor=blue]{hyperref}
\usepackage{mathrsfs}
\usepackage{calrsfs}
\usepackage{mathtools}
\usepackage{multirow}
\usepackage{prooftree}
\usepackage{algorithm}
\usepackage{algpseudocode}
\usepackage{multicol}
\usepackage{multirow}
\usepackage{graphicx}
\graphicspath{{Images/}}
\usepackage{cmll}
\usepackage{tabu}
\usepackage{tikz}
\usetikzlibrary{calc}
\usetikzlibrary{arrows}
\usetikzlibrary{positioning}
\usetikzlibrary{fit}
\usetikzlibrary{shapes}
\usepackage{thmtools} 
\usepackage{thm-restate}
\usepackage[inline]{enumitem}
\usepackage{tabularx}
\usepackage{subcaption}

\renewcommand{\phi}{\varphi}
\renewcommand{\epsilon}{\varepsilon}
\renewcommand{\emptyset}{\varnothing}
\renewcommand{\leq}{\leqslant}
\renewcommand{\le}{\leqslant}
\renewcommand{\geq}{\geqslant}

\newtheoremstyle{bfnote}%
{}{}
{\itshape}{}
{\bfseries}{.}
{ }{\thmname{#1}\thmnumber{ #2}\thmnote{ (#3)}}
\theoremstyle{bfnote}
\newtheorem{theorem}{Theorem}[section]
\newtheorem{obs}[theorem]{Observation}

\newtheorem{lemma}[theorem]{Lemma}
\newtheorem{definition}[theorem]{Definition}

\newtheorem{example}[theorem]{Example}

\makeatletter
\renewenvironment{proof}[1][\proofname] {\par\pushQED{\qed}\normalfont\topsep6\p@\@plus6\p@\relax\trivlist\item[\hskip\labelsep\bfseries#1\@addpunct{.}]\ignorespaces}{\popQED\endtrivlist\@endpefalse}
\makeatother% 

\newcommand{\defemph}{\textbf}

\newcommand{\dom}{{\sf dom}}
\newcommand{\rng}{{\sf rng}}
\newcommand{\restr}{\upharpoonright}
\newcommand{\prot}{\mathit{Pr}}

\newcommand{\recsend}[2]{{#1}\!\Rightarrow\!{#2}}
\renewcommand{\phi}{\varphi}
\renewcommand{\epsilon}{\varepsilon}

\newcommand{\knowfunc}{{\sf k}}

\newcommand{\knowupd}{{\it update}}

\newcommand{\names}{\mathscr{N}}

\newcommand{\Tterms}{\mathscr{T}}

\newcommand{\vars}{\mathscr{V}}

\newcommand{\qvar}{\mathscr{V}_{q}}
\newcommand{\ivar}{\mathscr{V}_{i}}
\newcommand{\ag}{\mathscr{A}}

\newcommand{\keys}{\mathscr{K}}

\newcommand{\inv}[1]{{#1}^{-1}}

\newcommand{\varsof}{{\sf vars}}

\newcommand{\freevars}{{\sf fv}}
\newcommand{\boundvars}{{\sf bv}}
\newcommand{\subterms}{{\sf st}}

\newcommand{\func}{{\sf f}}
\newcommand{\gunc}{{\sf g}}

\newcommand{\sk}{\mathit{sk}}
\newcommand{\pk}{\mathit{pk}}

\newcommand{\enc}{{\sf enc}}
\newcommand{\dec}{{\sf dec}}

\newcommand{\sign}{{\sf sign}}

\newcommand{\blind}{{\sf blind}}

\newcommand{\positions}[1]{\mathbb{P}(#1)}
\newcommand{\posof}[2]{\mathbb{P}_{#1}({#2})}
\newcommand{\subtermat}[2]{{#1}|_{#2}}
\newcommand{\replsubtermat}[3]{{#1}[{#3}]_{#2}}
\newcommand{\abstractable}{\mathbb{A}}

\newcommand{\nat}{\mathbb{N}}

\newcommand{\subformulas}{{\sf sf}}

\newcommand{\equals}[2]{{{#1}\bowtie{#2}}}
\newcommand{\listmemb}{\twoheadleftarrow}
\newcommand{\disj}{\vee}
\newcommand{\conj}{\wedge}

\newcommand{\says}{\ \mathit{says}\ }
\newcommand{\assertact}{{\sf assert}}
\newcommand{\elg}{{\sf el}}

\newcommand{\publics}{{\sf pubs}}

\newcommand{\rnrule}{{\sf r}}

\newcommand{\rnpair}{{\sf pair}}

\newcommand{\rnfst}{{\sf fst}}
\newcommand{\rnsnd}{{\sf snd}}
\newcommand{\rnenc}{{\sf enc}}
\newcommand{\rndec}{{\sf dec}}

\newcommand{\introd}{\sf i}
\newcommand{\elimin}{\sf e}

\newcommand{\rnax}{{\sf ax}}
\newcommand{\rneq}{{\sf eq}}
\newcommand{\rnsubst}{{\sf subst}}
\newcommand{\rncons}{{\sf cons}}
\newcommand{\rntrans}{{\sf trans}}
\newcommand{\rnsymm}{{\sf sym}}
\newcommand{\rnsays}{{\sf say}}

\newcommand{\rnweak}{{\sf wk}}
\newcommand{\rnproj}{{\sf proj}}

\newcommand{\rncut}{{\sf cut}}
\newcommand{\rnconji}{\wedge\introd}
\newcommand{\rnconje}{\wedge\elimin}

\newcommand{\rnexintro}{\exists\introd}
\newcommand{\rnexelim}{\exists\elimin}
\newcommand{\rnlint}{{\sf int}}
\newcommand{\rnlwk}{{\sf wk}}
\newcommand{\rnprom}{{\sf prom}}

\newcommand{\derives}{\vdash}
\newcommand{\DYderives}{\vdash_{\mathit{dy}}}

\newcommand{\assderives}{\vdash_{\mathit{a}}}

\newcommand{\eqderives}{\vdash_{\mathit{eq}}}
\newcommand{\measure}{\delta}
\newcommand{\conseq}[4]{\mathscr{E}^{{#3},{#4}}_{{#1},{#2}}}

\newcommand{\dc}{\mathit{ker}}
\newcommand{\Substs}{\mathscr{S}}

\newcommand{\honterms}{\mathit{HT}}

\newcommand{\intterms}{\mathit{IT}}

\newcommand{\intsub}{\sigma}
\newcommand{\vintsub}{\sigma^{\!*}}
\newcommand{\hwsub}{\theta}

\newcommand{\iwsub}{\mu}
\newcommand{\viwsub}{\mu^{\!*}}
\newcommand{\bigsub}{\omega}
\newcommand{\vbigsub}{\omega^{\!*}}
\newcommand{\vlambda}{\lambda^{\!*}}
\newcommand{\wh}{\widehat}

\newcommand{\dommu}{Z}
\newcommand{\expplus}[1]{{#1}^{*}}

\newcommand{\zap}[1]{\overline{\mspace{1mu}{#1}\mspace{1mu}}}

\newcommand{\dagsize}[1]{|\st({#1})|}

\newcommand{\fixedname}{{\sf m}}

\newcommand{\constst}{\mathscr{C}}
\newcommand{\stnonvars}{\mathscr{D}}
\newcommand{\hatst}{\wh{\constst}}

\newcommand{\rank}{\mathit{rank}}

\newcommand{\SF}{\subformulas}

\newcommand{\st}{\subterms}

\newcommand{\axiomsof}{{\sf axioms}}
\newcommand{\concof}{{\sf conc}}

\newcommand{\termsof}{{\sf terms}}
\newcommand{\listsof}{{\sf lists}}
\newcommand{\atomsof}{{\sf at}}

\title{Insecurity problem for assertions remains in NP}
\author{R Ramanujam \\ 
The Institute of Mathematical Sciences, Chennai (Retd.) \\
Homi Bhabha National Institute, Mumbai (Retd.) \\
Azim Premji University, Bengaluru (Visiting) \\ 
\url{jam@imsc.res.in}
\and 
Vaishnavi Sundararajan \\
Visiting Post-Doctoral Fellow \\ 
Chennai Mathematical Institute \\
\url{vaishnavi@cmi.ac.in}
\and
S P Suresh \\
Chennai Mathematical Institute \\
\url{spsuresh@cmi.ac.in}}

\begin{document}
\maketitle

\begin{abstract}
    In the symbolic verification of cryptographic protocols, a central problem is deciding whether a protocol admits an execution which leaks a designated secret to the malicious intruder. \cite{RT03} shows that, when considering finitely many sessions and a protocol model where only terms are communicated, this ``insecurity problem'' is NP-complete. Central to their proof strategy is the observation that any execution of a protocol can be simulated by one where the intruder only communicates terms of bounded size. 
    However, when we consider models where, in addition to terms, one can also communicate logical formulas, the analysis of the insecurity problem becomes tricky. In this paper we consider the insecurity problem for protocols with logical statements that include {\em equality on terms} and {\em existential quantification}. Witnesses for existential quantifiers may be of unbounded size, and obtaining small witnesses while maintaining equality proofs complicates the analysis. We use a notion of {\em typed} equality proofs, and extend techniques from~\cite{RT03} to show that this problem is also in NP. We also show that these techniques can be used to analyze the insecurity problem for systems such as the one proposed in~\cite{RSS17}.
\end{abstract}

\section{Introduction}\label{sec:intro}
\subsection{Background}
The symbolic analysis of security protocols is a long-standing field of study, with the Dolev-Yao model~\cite{DY83} being the standard. Here, messages in a protocol are abstracted as elements in a \emph{term algebra}, which usually includes operators for \emph{pairing} and \emph{encryption}, among others. \emph{Proof rules} govern how to derive new terms from existing ones. %how to construct terms containing these operators, and how to disassemble a complex term into its components. 
This model also considers an idealized intruder who is assumed to control the network, but cannot break cryptography. This ``all-powerful'' intruder can see/block/inject/redirect messages, as well as derive new terms from knowledge accumulated over the execution of a protocol. In particular, every communication is split into a send and a receive, with the sends being implicitly captured by the intruder, and receives assumed to come from the intruder. A send need not be accompanied by a corresponding receive, and vice versa.
Building on this abstraction, there have been tremendous advancements in the formal study of security protocols. People have studied extensions of this basic model to express richer classes of protocols and security properties~\cite{ABF17, Bla01, BMU08, CDL06}, and associated decidability and complexity results~\cite{AC04,BP03,BRS10,CS03,CKRT05,LLT07,CDS21,RS05,RS06,DLMS04}.  Various verification tools have also been built based on these formal models \cite{Cr08,Bla01,Bla16,MSCB13}. A tutorial introduction to this area with many more references is~\cite{CK14}.

\subsection{The insecurity problem}
A basic problem of interest is the \emph{insecurity problem} -- is there an execution of a given protocol that leaks a designated secret? Following the abstract model of~\cite{RT03}, protocols can be represented as a set of \emph{roles}, each of which is a sequence of pairs $\recsend{r_{1}}{s_{1}},\ldots,\recsend{r_{n}}{s_{n}}$. The $r_{i}$s are messages received from the network/intruder, and each $s_{i}$ is a message send (by an honest agent) in response to $r_{i}$. The $r_{i}$s and $s_{i}$s can have variables occurring in them. Variables which occur in an $s_{i}$ first are parameters of the protocol that are in the control of an agent enacting the role. The agent instantiates them with concrete values to obtain a \emph{session}. Any variable remaining in a session will, therefore, first occur in one of the $r_{i}$s, and represents an unknown value received from the intruder. An execution (or \emph{run}) of a protocol is a an interleaving of a finite number of sessions, along with a mapping of all the variables occurring in the run to concrete values. (We assume that variables occurring in a session are distinct from those occurring in other sessions.) We further require that all the messages received (resp.\ sent) can be derived from the knowledge accumulated so far in the run by the intruder (resp.\ an honest agent). A protocol is insecure if there is a run where the intruder can derive a designated secret from their knowledge accumulated over the course of the run. 

A run can involve arbitrarily many sessions in parallel, each session may introduce new fresh values (\emph{nonces}), and the intruder can instantiate the variables with arbitrarily large terms. Each of these sources of unboundedness contributes to the intractability of the insecurity problem, rendering it undecidable in general~\cite{ALV03, DLMS04, HT96}. One way to obtain decidability is to impose bounds on some of these parameters. In particular,~\cite{RT03} considers the insecurity problem for \textbf{finitely many sessions}, i.e., is there a run of the protocol consisting of at most $K$ sessions (for a fixed $K$) which leaks a secret? Bounding the number of sessions automatically bounds the number of fresh nonces used in a run, but the intruder can still instantiate variables with arbitrarily large terms, as illustrated in the following example. 

\begin{example}\label{ex:nonatomic}
Consider a protocol given by the two roles:
\begin{align*}
\eta_{1} &\coloneqq \quad \recsend{\ast}{\{(\pk_{a}, \{m\}_{\pk_{b}})\}_{\pk_{b}}} \\
\eta_{2}(x,y) &\coloneqq \quad \recsend{\{(x, \{y\}_{\pk_{b}})\}_{\pk_{b}}}{\{y\}_{x}}
\end{align*}
Here, $\ast$ is a dummy term triggering the start of $\eta_{1}$, $m$ is a fresh name generated by $A$ intended to be kept secret, $\{t\}_{k}$ denotes the encryption of $t$ with the key $k$, $(t,u)$ represents the pairing of terms $t$ and $u$, and $\pk_{a}; \sk_{a}$ and $\pk_{b}; \sk_{b}$ are the public-private key pairs of agents $A$ and $B$. %$B$ has $\sk_{b}$, the key using which it can decrypt messages of the form $\{t\}_{\pk_{b}}$, so it can derive $\{y\}_{x}$ from the received message. 

%Suppose the intruder $I$ stores the term $\{(\pk_{a},\{m\}_{\pk_{b}})\}_{\pk_{b}}$ sent by $A$. It also has $\pk_{a}, \pk_{b}, \pk_{i}$ and $\sk_{i}$, but cannot decrypt the message sent by $A$. However, it can still 

However, the intruder can get $m$ via the run in Figure~\ref{fig:nonatomic}, obtained by interleaving the sessions $\eta_{1}$, $\eta_{2}(x_{1},y_{1})$ and $\eta_{2}(x_{2},y_{2})$ under the  substitution $\intsub$. Here, $A \rightarrow: t$ denotes a send of $t$ by $A$, and $\rightarrow A: t$ denotes a receive of $t$ by $A$. 
\end{example}

\begin{figure}
{	\centering
%\small
\renewcommand{\arraystretch}{1.0}
	\begin{tabular}{l@{\hspace{0.2em}}c@{\hspace{1em}}l}
		$A \rightarrow$ & : & $\{(\pk_{a},\{m\}_{\pk_{b}})\}_{\pk_{b}}$ \\ 
		$\rightarrow B$ & : & $\{(\pk_{i}, \{(pk_{a}, \{m\}_{\pk_{b}})\}_{\pk_{b}})\}_{\pk_{b}}$ \\
		$B \rightarrow$ & : & $\{(pk_{a}, \{m\}_{\pk_{b}})\}_{\pk_{i}}$ \\
		$\rightarrow B$ & : & $\{(pk_{i}, \{m\}_{\pk_{b}})\}_{\pk_{b}}$ \\ 
		$B \rightarrow$ & : & $\{m\}_{\pk_{i}}$. 
	\end{tabular}
	\[\intsub \coloneqq [x_{1} \mapsto \pk_{i}, x_{2} \mapsto \pk_{i}, y_{1} \mapsto (\pk_{a},\{m\}_{\pk_{b}}), y_{2} \mapsto m].\]
%	\small
%      \setlength{\tabcolsep}{0.6em}
%		\begin{tabular}{l@{\hspace{0em}}l|l}
%				& &\\
%				$A \rightarrow$ &$: \{(\pk_{a},\{m\}_{\pk_{b}})\}_{\pk_{b}}$ & $\intsub \coloneqq \bigl[$ \\ 
%				$\rightarrow B$ &$: \{(\pk_{i}, \{(pk_{a}, \{m\}_{\pk_{b}})\}_{\pk_{b}})\}_{\pk_{a}}$ & $x_{1} \mapsto \pk_{i}$, \\
%				$B \rightarrow$ &$: \{(pk_{a}, \{m\}_{\pk_{b}})\}_{\pk_{i}}$ & $x_{2} \mapsto \pk_{i}$, \\
%				$\rightarrow B$ &$: \{(pk_{i}, \{m\}_{\pk_{b}})\}_{\pk_{b}}$ & $y_{1} \mapsto (\pk_{a},\{m\}_{\pk_{b}})$, \\ 
%				$B \rightarrow$ &$: \{m\}_{\pk_{i}}$ & $y_{2} \mapsto m \bigr]$. 
%				\\
%		\end{tabular}
%		\begin{tabular}{ll}
%		&
%		\end{tabular}
}
\caption{Attack for Example~\ref{ex:nonatomic}}
\label{fig:nonatomic}
\end{figure}

There are two points of interest in the above example. The first is that $\intsub(y_{1})$ is a non-atomic term, and therefore the first term received by $B$ has larger size than intended in the protocol. The second is that $y_{1}$ unifies under $\intsub$ with a non-atomic term mentioned in the protocol, namely, $\intsub(y_{1}) = \intsub(\pk_{a}, \{m\}_{\pk_{b}})$. This last fact is important -- this pattern match is needed for the intruder to get $B$ to believe that $\intsub(y_{1})$ is expected according to the protocol specification, and elicit an appropriate response. 

\begin{example}\label{ex:largeterms}
Consider the protocol from Example~\ref{ex:nonatomic}. Let $\intsub'$ map $y_{1}$ to a large term $t$ that does not have the same pattern as any term mentioned in $\eta$, say $((\pk_{i}, \pk_{a}), (\pk_{i}, \pk_{a}))$. This is also a non-atomic term, but this does not ``help'' the intruder in any way. Suppose $B$ receives $\{(\pk_{i}, \{t\}_{\pk_{b}})\}_{\pk_{b}}$ and responds with $\{t\}_{\pk_{i}}$. What has $I$ gained? On decrypting with $\sk_{i}$, it can learn $t$, but that was already constructed by $I$. Further, it cannot trick $B$ into looking inside $t$ and extracting some embedded secret for it (like in the earlier attack $\intsub$), since $t$ does not match any expected pattern. So, the intruder might as well take $\intsub'(y_{1})$ to be an atomic term. 
\end{example}

Example~\ref{ex:largeterms} illustrates the core reasoning in~\cite{RT03}. Essentially, a variable $x$ is good if it unifies with a non-variable pattern (from the protocol specification) under $\intsub$, and bad otherwise. If there is an attack $(\xi, \intsub)$ on the protocol which involves a bad variable, one can ``zap'' it to an atomic value to obtain a ``smaller'' attack $(\xi, \intsub')$, i.e., the sum of the sizes of $\intsub'(x)$ (over all $x$) is smaller than the corresponding sum for $\intsub$. By repeatedly zapping one bad variable after another, one can obtain an attack $\vintsub$ where all variables are good. One can now bound the size of each $\vintsub(x)$ by a polynomial of the sizes of terms mentioned in the protocol specification, and thus prove that insecurity is in NP. The technical challenge is to prove that after each zap, all messages sent by the intruder in $\xi$ are still derivable from $I$'s knowledge. The interested reader can consult~\cite{RT03} for further details. 

\subsection{Our contribution}
In this paper, we extend the techniques of~\cite{RT03} to protocols where agents can communicate logical formulas in addition to terms. These formulas convey properties of terms, and serve as an abstract vehicle to capture certification. To highlight the essential proof ideas, we consider a syntax consisting only of equality between terms (denoted $\equals{t}{u}$ to avoid overloading the $=$ operator) and existential quantification. 

%In this paper, we extend the techniques of~\cite{RT03} to protocols which communicate \emph{assertions} in addition to terms. This model was proposed by~\cite{RSS14,RSS17}. These assertions can be thought of as certificates conveying properties of terms. In order to highlight the essential proof ideas, we consider a restricted assertion syntax, consisting of equality between terms (denoted $\equals{t}{u}$, to avoid conflict with $=$, which denotes identity between terms) and existential quantification. 

When lifting the zapping procedure outlined above to formulas, equalities between terms might be violated by a zap. 

\begin{example}
Consider a protocol which mentions $\{x\}_{k}$ in its specification, and consider a run $(\xi, \intsub)$ where $x$ is bad and $y$ is good, with $\intsub(x) = t$ and $\intsub(y) = \{t\}_{k}$. If we zap $x$ to an atomic value, $y$ is no longer equal to $\{x\}_{k}$, and thus becomes bad. We could zap $y$ in the next round, but the equality between $y$ and $\{x\}_{k}$ would be lost.

Suppose the formula $\equals{y}{\{x\}_{k}}$ is also communicated. Then we cannot zap $x$ without modifying $\{y\}_{k}$ appropriately. 
\end{example}

Thus, formulas introduce an additional concern -- the need to preserve equalities. Zapping and derivability can also interact in more complicated ways, as illustrated below. 

\begin{example}
Consider a run $(\xi,\intsub)$ of some protocol, with a bad variable $x$, where $\intsub(x) = \{n\}_{k}$. Suppose the intruder derives $n$ by decrypting $\intsub(x)$, perhaps as a pre-requisite to some message sent in this run. When we zap $x$, we would like to be able to mimic all the derivations that happen in the original run. But the zapped image of $x$ is now an atomic value, on which decryption can no longer be applied. 
\end{example}

%Situations where a destructor rule is applied on $\intsub(x)$ for a bad variable $x$ need to be handled carefully. 

Proof systems for equalities often include \emph{projection}, which lets one infer $\equals{m}{n}$ and $\equals{m'}{n'}$ from $\equals{(m,m')}{(n,n')}$. However, this means one cannot zap ``indiscriminately''.

\begin{example}
Suppose $x$ is bad and $\intsub(x) = (m,m')$, and we have a proof of $\equals{m}{n}$ by applying the projection rule on $\equals{\intsub(x)}{(n,n')}$. If we zap $x$ to get a new substitution $\intsub'$, we see that we can no longer apply projection on $\equals{\intsub'(x)}{(n,n')}$.
\end{example}

Thus, we cannot zap one variable in an equality independent of the other component. We have to zap all bad variables in $\intsub$ to get a $\vintsub$, without breaking any of the ``relevant'' equalities. 

%\subsection{Our contribution}
When a variable is good, it matches a pattern specified in the protocol. We can view the pattern as a \emph{type} of the term. Note that, as in Example~\ref{ex:nonatomic}, the intruder needs a variable to match a pattern from the protocol only in situations where it wants an honest agent to look into the term and extract some component. These variables need to be good, and therefore ``typed''. We can use this to get the notion of a \emph{typed proof}. In a typed proof, all good variables have a type, and any term that cannot be typed is built from its components by the intruder. 

Typed proofs seem to handle the problems introduced by equality formulas as above. However, quantifiers complicate the situation further. Consider a formula $\exists{x_{1}\ldots{}x_{k}}. \equals{t}{u}$. To derive such a formula, one needs to apply existential quantification $k$ times. Each occurrence of existential quantification is applied to an $\alpha(w_{i})$ to get $\exists x_{i}. \alpha(x_{i})$. Collecting all such $(x_{i}, w_{i})$ pairs that occur in a derivation, we get a substitution $\iwsub$. Proving $\exists{x_{1}\ldots{}x_{k}}.\equals{t}{u}$ thus reduces to proving $\equals{\iwsub(t)}{\iwsub(u)}$. In order to get a bound on the size of these witnesses $w_{i}$, we might have to zap some of the values assigned by $\iwsub$ as well. Furthermore, there are complex interactions between $\intsub$ (which provides meaning to variables occurring in a role) and $\iwsub$. We will see that in order to obtain a correct classification of good and bad variables, we need to consider both substitutions simultaneously. In fact, there are more than two substitutions at play, and the interactions are even more complex. We solve all these issues in Section~\ref{sec:insecurity}, the heart of this paper. We then show how to extend these results to the ``assertions'' used to model and analyze the FOO e-voting protocol in~\cite{RSS17}.

%A final note about the proof system for formulas. Many of these rules propagate equalities, which is well studied in the literature on logic and term rewriting~\cit{}. But some features of our proof system are driven by security considerations. For instance, in any proof system for equality, we would expect to derive $\equals{m}{n}$ from $\equals{\func(m,m')}{\func(n,n')}$, but this is not always a valid rule for security protocols. For example, if $\func$ is encryption, then we cannot always allow one to derive $\equals{m}{x}$ from $\equals{\{m\}_{k}}{\{x\}_{y}}$. If the reasoner has no access to the inverse of the key $k$, then $\{m\}_{k}$ is just a bitstring (importantly, one that does not reveal $m$), but if someone manages to derive $\equals{m}{x}$, then it is reasonable to assume that they have obtained access to $m$. So we should not allow the above inference, unless one can derive $m$. This shows that even to reason only with equalities, we need to keep some terms on hand to check for the right security ``contexts''. We discuss these subtleties when we introduce the proof rules in Section~\ref{sec:dolevyao}. 

\subsection{Related Work: Equational Theories}
A lot of work has been done in the field of symbolic verification for security involving equalities of a different nature, namely those arising from \emph{equational theories}~\cite{CKRT05, CS03, CDL06}. In these theories, one considers terms built using two kinds of operators  -- \emph{constructors} building up a term $\func(t_{1}, \ldots, t_{r})$ from its constituents $t_{1}, \ldots, t_{r}$, and (zero or more) \emph{destructors} corresponding to each constructor, that serve to break down a complex term and obtain some part of it. These destructors are governed by equations of the form $\gunc(\func(t_{1},\dots,t_{r}),u_{1},\dots,u_{s}) = t_{i}$, for some $i \leq r$. For instance, we might have the constructor $\enc$ and the corresponding destructor $\dec$, with the equation $\dec(\enc(x,y), \inv{y}) = x$. In fact, one usually converts these equations to \emph{rewrite rules}, obtaining a rewrite system, and studies the convergence of such a system. For example, one can convert the earlier equation to the rewrite rule $\dec(\enc(x,y), \inv{y}) \rightarrow x$. Such rules are implicitly universally quantified. Applying the above rule on a term $t$ boils down to choosing any subterm of $t$ matching $\dec(\enc(x,y), \inv{y})$ (under some substitution $\intsub$) and replacing it by $\intsub(x)$. 

In fact, these rewrite rules correspond to inference rules. The above rewrite rule says, for instance, that if one can derive $\enc(x,y)$ and $\inv{y}$, one can also derive $x$ (and $\dec$ is the name of the inference rule). In contrast to the equations and rewrites as seen above (which \emph{operate} on terms), the equalities that we consider in this paper have the same status as terms -- we start with an initial stock of terms and equalities, and see which other equalities can be derived using inference rules. One might, for instance, encode the standard rule for transitivity of equalities into the following rewrite rule: $\rntrans(\equals{x}{y}, \equals{y}{z}) \rightarrow \equals{x}{z}$. Similarly, one might encode the standard substitution rule for equalities as the rewrite rule: $\rnsubst(\equals{r(t)}{s(t)}, \equals{t}{u}) \rightarrow \equals{r(u)}{s(u)}$. Here, we replace every occurrence of $t$ in $\equals{r(t)}{s(t)}$ (the same $t$, and not a substitution instance like with $\dec$ and $\enc$ earlier) by $u$.

Thus, there is a clear distinction between the equalities we consider and those in an equational theory. We view formulas as first-class objects enjoying the same status as terms, which can be governed by equations in an appropriate equational theory. We choose to work with inference systems (and terms without destructors) here, but one can translate the system here into the language of equational theories and rewrite rules.
  
\subsection{Organization of the paper}
In Section~\ref{sec:dolevyao}, we introduce the syntax and derivation systems for terms and assertions. We deal with the \emph{derivability problem} for assertions in Section~\ref{sec:derivability}. %We reduce the general problem to one of deriving a substitution instance of an equality from a set of terms and equalities. The important challenge here is to obtain a bound on the size of the terms assigned by the ``witness substitution''. This proof itself is provided in Appendix~\ref{app:derivability} (since the active intruder problem, solved later, subsumes some of the ideas in this proof). 
In Section~\ref{sec:protruns}, we present a protocol model to introduce the insecurity problem. In Section~\ref{sec:insecurity}, we show how to properly zap all the bad variables associated with a run. We thus prove that insecurity for a finite number of sessions remains in NP, even after adding equalities and existential quantification. In Section~\ref{sec:fullsyntax}, we generalise the definitions to a richer syntax in which we can model some widely-used protocols and specify useful security properties, as presented in~\cite{RSS17}. We discuss some possible avenues for future work in Section~\ref{sec:disc}. 

\section{The Dolev-Yao model: terms and formulas}\label{sec:dolevyao}
\subsection{Terms: Syntax and Derivation System}
Each communicated message is modelled as a term in an algebra. New terms can be derived from old ones using a proof system which specifies how these operators behave. We have a set $\names$ of names (atomic terms, with no further structure), and a set $\vars$ of variables. 
We assume that $\qvar \subset \vars$ is the set of variables used for quantification. We denote by $\ivar$ the set $\vars\setminus\qvar$. We assume a set of \emph{keys} $\keys \subseteq \names$, and that every  $k \in \keys$ has an inverse $\inv{k} \in \keys$. $\ag \subseteq \names$ is the set of agents, with $I \in \ag$ being the intruder. We consider the following simple syntax for terms\footnote{Our main results are not dependent on the exact choice of term algebra, so we can add other operators and proof rules, as long as the system has normalization and subterm property, and derivability is efficiently decidable.}\!\!, with operators for pairing and encryption.
\[
        t \in \Tterms := x \mid m \mid (t, u) \mid \enc(t,k)
\]
where $x \in \vars$, $m \in \names$, $k \in \keys \cup \vars$ and $t, u \in \Tterms$. We use the shorthand $\{t\}_{k}$ for $\enc(t,k)$. \emph{Ground} terms are those without variables. The set of subterms of $t$ is denoted by $\st(t)$, and defined as usual. $\varsof(t)$ denotes the set of variables in $t$.
 
To derive new terms from existing ones, we use the proof system given in Table~\ref{tab:termalgtab}. $\rnenc$ and $\rnpair$ are constructor rules, while the others are treated as destructors.
%\footnote{Strictly speaking, $\rnax$ is not a destructor rule, but is classified as one since much of the treatment for $\rnax$ will be similar to the other destructors.}\!\!. 
We say $X \DYderives t$ if there is a proof of $X \vdash t$ using these rules, and use $X \DYderives S$ to denote that $X \DYderives t$ for every $t \in S$. For destructors, the leftmost premise is the \emph{major premise}.

\begin{table}
    \centering {\small 
    \tabulinesep=0.8mm
    \setlength{\tabcolsep}{0.4em}
    \begin{tabu}{|c|c|c|c|c|c|}
        \hline
        \multicolumn{2}{|c|}{
            \begin{prooftree}
                \justifies X \vdash t \using \rnax (t \in X)
            \end{prooftree}
        }
        &
        \multicolumn{2}{c|}{
            \begin{prooftree}
                X \vdash (t, u)
                \justifies X \vdash t \using \rnfst
            \end{prooftree}
        }
        &
        \multicolumn{2}{c|}{
            \begin{prooftree}
                X \vdash (t, u)
                \justifies X \vdash u \using \rnsnd
            \end{prooftree}
        }
        \\
        \hline
        \multicolumn{2}{|c|}{
            \begin{prooftree}
                X \vdash t \quad X \vdash u
                \justifies X \vdash (t, u) \using \rnpair
            \end{prooftree}
        }
        &
        \multicolumn{2}{c|}{
            \begin{prooftree}
                X \vdash \{t\}_{k} \quad X \vdash \inv{k}
                \justifies X \vdash t \using \rndec
            \end{prooftree}
        }
        &
        \multicolumn{2}{c|}{
            \begin{prooftree}
                X \vdash t \quad X \vdash k
                \justifies X \vdash \{t\}_{k} \using \rnenc
            \end{prooftree}
        }
        \\
        \hline
%        \multicolumn{3}{|c|}{
%            \begin{prooftree}
%                X \vdash \{\!|t|\!\}_{\pk(k)} \quad X \vdash k
%                \justifies X \vdash t \using \rnadec
%            \end{prooftree}
%        }
%        &
%        \multicolumn{3}{c|}{
%            \begin{prooftree}
%                X \vdash t \quad X \vdash \pk(k)
%                \justifies X \vdash \{\!|t|\!\}_{\pk(k)} \using \rnaenc
%            \end{prooftree}
%        }
%        \\
%        \hline
    \end{tabu}
    }
    \caption{Proof system for terms}
    \label{tab:termalgtab}
\end{table}

For any proof $\pi$ of $X \vdash t$, we often refer to $X$ as $\axiomsof(\pi)$, and we denote by $\concof(\pi)$ the term $t$, and by $\termsof(\pi)$ all terms occurring in $\pi$. $\pi$ is said to be \emph{normal} if no constructor rule is the major premise of a destructor rule. The following properties are standard (see~\cite{RT03}, for example). 
\begin{itemize}
    \item \emph{Normalization}: Every proof $\pi$ of $X \vdash t$ can be converted into a normal proof $\varpi$ of the same.
    \item \emph{Subterm property}: For any normal proof $\varpi$ of $X \vdash t$, $\termsof(\varpi) \subseteq \st(X \cup \{t\})$, and if $\varpi$ ends in a destructor rule, $\termsof(\varpi) \subseteq \st(X)$.
    \item \emph{Derivability problem}: There is a PTIME algorithm to check whether $X \DYderives t$, given $X$ and $t$.
\end{itemize}

%We also assume a set $\mathcal{F}$ of operators. The set of terms, denoted by $\Tterms$, is given by
%\[
%    t \in \Tterms ::= x \mid m \mid \func (t_{1}, \ldots, t_{n})
%\]
%where $x \in \vars$,  $m \in \names$, $t_{1}, \ldots t_{n} \in \Tterms$, and $\func \in \mathcal{F}$ is an $n$-ary operator. 
%Each $\func \in \mathcal{F}$ has constructor rules and destructor rules, expressed in terms of sequents of the form $X \vdash t$, where $X \cup \{t\}$ is a set of terms. 

\subsection{Formulas}
We consider a syntax which includes equality over terms (denoted by $\bowtie$)
%(we denote the equality of $t$ and $u$ by $\equals{t}{u}$ to avoid overloading $=$) 
and existential quantification.
%\footnote{Just as for terms, we choose a simple syntax which highlights the complexity involved in solving the insecurity problem. In a later section, we will present an extended syntax which can be used to model actual protocols.}\!\!. 
An \emph{assertion} is a formula of the form $\exists{x_{1}\ldots{}x_{k}}.(\equals{t}{u})$ where $k \geq 0$, $\{x_{1}, \ldots, x_{k}\} \subseteq \qvar$, and $t, u \in \Tterms$. If $k = 0$, the assertion is just $\equals{t}{u}$.

%
%In the following, $t,u \in \Tterms$, $P$ is an $m$-ary predicate, $u_{1}, \ldots, u_{m}, t_{0} \in \names \cup \vars$, and $t_{1}, \ldots, t_{n} \in \names$,\footnote{We could consider arbitrary terms in list membership, but this simple syntax suffices for most examples. Similarly for $P(u_{1}, \ldots. u_{m})$.} $x \in \qvar$, and $\pk(k)$ is the public key corresponding to the secret key $k$.
%\begin{align*}
%    \alpha &:= \equals{t}{u} \mid P(u_{1}, \ldots, u_{m}) \mid t_{0} \listmemb [t_{1}, \ldots, t_{n}] \\
%    & \hspace{5mm} \mid \alpha_{0} \conj \alpha_{1} \mid \exists x.~\alpha(x) \mid \pk(k) \says \alpha
%\end{align*}
%We will often use $A\says$as shorthand for $\pk(k_{A})\says$in examples, where $k_{A}$ is the secret key of $A$. 
We denote the free (resp. bound) variables occurring in an assertion $\alpha$ by $\freevars(\alpha)$
%\footnote{$\freevars(\alpha) = (\varsof(t) \cup \varsof(u)) \setminus \{x_{1}, \ldots, x_{k}\}.$} 
and $\boundvars(\alpha)$. $\varsof(\alpha) = \freevars(\alpha) \cup \boundvars(\alpha)$. The set of subterms of terms occurring in $\alpha$ is given by $\st(\alpha)$. The set of subformulas of $\alpha$ is given by $\SF(\alpha)$. We can lift these notions to sets of assertions as usual. For $\alpha = \exists{x_{1}\ldots{}x_{k}}.(\equals{t}{u})$ and a substitution $\lambda$, $\lambda(\alpha) \coloneqq \exists{x_{1}\ldots{}x_{k}}.(\equals{\lambda'(t)}{\lambda'(u)})$ where $\lambda' = \lambda\restriction\freevars(\alpha)$.
%\footnote{$\lambda'(x) \coloneqq \lambda(x)$ for $x \in \freevars(\alpha)$, and $\lambda'(x)$ is undefined for $x \notin \freevars(\alpha)$.}\!\!.

We define the \emph{public terms} of an assertion $\alpha$, denoted $\publics(\alpha)$, to be the set of maximal subterms of $\alpha$ that do not contain any quantified variables. Formally, $t \in \publics(\alpha)$ iff $t \in \st(\alpha)$, $\varsof(t) \cap \qvar = \emptyset$ and there is no $u \in \st(\alpha)$ s.t. $t \neq u, t \in \st(u)$ and $\varsof(u) \cap \qvar = \emptyset$.

\begin{example}
Suppose an assertion $\alpha = \exists{x}.\equals{t}{\{m\}_{x}}$ is communicated. $\alpha$ states that $t$ is an encryption of $m$ with some key. Thus $t$ and $m$ are revealed by $\alpha$, and $t,m \in \publics(\alpha)$.
\end{example}

\subsection{Proof System for Assertions}
Before we present the proof system, we need to fix the conditions under which one is allowed to derive a new assertion from existing ones. In a security context, it becomes important to distinguish when a term is accessible inside an assertion versus when it is not. To substitute a term $u$ (with, say, $x$) inside a term $t$, an agent $A$ must essentially break the term down to that position, replace $u$ with $x$, and construct the whole term back. This depends on other terms $A$ has access to. We formalize this notion as ``abstractability''.

%We will view terms as trees, with $\positions{t} \subseteq \nat^{*}$ denoting the set of positions of the term $t$, with $\epsilon$ the empty word in $\nat^{*}$.
%
\begin{definition}
	The set of \defemph{positions of a term} $t$, denoted $\positions{t}$, is a subset of $\nat^{*}$ defined as follows (where $\func$ is $(\cdot , \cdot)$ or $\enc$):  
	\begin{itemize}
	\item $\positions{m} \coloneqq \{\epsilon\}$, for $m \in \names \cup \vars$
	\item $\positions{\func(t,u)} \coloneqq \{\epsilon\} \cup \{0p \mid p \in \positions{t}\} \cup \{1p \mid p \in \positions{u}\}$
	\end{itemize}
	The \defemph{term positions of an assertion} $\alpha = \exists{}x_{1}\ldots{}x_{k}.(\equals{t}{u})$ are defined to be $\positions{\alpha} \coloneqq \{0^{k}0p \mid p \in \positions{t}\} \cup \{0^{k}1p \mid p \in \positions{u}\}$.
%	\begin{itemize}
%        \item $\positions{\equals{t}{t'}} = \{0\cdot p \mid p \in \positions{t}\} \cup \{1\cdot p \mid p \in \positions{t'}\}$
%        \item $\positions{P(u_{1}, \ldots, u_{m})} = \{1, \ldots, m\}$
%        \item $\positions{t \listmemb [t_{1}, \ldots, t_{n}]} = \{0, 1, \ldots, n\}$
%        \item $\positions{\alpha \conj \beta} = \{0\cdot p \mid p \in \positions{\alpha}\} \cup \{1\cdot p \mid p \in \positions{\beta}\}$
%	    \item $\positions{\exists{x}.\alpha} = \{0\cdot p \mid p \in \positions{\alpha}\}$				
%		\item $\positions{\pk(k) \says \alpha} = \{0, 00\} \cup \{1\cdot p \mid p \in \positions{\alpha}\}$
%	\end{itemize}    
\end{definition}

For $t, r \in \Tterms$, and $p \in \positions{t}$, $\subtermat{t}{p}$ is the subterm of $t$ rooted at $p$. The set of positions where $r$ occurs in $t$ is given by $\posof{r}{t} \coloneqq \{p \in \positions{t} \mid \subtermat{t}{p} = r\}$. For $P \subseteq \positions{t}$, $\replsubtermat{t}{P}{r}$ is the term obtained by replacing the subterm of $t$ occurring at each $p \in P$ with $r$. We will use analogous notation for assertions. 

\begin{definition}%[Abstractable positions of a term]
%Let $S \DYderives t$. The set of \emph{abstractable positions of $t$ w.r.t.\ $S$} is denoted by $\abstractable(S, t)$. $p \in       \positions{t}$ is in $\abstractable(S, t)$ iff $S \DYderives \subtermat{t}{q\cdot{i}}$ for all $q\cdot{i} \in \positions{t}$ such that  $q$ is a proper prefix of $p$.
Let $t \in \Tterms$. The set of \defemph{abstractable positions of $t$ w.r.t.\ $S$}, denoted $\abstractable(S, t)$, is defined as follows. For $p \in \positions{t}$, let $\mathbb{Q}_{p} = \{\varepsilon\} \cup \{qi \in \positions{t} \mid q$ is a proper prefix of $p\}$. Then $\abstractable(S,t) \coloneqq \{p \in \positions{t} \mid S \DYderives \subtermat{t}{q}$ for all $q \in \mathbb{Q}_{p}\}$.
\end{definition}

\begin{example}
Let $t = (\{\{m\}_{k}\}_{k'}, (n,n'))$. 
\[ \positions{t} = \{\epsilon, 0, 1, 00, 01, 10, 11, 000, 001\} \]
For $S = \{\{m\}_{k}, k', n,n'\}$, $\abstractable(S,t) = \{\epsilon,0,1,00,01,10,11\}$. For $T = \{m,k,k',(n,n')\}$, $\abstractable(T,t) = \positions{t}$.
% iff $S \DYderives \subtermat{t}{p}$ for all $p \in \{\epsilon, 0, 1, 00, 01, 000, 001\}$.
\end{example}

\begin{definition}[Abstractable positions of an assertion]\label{def:assabs}
	The set of \defemph{abstractable positions of $\alpha = \exists{}x_{1}\ldots{}x_{k}.(\equals{t}{u})$ w.r.t.\ $S$}, is defined to be the following, where $T = S \cup \{x_{1}, \ldots, x_{k}\}$:
	\[\abstractable(S,\alpha) \coloneqq \{0^{k}0p \mid p \in \abstractable(T, t)\} \cup \{0^{k}1p \mid p \in \abstractable(T, u)\}\]
%	
%	\begin{itemize}[leftmargin=*,labelsep=1pt]
%		\item $\abstractable(S, \equals{t_{0}}{t_{1}}) = \{i\cdot p \mid i \in \{0, 1\}, \ p \in \abstractable(S,t_{i})\}$
%		\item $\abstractable(S, P(u_{1}, \ldots, u_{m})) = \{i \mid 1 \leq i \leq m, S \DYderives u_{i}\}$
%		\item $\abstractable(S, t \listmemb [t_{1}, \ldots, t_{n}]) = \{0\}$
%		\item $\abstractable(S, \alpha_{0}\conj\alpha_{1}) = \{i\cdot p \mid i \in \{0, 1\}, \ p \in \abstractable(S,\alpha_{i})\} $
%		\item $\abstractable(S, \exists{x}.\alpha) = \{0\cdot p \mid p \in \abstractable(S \cup \{x\},\alpha) \}$
%		\item $\abstractable(S, \pk(k) \says \alpha) =  \{0\} \cup \{1\cdot p \mid p \in \abstractable(S, \alpha)\}$
%	\end{itemize}
\end{definition}

\begin{example}\label{ex:abs}
    Let $\alpha = \exists b.~\{\equals{\{m\}_{b}}{\{m\}_{k}}\}$. Suppose we want to get $\exists ab.~\{\equals{\{a\}_{b}}{\{m\}_{k}}\}$ from $\alpha$ w.r.t.~the set $S = \{m\}$.
    That position of $m$ in $\alpha$ (i.e. $000$) must be abstractable. This requires the sibling position $001$ to also be abstractable, i.e. $S \DYderives b$. However, $S$ does not contain the quantified variable $b$. Therefore, Definition~\ref{def:assabs} is set up to consider derivability from $S \cup \{b\}$, not $S$. Thus, $\abstractable(S, \alpha) = \{00, 000, 001\}$. The tree for $\alpha$ is shown in Figure~\ref{fig:exabs}, with the abstractable positions enclosed in boxes. 
\end{example}

\begin{figure}
\begin{center}
\begin{tikzpicture}[level distance = 8mm]
  \tikzstyle{every node}=[draw=white,text=black]
  \tikzstyle{level 1}=[sibling distance=25mm]
  \tikzstyle{level 2}=[sibling distance=25mm]
  \tikzstyle{level 3}=[sibling distance=15mm]
  \tikzstyle{level 4}=[sibling distance=10mm]
  \node(e) {$\exists b$}
    child 
    { 
         node(0) { $\bowtie$ }
	   child
	   {
       		node(00) {$\enc$}
            	child
            	{
            		node(000) {$m$}
            		node[below = 0.1mm of 000, draw=black]{000}
            	}
            	child
            	{
            		node(001) {$b$}
            		node[below = 0.1mm of 001, draw=black]{001}
            	}
          	     node[left = 0.1mm of 00, draw=black]{00}
      	 } 
             child
             { 
           	node(01) {$\enc$}
    		child
    		{
    			node(010) {$m$}
    			node[below = 0.1mm of 010]{010}
    		}
    		child
    		{
    			node(011) {$k$}
    			node[below = 0.1mm of 011]{011}
    		}
              node[right = 0.1mm of 01]{01}
               }
       %node[right = 0.1mm of 0] {0}
    }
    %node[right = 0.1mm of e]{$\epsilon$}	
;  
\end{tikzpicture}
\caption{Assertion tree for Example~\ref{ex:abs}}
\label{fig:exabs}
\end{center}
\end{figure}
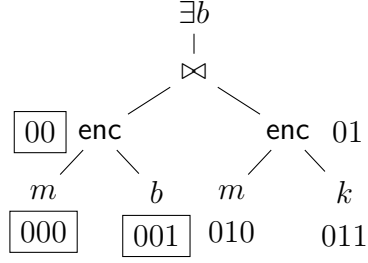

It can be seen that if we substitute variables occurring in some abstractable positions with derivable terms, the abstractability of other positions does not change. This is stated below and proved in Appendix~\ref{app:derivability}. 
%A similar result holds for assertions.
\begin{restatable}{lemma}{absderiv}\label{lem:absderiv}
Let $S \cup \{t,r\} \subseteq \Tterms$ s.t.\ $S \DYderives r$. If $x \notin \varsof(S)$ and $P = \posof{x}{t} \subseteq \abstractable(S \cup \{x\}, t)$, then $\abstractable(S, \replsubtermat{t}{P}{r}) \cap \positions{t} = \abstractable(S \cup \{x\}, t)$. 
%If $Q = \posof{x}{\alpha} \subseteq \abstractable(S \cup \{x\}, \alpha)$, then $\abstractable(S, \replsubtermat{\alpha}{Q}{r}) \cap \positions{\alpha} = \abstractable(S \cup \{x\}, \alpha)$.
\end{restatable}

\begin{table*}[ht]
	\centering {\small
        \tabulinesep=2mm
        \setlength{\tabcolsep}{0.22em}
        \begin{tabu}{|c|c|c|c|c|c|}
            \hline
            \multicolumn{2}{|c|}{
			\begin{math}
                {
                    \begin{prooftree}
                        \justifies S; A \cup \{\alpha\} \vdash \alpha \using \rnax
					\end{prooftree}
			    }
			\end{math}
            }
			&			
            \multicolumn{2}{|c|}{
			\begin{math}
                {
					\begin{prooftree}
						S \DYderives t
						\justifies S; A \vdash \equals{t}{t} \using \rneq
					\end{prooftree}
                }
			\end{math}
            }
			&
            \multicolumn{2}{|c|}{
            \begin{math}
                {
                    \begin{prooftree}
                        S;A \vdash \equals{t}{u}
                        \justifies S;A \vdash \equals{u}{t} \using \rnsymm
                    \end{prooftree}
                }
            \end{math}
            }
			\\
			\hline
            \multicolumn{3}{|c|}{
            \begin{math}
                {
                    \begin{prooftree}
                        S; A \vdash \equals{t_{0}}{u_{0}} \quad S; A \vdash \equals{t_{1}}{u_{1}}
                        \justifies S; A \vdash \equals{\func(t_{0},t_{1})}{\func(u_{0},u_{1})} \using \rncons
					\end{prooftree}
                }
            \end{math}
            }
			&
            \multicolumn{3}{|c|}{
			\begin{math}
                {
                    \begin{prooftree}
                        S; A \vdash \equals{\func(t_{0},t_{1})}{\func(u_{0}, u_{1})}
                        \justifies S; A \vdash \equals{t_{i}}{u_{i}} \using \rnproj_{i}^{\P}
					\end{prooftree}
                }
            \end{math}
            }
			\\
			\hline
            \multicolumn{6}{|c|}{
			\begin{math}
                {
                    \begin{prooftree}
                        S; A \vdash \equals{t_{1}}{t_{2}} \quad \cdots \quad S;A \vdash \equals{t_{k}}{t_{k+1}}
                        \justifies S; A \vdash \equals{t_{1}}{t_{k+1}} \using \rntrans
					\end{prooftree}
                }
			\end{math}
            }
            \\
            \hline 
            \multicolumn{3}{|c|}{
			\begin{math}
                {
                    \begin{prooftree}
                        S; A \vdash \replsubtermat{\alpha}{P}{t} \quad S \DYderives t
                        \justifies S; A \vdash \exists{x}.\alpha \using \rnexintro^{\ddag}
                    \end{prooftree}
                }
            \end{math}
            }
            &
            \multicolumn{3}{|c|}{
                \begin{math}
                    {
                        \begin{prooftree}
                            S; A \vdash \exists{}{x}.\alpha \quad S \cup \{x\}; A \cup \{\alpha\} \vdash \gamma
                            \justifies S; A \vdash \gamma \using \rnexelim^{\S}
                        \end{prooftree}
                    }
                \end{math}
            }
			\\
			\hline
		\end{tabu}
	}
	\caption{
        Derivation system $\assderives$ for assertions. $\func$ is pairing or encryption.
        $\P$ states that $S \DYderives \{t_{0}, t_{1}, u_{0}, u_{1}\}$.
%        $\P$ states that $\{00,01,10,11\} \subseteq \abstractable(S, \equals{\func(t_{0},t_{1})}{\func(u_{0},u_{1}))}$.
%        $\dag$ demands that $P \subseteq \posof{x}{\alpha} \cap \abstractable(S \cup \{x\}, \alpha)$, and no position in $P$ occurs in the scope of a$\says$\!\!. 
        $\ddag$ stands for $P = \posof{x}{\alpha} \subseteq \abstractable(S \cup \{x\}, \alpha)$. $\S$ states that $x \notin \freevars(S) \cup \freevars(A) \cup \freevars(\gamma)$. 
    }
	\label{tab:asstheory}
\end{table*}

The assertion proof system is shown in Table~\ref{tab:asstheory}. We say that $(S; A) \assderives \alpha$ if there is a proof of $(S;A) \vdash \alpha$ using these rules. $(S;A) \assderives \Gamma$ means that $(S;A) \assderives \gamma$ for every $\gamma \in \Gamma$. We say that $(S;A) \eqderives \alpha$ if there is a proof of $(S;A) \vdash \alpha$ that does not use the $\rnexintro$ and $\rnexelim$ rules. We use the notations $\axiomsof(\pi)$, $\termsof(\pi)$ and $\concof(\pi)$, just as with $\DYderives$ proofs.

\subsection{Remarks about the proof rules}
Since we reason about formulas in a security context, a few subtleties arise in our proof system. We discuss these below. 
\begin{itemize}[leftmargin=*]
\item ``Standard'' reasoning with equality would allow us to derive $\equals{m}{x}$ from $\equals{\func(m,k)}{\func(x,y)}$. But suppose $\func$ is encryption. Security considerations dictate that one should not be able to derive $\equals{m}{x}$ from $\equals{\{m\}_{k}}{\{x\}_{y}}$, unless one can derive $m$. So, in general, one needs to reason with equalities in the presence of a set of terms $S$. For the $\rnproj$ rule in particular, we require that all immediate subterms are derivable from this set $S$. This, however, does not rule out reasonable derivations, as illustrated in the next example.

\begin{example}\label{ex:eqderiv}
Let $A = \bigl\{\equals{x}{\{m\}_{y}}, \equals{x}{\{z\}_{k}}\bigr\}$ and $S = \{m,k,y,z\}$. One can derive $\equals{x}{\{m\}_{k}}$ from $S$ and $A$, as in Figure~\ref{fig:exeqderiv}. (For readability, we only present the RHS of every sequent in the proof.) 
%Note that $\rnproj$ can be applied in this proof because $S$ can derive all the relevant subterms. Note that the above proof along with applications of $\rnexelim$ at the end can derive $\equals{x}{\{m\}_{k}}$ from $S = \{m,k\}$ and $A = \{\exists{y}.\equals{x}{\{m\}_{y}}, \exists{z}.\equals{x}{\{z\}_{k}}\}$.
\end{example}

\begin{figure}
\begin{center}
{\footnotesize
\begin{math}
\begin{prooftree}
\[
\justifies \equals{x}{\{m\}_{y}} \using \rnax
\]
\[
\[
S \DYderives m
\justifies \equals{m}{m} \using \rneq
\]
\[
\[
\[
\[
\justifies \equals{x}{\{m\}_{y}} \using \rnax
\]
\justifies \equals{\{m\}_{y}}{x} \using \rnsymm
\]
\[
\justifies \equals{x}{\{z\}_{k}} \using \rnax
\]
\justifies \equals{\{m\}_{y}}{\{z\}_{k}} \using \rntrans
\]
\justifies \equals{y}{k} \using \rnproj
\]
\justifies \equals{\{m\}_{y}}{\{m\}_{k}} \using \rncons
\]
\justifies \equals{x}{\{m\}_{k}} \using \rntrans
\end{prooftree}
\end{math}
}
\end{center}
\caption{Proof tree for Example~\ref{ex:eqderiv}}
\label{fig:exeqderiv}
\end{figure}

\item As mentioned earlier, abstractability plays a crucial role in the application of the $\rnexintro$ rule. Note that the $\rnexintro$ rule in Table~\ref{tab:asstheory} requires a \emph{constructible witness}. We now explain this modification from the ``standard'' $\rnexintro$ rule.

Suppose $P$ is a set of positions in an assertion $\beta$ where a term $t$ occurs. Abstracting the occurrences of $t$ in these positions can be thought of as obtaining $\exists{x}.\replsubtermat{\beta}{P}{x}$ from $\beta$. The above is possible in the presence of a set of terms $S$ only when $P \subseteq \abstractable(S,\beta)$. Alternatively, suppose $\alpha = \replsubtermat{\beta}{P}{x}$. Then $\beta = \replsubtermat{\alpha}{P}{t}$. Lemma~\ref{lem:absderiv} guarantees that $P \subseteq \abstractable(S,\beta)$ iff $P \subseteq \abstractable(S \cup \{x\}, \alpha)$, under the assumption $S \DYderives t$, so we take that as a pre-requisite for $\rnexintro$.

\item A \emph{substitution} rule is standard for systems handling equality.  For $\assderives$, we consider the following $\rnsubst$ rule, with the abstractability side conditions $P \subseteq \posof{x}{\equals{t}{u}} \cap \abstractable(S \cup \{x\}, \equals{t}{u})$ and $S \DYderives \{r,s\}$). 
%\begin{center}
%\begin{math}
\[
{
\begin{prooftree}
(S;A) \vdash \replsubtermat{(\equals{t}{u})}{P}{r} \quad (S;A) \vdash \equals{r}{s}
\justifies (S;A) \vdash \replsubtermat{(\equals{t}{u})}{P}{s} \using \rnsubst
\end{prooftree}
}
\]
%\end{math}
%\end{center}
We can simulate this rule using $\rneq$, $\rnsymm$, $\rntrans$ and $\rncons$, as proved in Appendix~\ref{app:derivability}. Hence, we omit it from Table~\ref{tab:asstheory}.
%\begin{restatable}{lemma}{substadm}\label{lem:subst-adm}
%The $\rnsubst$ rule is admissible in $\assderives$.
%\end{restatable}
\end{itemize}

\section{Derivability problem for \texorpdfstring{$\assderives$}{assertions}}\label{sec:derivability}
The derivability problem requires one to check if $\alpha$ is derivable from $(S; A)$. To check whether an assertion of the form $\exists{x}.~\beta$ is derivable, one would in general have to check if $\beta(t)$ is derivable for some witness $t$. Here we show that if there is a witness at all, there is one of small size. 

One way to represent this is via a substitution $\mu$ which maps each quantified variable $x$ to the corresponding $t$. In any proof, one can remove quantifiers from the LHS (instantiated by ``eigenvariables'') to move to an LHS consisting solely of equality formulas. Any normal proof from such an LHS will not involve $\rnexelim$. One can also simplify the RHS -- instead of proving an assertion $\exists \vec{x}. \equals{t}{u}$, it suffices to prove $\equals{t}{u}$ in the $\eqderives$ system (the system without the rules for quantification) with variables instantiated appropriately. 
%We now present the details.

Hereon, in all assertions we consider, no variable appears both free and bound. Further, whenever we use $(S; A)$, we mean that $S$ is a set of terms, $A$ is a set of assertions, and $(S; A)$ is \emph{sanitized}, as defined below. We use the notation $\varsof(S;A)$ to mean $\varsof(S) \cup \varsof(A)$ and $\freevars(S;A)$ to mean $\varsof(S) \cup \freevars(A)$. 
%Unless otherwise specified, some of the more elaborate proofs for this section are presented in Appendix~\ref{app:derivability}. 

\begin{definition}
$(S; A)$ is \defemph{sanitized} if $\freevars(S;A) \cap \qvar = \emptyset$ and $\publics(\alpha) \in S$ for all $\alpha \in A$.
\end{definition}

\begin{definition}[Derivability problem for $\assderives$] 
    Given a set $S$ of terms, and a set $A \cup \{\alpha\}$ of assertions such that $(S; A)$ is sanitized, check whether $S; A \assderives \alpha$.
\end{definition}

\begin{restatable}{lemma}{leftex}\label{lem:leftex}
    Let $S, A, \exists x.\alpha$ and $\gamma$ be such that $x \notin \varsof(S) \cup \varsof(A \cup \{\gamma\})$ and $\posof{x}{\alpha}\subseteq \abstractable(S \cup \{x\}, \alpha)$. Then $(S; A \cup \{\exists x.\alpha\}) \assderives \gamma$ iff $(S \cup \{x\}; A \cup \{\alpha\}) \assderives \gamma$.
\end{restatable}

With the help of this lemma (proved in Appendix~\ref{app:derivability}), we can transform any proof to one where the LHS consists solely of atomic formulas. This leads us to a notion of \emph{kernel}.

\begin{definition}
    The \defemph{kernel} of $(S;A)$, denoted $\dc(S; A)$, is defined to be $(S \cup \boundvars(A); \{\equals{t}{u} \mid \exists{\vec{x}}.\equals{t}{u} \in A\})$. 
\end{definition}

Each $x \in \boundvars(A)$ that is added to the kernel can be thought of as an eigenvariable which stands for the witness for an existential assertion in $A$. The following statement is an easy consequence of the way $\dc(S;A)$ is defined.
\begin{restatable}{lemma}{breakdownass}\label{lem:breakdownass}
	For any $\gamma$, $(S; A) \assderives \gamma$ iff $\dc(S; A) \assderives \gamma$ without using the $\rnexelim$ rule.
\end{restatable}

We will refer to $\dc(S;A)$ for any sanitized $(S;A)$ as \emph{pure}. The following property of pure $(T;E)$ is proved in Appendix~\ref{app:derivability}.

\begin{restatable}{obs}{dcpure}\label{obs:dc-pure}
Let $(T; E)$ be pure. If $(T;E) \assderives \alpha$ and $a \in \publics(\alpha)$, then $T \DYderives a$. If $(T;E) \eqderives \equals{t}{u}$ then $T \DYderives t$ and $T \DYderives u$.
\end{restatable} 

Another property we desire of $(T;E)$ is \emph{consistency} -- one should not be able to prove absurdities like $\equals{(m,n)}{\{p\}_{k}}$ or $\equals{m}{(m,n)}$. This idea is formalised below. 
\begin{definition}\label{def:consistent}
    $(T;E)$ is \defemph{consistent} if there is a ground substitution $\lambda$ s.t.\ $\lambda(t) = \lambda(u)$ for each $\equals{t}{u} \in E$.
\end{definition} 

Hereafter, we will only consider pure and consistent $(T; E)$. 
%We will ensure that in all derivability checks of the form $(S;A) \assderives \alpha$ that arise in the context of protocols, $\dc(S;A)$ is consistent. 

A proof $\pi$ of $(T;E) \vdash \exists{}x_{1}\ldots{}x_{k}.(\equals{t}{u})$ can be separated into a proof of $\equals{t}{u}$ with each $x_{i}$ instantiated by some witness, followed by $k$ applications of $\rnexintro$. This leads to the following theorem (whose proof is in Appendix~\ref{app:derivability}).
%To apply $\rnexintro$, some prerequisities must be met. These are formalized in the theorem below.
\begin{restatable}{theorem}{eqstogamma}\label{thm:eqs-to-gamma}
    Let $\alpha = \exists{}x_{1}\ldots{}x_{k}.(\equals{t}{u})$, and $(S;A)$ be s.t. $\boundvars(\alpha) \cap \varsof(S;A) = \emptyset$, with $(T;E) = \dc(S;A)$. $(S;A) \assderives \alpha$ iff there is a substitution $\mu$ with $\dom(\mu) = \boundvars(\alpha)$ s.t.: 
    \begin{enumerate}[label={[\arabic*]}]
        \item\label{item:dycond} $\forall{}x \in \dom(\mu): T \DYderives \mu(x)$.
        \item\label{item:abscond} $\forall{}x \in \dom(\mu), r \in \{t,u\}: \posof{x}{r} \subseteq \abstractable(T \cup \dom(\mu), r)$.
        \item\label{item:eqcond} $(T;E) \eqderives \equals{\mu(t)}{\mu(u)}$.  
    \end{enumerate} 
\end{restatable}

To decide whether $(S;A) \assderives \alpha$, one can guess a $\mu$ and check if the above conditions are satisfied. This amounts to guessing $\mu(x)$ for each $x \in \boundvars(\alpha)$, so it would be good to have a bound on  the size of $\mu(x)$. We choose the measure for size of a term to be the number of distinct subterms. 
\begin{definition} 
    A substitution $\lambda$ is said to be $M$-bounded if $\dagsize{\lambda(x)} \le M$, for all $x \in \dom(\lambda)$. 
\end{definition}

The following theorem, proved in Appendix~\ref{app:decidable-assderives}, is key to solving the derivability problem. 
\begin{theorem}\label{thm:deriv-witness-bounds}
    If there is a $\mu$ satisfying the conditions in Theorem~\ref{thm:eqs-to-gamma}, there is an $M$-bounded $\nu$ satisfying the same conditions, where $M = |\st(S) \cup \st(A\cup\{\alpha\})|$. 
\end{theorem}

After obtaining a small substitution $\nu$, we still have to check if $(T;E) \eqderives \equals{\nu(t)}{\nu(u)}$. We solve this using normalization and subterm property for $\eqderives$ (proved in Appendix~\ref{app:normalization}).  

\begin{restatable}{definition}{eqnormal}\label{def:eqnormal}
	An $\eqderives$ proof $\pi$ of $T; E \vdash \equals{t}{u}$ is \defemph{normal} if:
	\begin{itemize}[leftmargin=*]
		\item All $\DYderives$ subproofs are normal
		\item The premise of $\rnsymm$ can only be the conclusion of $\rnax$
		\item The premise of $\rneq$ can only be the conclusion of a destructor
		\item No premise of a $\rntrans$ is of the form $\equals{a}{a}$ 
		\item The conclusion of a $\rntrans$ is not a premise of $\rntrans$
		\item Adjacent premises of a $\rntrans$ are not conclusions of $\rncons$
		\item No subproof ending in $\rnproj$ contains $\rncons$. 
	\end{itemize}
\end{restatable}

\begin{restatable}{theorem}{normeq}\label{thm:normeq}
    If $(T;E) \eqderives \equals{t}{u}$, then $(T;E) \eqderives \equals{t}{u}$ via a normal proof.
\end{restatable}

\begin{restatable}[Subterm property]{theorem}{subtermeq}\label{thm:subtermeq}
    For any normal proof $\pi$ of $T; E \eqderives \equals{t}{u}$, then $\termsof(\pi) \subseteq \st(T \cup \{t,u\}) \cup \st(E)$. If $\rncons$ does not occur in $\pi$, then $\termsof(\pi) \subseteq \st(T) \cup \st(E)$.
\end{restatable} 

%Theorems~\ref{thm:normeq} and \ref{thm:subtermeq} are proved in Appendix~\ref{app:normalization}.

To check whether $(T;E) \eqderives \equals{t}{u}$ we compute the set $\conseq{T}{E}{t}{u}$ given by
\[\bigl\{\equals{r}{s} \mid r,s \in \st(T\cup\{t,u\}) \cup \st(E) \text{ and } (T; E) \eqderives \equals{r}{s}\bigr\}\]
and then check whether $\equals{t}{u} \in \conseq{T}{E}{t}{u}$. The set $\conseq{T}{E}{t}{u}$ can be computed using a saturation-based procedure that runs in polynomial time. The details are provided in Appendix~\ref{app:conseq}.
%
%Finally, we can sketch an NP procedure for checking $\assderives$. By Theorems~\ref check whether $(S; A) \assderives \exists{}x_{1}\ldots{}x_{k}.(\equals{t}{u})$. By Theorem~\ref{thm:eqs-to-gamma}, this reduces to checking a set of derivabilities from $(T; E) = \dc(S; A)$ under a fixed, $M$-bounded $\nu$ (by Theorem~\ref{thm:deriv-witness-bounds}). Each such check involves deriving $\equals{\nu(r)}{\nu(s)}$ where $\equals{r}{s} \in \SF(\gamma)$. Finally, by Theorem~\ref{thm:coretoeq}, we can use the above PTIME procedure to check these in $\eqderives$. For each $\equals{\nu(r)}{\nu(s)}$, the $\conseq$ set is a subset of $\st(T) \cup \st(E) \cup \st(\nu(\gamma))$. This set is polynomial in $|\st(S) \cup \st(A \cup \{\gamma\})|$, and so derivability for $\assderives$ is in NP. The details are presented in Appendix~\ref{app:conseq}.

\section{Protocols and the insecurity problem}\label{sec:protruns}
A \emph{protocol} is given by a finite set of roles corresponding to the actions of \emph{honest agents}. Each role consists of a finite sequence of alternating receives and sends (each send triggered by a receive). Every sent message is added to the intruder's knowledge base. Each received message is assumed to have come from the intruder, so it must be derivable by the intruder. Only assertions are communicated -- a term $t$ can be modelled via the assertion $\equals{t}{t}$, whose only public term is $t$.
\begin{definition}
    A \defemph{protocol} $\prot$ is a finite set of \defemph{roles}, each of the form $({\beta_{1}},{\alpha_{1}})\ldots({\beta_{m}},{\alpha_{m}})$, where the $\alpha_{i}$s and $\beta_{i}$s are assertions. An $x \in \freevars(\prot)$ is said to be an \emph{agent variable} if it occurs first in an $\alpha_{i}$; otherwise it is an \emph{intruder variable}. 
\end{definition}

\begin{definition} A \defemph{session} of a protocol $\prot$ is a sequence of the form $u:\recsend{\beta_{1}}{\alpha_{1}}\ \cdots\ u:\recsend{\beta_{\ell}}{\alpha_{\ell}}$ where $u \in \ag$ and $(\beta_{1}, \alpha_{1}) \cdots (\beta_{\ell}, \alpha_{\ell})$ is a prefix of a role of $\prot$ with all the agent variables instantiated by values from $\names$. A set of sessions $S$ of $\prot$ is \defemph{coherent} if $\freevars(\xi) \cap \freevars(\xi') = \emptyset$ for distinct $\xi, \xi' \in S$. One can always achieve coherence by renaming intruder variables as necessary.
\end{definition}

\begin{obs}\label{obs:earlierbetaj}
Since agent variables are mapped to names, the only free variables in sessions are intruder variables. Thus, for $i \leq \ell$ and any $x \in \freevars(\alpha_{i})$, there is a $j < i$ s.t.\ $x \in \freevars(\beta_{j})$.
\end{obs}

A run is an interleaving of sessions where each message sent by an agent should be constructible from their knowledge. This is formalised by the notion of \emph{knowledge state}. 
\begin{definition}
    A \defemph{knowledge state} is a pair $(X; \Phi)$ where $X$ is a finite set of terms and $\Phi$ is a finite set of assertions. A \defemph{knowledge function} $\knowfunc$ is such that $\dom(\knowfunc) = \ag$ and for each $a \in \ag$, $\knowfunc(a)$ is a knowledge state. 

    Given a knowledge state $(X; \Phi)$ and an assertion $\alpha$, we define $\knowupd((X;\Phi), \alpha) \coloneqq (X \cup \publics(\alpha), \Phi \cup \{\alpha\})$.
\end{definition}

\begin{definition}\label{def:protruns}
    A \defemph{run} of a protocol $\prot$ is a pair $(\xi, \intsub)$ where:
    \begin{itemize}[leftmargin=*]
		\item $\xi \coloneqq u_{1}:\recsend{\beta_{1}}{\alpha_{1}}, \ldots, u_{n}:\recsend{\beta_{n}}{\alpha_{n}}$ is an interleaving of a finite, coherent set of sessions of $\prot$.
        \item $\intsub$ is a ground substitution with $\dom(\intsub) = \freevars(\xi)$.
        \item There is a sequence $\knowfunc_{0} \ldots \knowfunc_{n}$ of knowledge functions s.t.:
        \begin{itemize}
            \item $\knowfunc_{0}(a) = (X_{a}; \emptyset)$, where $X_{a}$ is a finite set of initial terms known to $a$ ($a$'s secret key, public keys, public names etc). 
            \item For all $i < n$, 
            \[
                \knowfunc_{i+1}(a) = 
                \begin{cases}
                    \knowfunc_{i}(a) & \mbox{if $a \neq u_{i}, a \neq I$} \\ 
                    \knowupd(\knowfunc_{i}(a), \beta_{i}) & \mbox{if $a = u_{i}$} \\
                    \knowupd(\knowfunc_{i}(a), \alpha_{i}) & \mbox{if $a = I$}
                \end{cases}
            \]
            \item For $i \leq n$, $\knowfunc_{i}(u_{i}) \assderives \alpha_{i}$ and $\intsub(\knowfunc_{i-1}(I)) \assderives \intsub(\beta_{i})$.
        \end{itemize}
    \end{itemize}
\end{definition}

The $\beta_{i}$s are implicit sends by $I$, and added to $u_{i}$'s state, while the $\alpha_{i}$s are sent by $u_{i}$, and added to $I$'s state. Every concrete message $\intsub(\beta_{i})$ should be derivable from the concrete knowledge state of $I$, but it suffices for each $\alpha_{i}$ to be derived from $u_{i}$'s state even without $\intsub$. This models the usual expectation of protocols -- that honest agent sends are enabled based on messages received earlier in their role, and not on accidental unification with terms generated by the intruder. We explicitly model honest agent derivability so as to allow \emph{conditional actions}, i.e. to take some action only if an assertion is derivable, and abort otherwise.  

\begin{definition}
    A \defemph{secrecy property} is given by an assertion $\gamma$ which the intruder should not know. A \defemph{$K$-bounded attack} which violates the secrecy of $\gamma$ is a run of the protocol with at most $K$ sessions where $\intsub(\knowfunc_{n}(I)) \assderives \intsub(\gamma)$.
\end{definition}

\begin{definition}[$K$-bounded insecurity problem]
    Given a protocol $\prot$, an assertion $\gamma$, and $K \in \mathbb{N}$, check if there exists a $K$-bounded attack on $\prot$ violating the secrecy of $\gamma$. 
\end{definition}

%\begin{definition}\label{def:protruns}
%    A \emph{run} of a protocol $\prot$ is a pair $(\xi, \intsub)$ where:
%    \begin{itemize}[leftmargin=*]
%    \item $\xi \coloneqq ({\beta_{1}},{\alpha_{1}})\ldots({\beta_{n}},{\alpha_{n}})$ is an interleaving of $\prot$. 
%        \item $\intsub$ is a ground substitution with $\dom(\intsub) = \freevars(\xi)$.
%        \item There are two sequences of states $(X_{0},\Phi_{0})\cdots(X_{n},\Phi_{n})$ and $(Y_{0},\Psi_{0})\cdots(Y_{n},\Psi_{n})$ s.t.
%        \begin{itemize}
%            \item $X_{0}$ and $Y_{0}$ are finite sets of terms, and $\Phi_{0} = \Psi_{0} = \emptyset$. 
%            \item For all $i < n$, $(X_{i+1};\Phi_{i+1}) = \knowupd((X_{i};\Phi_{i}), \alpha_{i})$ and $(Y_{i+1};\Psi_{i+1}) = \knowupd((Y_{i};\Psi_{i}), \beta_{i})$            
%            \item For $i \leq n$, $\intsub(X_{i-1};\Phi_{i-1}) \assderives \intsub(\beta_{i})$ and $(Y_{i};\Psi_{i}) \assderives \alpha_{i}$
%        \end{itemize}
%    \end{itemize}
%\end{definition}

\section{Insecurity is in NP}\label{sec:insecurity}
Given a protocol $\prot$, a secrecy property specified by $\gamma$ and a bound $K$, we now describe an NP algorithm to check if there is a $K$-bounded attack. To do this, we need to guess a coherent set of sessions of size $K$ and an interleaving $\xi = u_{1}:\recsend{\beta_{1}}{\alpha_{1}}, \ldots, u_{n}:\recsend{\beta_{n}}{\alpha_{n}}$. We also need to guess a substitution $\intsub$ with $\dom(\intsub) = \freevars(\xi)$, and check that $(\xi, \intsub)$ satisfies the conditions in Definition~\ref{def:protruns}. 
We need to ensure that $\intsub$ is $M$-bounded for some $M$ which is polynomial in $K$ and the sizes of $\prot$ and $\gamma$. For this, we apply Lemma~\ref{lem:breakdownass} and Theorem~\ref{thm:eqs-to-gamma} to each derivability of the form $\intsub(\knowfunc_{i-1}(I)) \assderives \intsub(\beta_{i})$. Thus we have a set of substitutions $\{\intsub, \iwsub_{1}, \dots, \iwsub_{n}\}$ s.t.\ $\dc(\intsub(\knowfunc_{i-1}(I))) \eqderives \equals{\intsub\iwsub_{i}(r_{i})}{\intsub\iwsub_{i}(s_{i})}$ for each $i \leq n$ and $r_{i}$ and $s_{i}$ as appropriate.  

We need to consider all these substitutions together and find $M$-bounded equivalents $\vintsub, \viwsub_{1}, \ldots, \viwsub_{n}$ such that $\vintsub(\knowfunc_{i-1}(I)) \eqderives \equals{\vintsub\viwsub_{i}(r_{i})}{\vintsub\viwsub_{i}(s_{i})}$ for each $i \leq n$. 
%We need to consider all these substitutions together to preserve equalities of the form $\equals{x}{y}$, where $x \in \dom(\intsub)$ and $y \in \boundvars(\beta_{i})$ for some $i$. 
Thus, it suffices to guess small substitutions and check that these derivabilities hold for the intruder. For honest agent derivations of the form $\knowfunc_{i}(u_{i}) \assderives \alpha_{i}$, we can use the NP algorithm to the derivability problem outlined in Section~\ref{sec:derivability}. Thus we get an NP algorithm for insecurity. 

An important difference between the derivability problem and the insecurity problem is the following. In the derivability problem, the LHS is fixed, the $\intsub$ is ``already applied'', and we only need to show that if there is a witness substitution $\mu$ that satisfies the conditions of Theorem~\ref{thm:eqs-to-gamma}, there is a bounded substitution $\nu$ satisfying the same. However, for the insecurity problem, $\intsub$ appears on the LHS of derivabilities, and preserving them even after changing $\intsub$ to $\vintsub$ is a challenge.

\subsection{Preliminaries}
We fix a protocol $\prot$ and a run $(\xi, \intsub)$ of $\prot$. By renaming variables if necessary, we can ensure that $\freevars(\xi) \cap \qvar = \emptyset$. Since $\dom(\intsub) = \freevars(\xi)$, we have $\intsub(x) = x$ for all $x \in \qvar$. It follows that $\intsub(\dc(S;A)) = \dc(\intsub(S;A))$, for any $(S;A)$. 

For the rest of this section, we fix the following notation. 
\begin{itemize}
\item $(T_{i};E_{i}) \coloneqq \dc(\knowfunc_{i}(I))$ and $(U_{i};F_{i}) \coloneqq \dc(\knowfunc_{i}(u_{i}))$.
\item $\intterms_{i} \coloneqq \publics(\beta_{i})$ and $\honterms_{i} \coloneqq \publics(\alpha_{i})$
%\footnote{These stand for intruder terms and honest agent terms respectively.}
\end{itemize}
Note that $T_{i} \subseteq T_{i+1}$ and $E_{i} \subseteq E_{i+1}$ for every $i$.
\begin{restatable}{obs}{varsindomsigma}\label{obs:varsTi-in-domsigma}
    Any $t \in T_{i}$ is either a bound variable from $\beta_{j}$ (and hence in $\qvar$) or in $\publics(\beta_{j})$ for some $j$. In the latter case, $t$ contains no variable from $\qvar$, and thus, $\varsof(t) \cap \boundvars(\xi) = \emptyset$, i.e. $\varsof(t) \subseteq \freevars(\xi) = \dom(\intsub)$. Reasoning similarly for $U_{i}$, we see that for any $t \in T_{i} \cup U_{i}$, $\varsof(t) \subseteq \dom(\intsub)$ or $t \in \qvar$. 
\end{restatable}
%\begin{proof}
%    Note that any $t \in T_{i}$ is either a bound variable from $\beta_{j}$ or in $\publics(\beta_{j})$ for some $j$, and similarly for $U_{i}$ with some $\alpha_{j}$. So any $t \in T_{i} \cup U_{i}$ belongs to $\qvar$, or is in $\publics(\gamma)$ for some $\gamma$. In the latter case, $t$ contains no variable from $\qvar$, and thus, $\varsof(t) \cap \boundvars(\xi) = \emptyset$, i.e. $\varsof(t) \subseteq \freevars(\xi) = \dom(\intsub)$. 
%\end{proof}

Applying Theorem~\ref{thm:eqs-to-gamma} to the derivabilities in Definition~\ref{def:protruns}, for every $i \leq n$ we get substitutions $\iwsub_{i}$ and $\hwsub_{i}$ (with respective domains $\boundvars(\beta_{i})$ and $\boundvars(\alpha_{i})$) s.t.: 
\begin{itemize}
    \item for every $x \in \dom(\iwsub_{i})$, $\intsub(T_{i-1}) \DYderives \iwsub_{i}(x)$
        \item $\intsub(T_{i-1}; E_{i-1}) \eqderives \intsub\iwsub_{i}(\equals{r}{s})$ for $\equals{r}{s} \in \SF(\beta_{i})$
    \item for every $x \in \dom(\hwsub_{i})$, $\intsub(U_{i}) \DYderives \hwsub_{i}(x)$
        \item $\intsub(U_{i}; F_{i}) \eqderives \intsub\hwsub_{i}(\equals{r}{s})$, where $\equals{r}{s} \in \SF(\alpha_{i})$
\end{itemize}

Note that for $i \leq n$, $\knowfunc_{i}(u_{i}) \assderives \alpha_{i}$, and thus $\intsub(\knowfunc_{i}(u_{i})) \assderives \intsub(\alpha_{i})$. So we have a $\hwsub_{i}$ for honest agent derivations.

Define $\Substs \coloneqq \{\intsub, \hwsub_{i}, \iwsub_{i} \mid i \leq n\}$. We can ensure that distinct $\lambda, \lambda' \in \Substs$ have disjoint domains. Assume that there is an $\fixedname \in T_{0} \cap \names$ s.t.\ $\fixedname \notin \st(\{\alpha_{i}, \beta_{i}\}) \cup \st(\rng(\hwsub_{i}) \cup \rng(\iwsub_{i}))$ for all $i$. This can be thought of as a fixed ``spare name'' that does not appear in the run.  
 
\begin{obs}
    For every $i \le n$ and for every $x \in \dom(\iwsub_{i})$, $\intsub(T_{i-1}) \DYderives \iwsub_{i}(x)$, and so any variable in $\iwsub_{i}(x)$ must come from $\varsof(\intsub(T_{i-1})) \subseteq \bigcup_{j<i}(\dom(\iwsub_{j}) \cup \dom(\hwsub_{j}))$. Similarly for any $x \in \dom(\hwsub_{i})$, $\varsof(\hwsub_{i}(x)) \subseteq \bigcup_{j<i}(\dom(\iwsub_{j}) \cup \dom(\hwsub_{j})) \cup \dom(\iwsub_{i})$. 
%    
%    So for any $\lambda \in \Substs$, $\varsof(\rng(\lambda)) \subseteq \qvar$, i.e. $\varsof(\rng(\lambda)) \cap \dom(\intsub) = \emptyset$, and hence $\intsub \cup \lambda = \intsub \circ \lambda$, denoted by $\intsub\lambda$.  
\end{obs}
%
%Theorem~\ref{thm:eqs-to-gamma} also guarantees proofs of the following.
%\begin{enumerate}[label={[$P_{\arabic*}$]}]
%    \item \label{item:intdy} $\intsub(T_{i-1}) \DYderives \intsub(r)$ for all $r \in \intterms_{i}$ 
%    \item \label{item:inteq} $\intsub(T_{i-1}; E_{i-1}) \eqderives \intsub\iwsub_{i}(\equals{r}{s})$, where $\equals{r}{s} \in \SF(\beta_{i})$
%    \item $\intsub(U_{i}) \DYderives \intsub(r)$ for all $r \in \honterms_{i}$
%    \item $\intsub(U_{i}; F_{i}) \eqderives \intsub\hwsub_{i}(\equals{r}{s})$, where $\equals{r}{s} \in \SF(\alpha_{i})$
%\end{enumerate}

Since each knowledge state arises from a sequence of updates to a sanitized initial state, each $(T_{i}; E_{i})$ and $(U_{i}; F_{i})$ is pure. But we also need them to be consistent (as in Definition~\ref{def:consistent}). Towards this, we define a substitution $\bigsub$, which is an appropriate composition of the substitutions in $\Substs$. We motivate this by the following example.

\begin{example}
Suppose $y \in \boundvars(\beta_{1})$, and $x \in \boundvars(\alpha_{2})$. Consider a situation where $\hwsub_{2}(x) = \{y\}_{k}$ and $\iwsub_{1}(y) = (m_{0}, m_{1})$. Also suppose $(T_{2};E_{2}) \vdash \equals{x}{z}$ for some $z \in \dom(\intsub)$. For consistency, we need a ground $\lambda$ s.t. $\lambda(x) = \lambda(z)$. 
%Thus, we need to have $\intsub \subseteq \lambda$ and $\iwsub_{1}\circ\hwsub_{2} \subseteq \lambda$. 
We can take $\lambda$ to be $\intsub\iwsub_{1}\hwsub_{2}$. We see that $\lambda(x) = \intsub(\iwsub_{1}(\hwsub_{2}(x))) = \intsub(\iwsub_{1}(\{y\}_{k})) = \intsub(\{(m_{0},m_{1})\}_{k}) = \{(m_{0},m_{1})\}_{k}$. Observe that $\dom(\lambda) = \dom(\intsub) \cup \dom(\iwsub_{1}) \cup \dom(\hwsub_{2})$, and since $z \notin \dom(\iwsub_{1}) \cup \dom(\hwsub_{2})$, $\lambda(z) = \intsub(z)$. 
\end{example}

\begin{definition}\label{def:bigsub}
$\bigsub \coloneqq \intsub\iwsub_{1}\hwsub_{1}\ldots\iwsub_{n}\hwsub_{n}$. 
%	$\bigsub$ is the ground substitution defined as follows:
%	\begin{itemize}[leftmargin=*]
%		\item For all $x \in \dom(\intsub)$, $\bigsub(x) \coloneqq \intsub(x)$
%		\item For all $x \in \dom(\iwsub_{i})$, assuming $\bigsub(y)$ is defined for all $y \in \dom(\intsub) \cup \dom(\hwsub_{j}) \cup \dom(\iwsub_{j})$ for $j < i$, $\bigsub(x) \coloneqq \bigsub(\iwsub_{i}(x))$.
%		\item For all $x \in \dom(\hwsub_{i})$, assuming $\bigsub(y)$ is defined for all $y \in \dom(\intsub) \cup \dom(\hwsub_{j}) \cup \dom(\iwsub_{l})$ for $j < i$ and $l \leq i$, $\bigsub(x) \coloneqq \bigsub(\hwsub_{i}(x))$.
%	\end{itemize}
\end{definition}

\begin{obs}\label{obs:bigsublambda}
$\bigsub$ is ground, and for $\lambda \in \Substs, \bigsub(\lambda(x)) = \bigsub(x)$.
\end{obs}

The next two lemmas show that each $E_{i}$ and $F_{i}$ is consistent and that derivations preserve consistency. Any proofs omitted hereon can be found in Appendix~\ref{app:insecurity}.
\begin{restatable}{lemma}{equndertheta}\label{lem:equndertheta}
	Suppose $\lambda$ is such that $\lambda(r) = \lambda(s)$ for each $\equals{r}{s} \in E$, and $T; E \eqderives \equals{t}{u}$. Then $\lambda(t) = \lambda(u)$.
\end{restatable}

\begin{restatable}{lemma}{prooftotruth}\label{lem:prooftotruth}
	For any $i \in \{1, \ldots, n\}$, 
	\begin{enumerate}[leftmargin=*]
		\item if $\equals{t}{u} \in E_{i} \cup F_{i}$, then $\bigsub(t) = \bigsub(u)$.
		\item if $\intsub(T_{i-1}; E_{i-1}) \eqderives \intsub\iwsub_{i}(\equals{t}{u})$, then $\bigsub(t) = \bigsub(u)$.
		\item if $\intsub(U_{i}; F_{i}) \eqderives \intsub\hwsub_{i}(\equals{t}{u})$, then $\bigsub(t) = \bigsub(u)$.
	\end{enumerate}
\end{restatable}

We developed this preliminary setup for both honest agent derivations as well as intruder derivations in order to demonstrate the interplay between $\hwsub$ and $\iwsub$, as evidenced in the definition of $\bigsub$. However, the insecurity problem itself is concerned only with intruder derivability, and therefore, in the next few sections we will focus on $\intterms_{i}, \beta_{i}, (T_{i}; E_{i})$, and $\iwsub_{i}$.  

\subsection{Typed \texorpdfstring{$\DYderives$}{DY} proofs}
\begin{definition}[Types]
%The set of terms occurring in $\xi$ (before applying any substitution).
We use the sets $\constst$ (consisting of the terms occurring in $\xi$ before applying any substitution) and $\stnonvars$ (the same set, but without variables) to \defemph{type} the terms appearing in any proof. They are defined as follows.
\[ \constst \coloneqq \bigcup_{i \leq n}\bigl\{\bigl(\st(T_{i} \cup U_{i}) \cup \st(E_{i} \cup F_{i})\bigr)\bigr\} \quad \quad \stnonvars \coloneqq \constst \setminus \vars \]
\end{definition}
%We use $\constst$ to denote the set of terms that occur in $\xi$, i.e., 
%$\bigcup_{i \leq n}\{\bigl(\st(T_{i} \cup U_{i}) \cup \st(E_{i} \cup F_{i})\bigr)\}$. We will also denote by $\stnonvars$ the subset $\constst \setminus \vars$. 

We show that under $\intsub$, any term $t$ with no type in $\constst$ appears in a received assertion $\beta$ first, and is generated by the intruder by putting information together, i.e. via a normal proof ending in a constructor rule. Using this, we show that every proof from $\intsub(T_{i})$ can be turned into a ``typed'' proof. Typed proofs help preserve derivability even after zapping variables.

Consider a proof $\pi$ of $\intsub(T_{i}) \vdash t$ for some $t$. It is possible that $\pi$ mentions terms from outside $\bigsub(\constst)$, even if $t \in \bigsub(\constst)$. In a ``typed'' equivalent of $\pi$, every subproof deriving such terms ends in a constructor rule. If a proof ends in a destructor rule, however, one can anchor the conclusion to some term in $\constst$, i.e. provide a ``type'' from $\constst$. 

\begin{definition}\label{def:welltyped}[Typed $\DYderives$ proof]
	A $\DYderives$ proof $\pi$ is \defemph{typed} if for each subproof $\pi'$, either $\pi'$ ends in a constructor rule, or $\concof(\pi') \in \intsub(\stnonvars) \cup \qvar$. 
\end{definition}

If a term $t$ does not match any pattern mentioned in the protocol, and appears in the intruder's knowledge set (even as a subterm of some term in the set), then it must have appeared first in an intruder send. We formalize this as follows.
\begin{restatable}{lemma}{sigmaxinsigmaTp}\label{lem:sigmax-in-sigmaTp}
    Suppose $t \notin \intsub(\stnonvars) \cup \qvar$. For  $i \leq n$, if $t \in \st(\intsub(T_{i}))$, then there is a $k < i$ such that $t \in \st(\intsub(\intterms_{k}))$.
\end{restatable}

If a term as above is derived by the intruder at some stage $i$, the derivation either ends in a construction rule, so $t$ is freshly put together at this stage, or there is an earlier point at which the intruder derived $t$. Thus, terms that cannot be provided a type from $\stnonvars$ do not originate in an honest agent send. 
\begin{restatable}{lemma}{earlierproof}\label{lem:earlier-proof}
    Suppose $i \leq n$, $t \notin \intsub(\stnonvars) \cup \qvar$ and $\intsub(T_{i}) \DYderives t$ via a normal proof $\pi$ ending in a destructor rule. Then there is an $\ell < i$ such that $\intsub(T_{\ell}) \DYderives t$.
\end{restatable}

\begin{theorem}\label{thm:rustur}
	For all $t$ and all $i \in \{0,\ldots,n\}$, if $\intsub(T_{i}) \DYderives t$, then there is a typed normal proof $\pi^{*}$ of the same.
\end{theorem}
\begin{proof}
    Assume the theorem holds for all $t$ and for all $j < i$. We show how to transform any proof $\pi$ of $\intsub(T_{i}) \vdash t$ into a typed normal proof $\pi^{*}$ of the same. 
    \begin{itemize}[leftmargin=*]
        \item {$\pi$ ends in $\textsf{ax}$:} $t \in \intsub(T_{i}) \subseteq \intsub(\constst)$. 
        If $t \in \intsub(\stnonvars) \cup \qvar$, we take $\pi^{*}$ to be $\pi$ itself. Otherwise, by Lemma~\ref{lem:earlier-proof}, there is some $j < i$ such that $\intsub(T_{j}) \DYderives t$. Since $j < i$, we can get a typed normal proof $\pi^{*}$ of $\intsub(T_{j}) \vdash t$ and obtain the required result by weakening the LHS.

        \item {$\pi$ ends in a constructor:} We can find typed normal equivalents for all immediate subproofs, and apply the same constructor rule to get the desired $\pi^{*}$. 
        
        \item {$\pi$ ends in a destructor:} Let $\pi_{1}$ and $\pi_{2}$ be the immediate subproofs of $\pi$, with $\concof(\pi_{1}) = s$, where $t$ is an immediate subterm of $s$. We can find typed normal equivalents $\pi^{*}_{1}$ and $\pi^{*}_{2}$. If $\pi^{*}_{1}$ ends in a constructor, then we choose $\pi^{*}$ to be the immediate subproof of $\pi^{*}_{1}$ s.t.\ $\concof(\pi^{*}) = t$. $\pi^{*}_{1}$ is typed normal, and so is $\pi^{*}$. 
        
        If $\pi^{*}_{1}$ does not end in a constructor, $s \in \intsub(\stnonvars) \cup \qvar$. Since a destructor rule $\rnrule$ was applied on $s$, $s \notin \qvar$. So $s \in \intsub(\stnonvars)$, and hence $t \in \intsub(\constst)$. If $t \in \intsub(\stnonvars) \cup \qvar$, we obtain a typed normal $\pi^{*}$ by applying $\rnrule$ on $\pi^{*}_{1}$. Otherwise, as with $\rnax$, by Lemma~\ref{lem:earlier-proof} and weakening, we get a typed and normal proof $\pi^{*}$ of $\intsub(T_{i}) \vdash t$. \qedhere 
    \end{itemize}
\end{proof}

\subsection{Typed equality proofs}
We now define a notion of ``minimal variables'', which are variables that do not unify with any non-atomic term mentioned in the protocol. The idea is that images of minimal variables can be freely ``zapped'' to an atomic constant. 
\begin{definition}[Minimal variable]
    $x \in \dom(\bigsub)$ is \defemph{minimal} if there is no $t \in \stnonvars$ such that $\bigsub(x) = \bigsub(t)$.
\end{definition}

\begin{definition}
    A term $t$ is \defemph{zappable} if there is a minimal $x$ such that $\bigsub(t) = \bigsub(x)$. 
\end{definition}

\begin{obs}\label{obs:tzap-iff-uzap}
\phantom{a}
\begin{itemize}
\item If a term $t$ is zappable, then $t \notin \stnonvars$.
\item If a term $t \in \bigsub(\constst)$ is not zappable, then $t \in \bigsub(\stnonvars)$. 
\item For $t, u$ s.t.\ $\bigsub(t) = \bigsub(u)$, $t$ is zappable iff $u$ is zappable. 
\end{itemize}
\end{obs}

%\begin{proof}
%    Suppose $t$ is zappable. Then $\bigsub(t) = \bigsub(x)$ for some minimal $x$, but then $\bigsub(u) = \bigsub(x)$ as well, so $u$ is zappable. We can similarly show that $u$ is zappable implies that $t$ is zappable.
%\end{proof}

\begin{definition}[Typed term]
    A term $t$ is \defemph{typed} if $t \in \intsub(\stnonvars) \cup \bigsub(\constst) \cup \qvar$. 
\end{definition} 

\begin{example}
Note that we consider $\intsub(\stnonvars)$ separately from $\bigsub(\constst)$. Consider a term of the form $(m, x) \in \stnonvars$, where $x \notin \dom(\intsub)$. $\intsub((m, x)) = (m, x)$, but this cannot be in $\bigsub(\constst)$, since $\bigsub(\constst)$ only contains ground terms. Thus, $\intsub(\stnonvars) \not\subseteq \bigsub(\constst)$.
\end{example}

We now consider a small example which will motivate our choices for the definition of a \emph{typed $\eqderives$ proof}.

\begin{example}
    Suppose $\intsub(x) = (t_{1}, t_{2})$ for some minimal variable $x$, and $\intsub(u) = (u_{1}, u_{2})$ for some term $u$. Suppose we also have a proof $\pi$ of $\equals{t_{1}}{u_{1}}$ obtained by applying $\rnproj_{1}$ to $\equals{\intsub(x)}{\intsub(u)}$, and we want to derive a ``corresponding'' equality assertion under a small $\vintsub$. The most straightforward strategy is to follow the structure of $\pi$. However, under $\vintsub$, we would ``zap'' $x$ to an atomic value. We cannot apply $\rnproj$ to this, so we can no longer preserve the structure of $\pi$. 
\end{example}

To avoid such situations, we define a typed $\eqderives$ proof as one 
%which does not allow projections on images of minimal variables, and 
whose structure can be preserved under zapping. We then show that any $\eqderives$ proof can be simulated by a typed one.

\begin{definition}\label{def:welltypedeq}[Typed $\eqderives$ proof]
    A proof $\pi$ of $X; A \vdash \equals{r}{s}$ is \defemph{typed} if for every subproof $\pi'$ with conclusion $X; A \vdash \equals{t}{u}$,
    \begin{itemize}
        \item $\pi'$ contains an occurrence of the $\rncons$ rule, or
        \item $t = u$, or 
        \item $t$ and $u$ are typed terms.
    \end{itemize}
\end{definition}

Observe that for any $t \in \intsub(\constst)$, we have $t \in \intsub(\stnonvars)$, or $t = \intsub(x) = x$ for some $x \notin \dom(\intsub)$ (in which case $t \in \qvar$), or $t = \intsub(x)$ for $x \in \dom(\intsub)$ (in which case $t = \bigsub(x)$, so $t \in \bigsub(\constst)$). Thus $\intsub(\constst) \subseteq \intsub(\stnonvars) \cup \bigsub(\constst) \cup \qvar$, i.e. every term in $\intsub(\constst)$ is typed. Using this, we can prove the following lemma.

\begin{restatable}{lemma}{projcase}\label{lem:projcase}
    Suppose $t = \func(t_{0},t_{1})$ and $u = \func(u_{0},u_{1})$ are typed, and $\bigsub(t) = \bigsub(u)$. One of the following is true:
 \begin{itemize}[leftmargin=*]
        \item $t$ and $u$ are not zappable, and $t_{0}, t_{1}, u_{0}, u_{1}$ are typed, or
        \item $t$ and $u$ are zappable, and $t = u$.  
    \end{itemize}   
\end{restatable}

\begin{restatable}{theorem}{rustureq}\label{thm:rustur-eq}
    For all $i \leq n$, every normal $\eqderives$ proof from $\intsub(T_{i}; E_{i})$ is a typed proof.
\end{restatable}
\begin{proof}
    Let $\pi$ be a normal $\eqderives$ proof of $\intsub(T_{i}; E_{i}) \vdash \equals{t}{u}$. Assume all proper subproofs of $\pi$ are typed. We do a case analysis on the last rule of $\pi$. Most of the cases are straightforward, so here we only show the case when $\pi$ ends in $\rnproj$. The full proof is presented in Appendix~\ref{app:insecurity}. 
    
     Suppose there are $a = \func(a_{0},a_{1})$ and $b = \func(b_{0},b_{1})$ such that $t = a_{0}$ and $u = b_{0}$ and the immediate subproof $\pi'$ proves $\equals{a}{b}$. By Lemma~\ref{lem:prooftotruth} and Observation~\ref{obs:bigsublambda}, $\bigsub(a) = \bigsub(b)$. By normality, $\rncons$ does not occur in $\pi$ or $\pi'$, and since $\pi'$ is typed, two cases arise.
        \begin{itemize}[leftmargin=*]
            \item {$a = b$:} $t$ and $u$ are immediate subterms of $a$ and $b$, so $t = u$ and $\pi$ is typed.
            \item {$a$ and $b$ are typed:} We have $\bigsub(a) = \bigsub(b)$. Hence, by Lemma~\ref{lem:projcase}, either $a = b$ (so $t = u$ as above) or $t$ and $u$ are typed. In both cases, $\pi$ is typed. \qedhere
        \end{itemize} 
\end{proof}

\subsection{Small substitutions \texorpdfstring{$\vintsub, \vbigsub$, and $\viwsub_{i}$}{}}
Recall that there is a fixed name $\displaystyle\fixedname \in \constst \setminus \st(\bigcup_{\lambda \in \Substs}\rng(\lambda))$. %for any $\lambda \in \Substs$. 

\begin{definition}
    For any term $t$, we inductively define the \defemph{zap} of $t$, denoted $\zap{t}$, as follows: 
    \begin{align*}
        \zap{x} &\coloneqq x \\ 
        \zap{n} &\coloneqq \begin{cases}
        	\fixedname & \qquad \hspace{2mm} \mbox{if $n$ is zappable} \\
            	n & \qquad \hspace{2mm} \mbox{otherwise}
        \end{cases} \\
        \zap{\func(t_{1}, t_{2})} &\coloneqq \begin{cases}
            \fixedname & \mbox{if $\func(t_{1}, t_{2})$ is zappable} \\
            \func(\zap{t_{1}}, \zap{t_{2}}) & \mbox{otherwise}
        \end{cases}
    \end{align*}
    For a set of terms $X$, $\zap{X} \coloneqq \{\zap{t} \mid t \in X\}$. For a set of equalities $E$, $\zap{E} \coloneqq \{\equals{\zap{t}}{\zap{u}} \mid \equals{t}{u} \in E\}$.
\end{definition}

\begin{definition}\label{def:vlambda}
    For $\lambda \in \{\intsub, \bigsub, \iwsub_{i} \mid i \leq n\}$, the \defemph{small substitution} $\vlambda$ corresponding to $\lambda$ is defined as $\vlambda(x) \coloneqq \zap{\lambda(x)}$. 
\end{definition}

\begin{example}\label{ex:zap}
Suppose $\constst = \st(\{\fixedname, y, (y_{1}, \{y_{2}\}_{k}) \})$, where $y_{1}, y_{2}$ are minimal, and $\iwsub_{2}(y) = (y_{1}, \{y_{2}\}_{k})$. Then, $\viwsub_{2}(y) = (y_{1}, \{y_{2}\}_{k})$, i.e. there is no zap even though $y_{1}$ and $y_{2}$ are minimal (because $\zap{y_{i}} = y_{i} \in \vars$). The pair itself is not zapped since it is in $\constst$. The zap for $y_{1}$ and $y_{2}$ occurs when we consider $\vbigsub(y)$, which is $(\fixedname, \{\fixedname\}_{k})$, since $\bigsub(y_{i})$ is a zappable term.

However, consider $\constst = \st(\{\fixedname, y, y_{2}, (y_{1}, x)\})$ and the same $\iwsub_{2}$, with $x$ minimal and $\intsub(x) = \iwsub_{2}(\{y_{2}\}_{k})$. Then, since $\{y_{2}\}_{k} \notin \constst$, we see that $\viwsub_{2}(y) = (y_{1}, \fixedname)$ and $\vbigsub(y) = (\fixedname, \fixedname)$. 

%Suppose $\constst = \st(\{\fixedname, n, y_{3}, (y_{1}, \{(y_{2}, x)\}_{k})\})$. Consider $\intsub, \hwsub_{1}, \iwsub_{2}$ with respective domains $\{x, y_{2}\}, \{y_{1}, y_{3}\}$ and $\{y\}$. Suppose $x$, $y_{2}$ and $y_{3}$ are minimal, $\hwsub_{1}(y_{1}) = n$, $\intsub(\{y_{3}\}_{k}) = \intsub(x)$ and $\iwsub_{2}(y) = (y_{1},\{(y_{2},\{y_{3}\}_{k})\}_{k})$.
%
%Then $\vintsub(x) = \fixedname$, $\viwsub_{2}(y) = (y_{1}, \{(y_{2}, \fixedname)\}_{k})$, and $\vbigsub(y) = (n, \{(\fixedname, \fixedname)\}_{k})$. Note that we zap only $x$, but not any of the other subterms of $y$ containing $x$, since those subterms belong to $\constst$. This zap of $x$ to $\fixedname$ is illustrated by the thin dashes in Figure~\ref{fig:exzap}. We can replace $x$ by $\fixedname$ to get the $\viwsub_{2}(y)$ tree, and replace $y_{1}$ by $n$ (apply $\hwsub_{1}$) to get the one for $\vbigsub(y)$.
%
%Now consider $\constst = \st(\{\fixedname, n, y_{2}, y_{3}, (y_{1},z)\})$ where $z$ is a minimal variable in $\dom(\intsub)$. Consider the same $\intsub, \hwsub_{1}$ and $\iwsub_{2}$ as above, along with $\intsub(z) = \bigsub(\{(y_{2}, x)\}_{k})$. Now $\viwsub_{2}(y) = \vbigsub(y) = (n, \fixedname)$. The term $\{(y_{2},\{y_{3}\}_{k})\}_{k}$ itself is zappable in this case, since it is not in $\constst$, and equivalent under $\bigsub$ to a minimal $z$. Figure~\ref{fig:exzap} illustrates this zap with the thick dashes.
\end{example}
%
%\begin{figure}
%\centering
%\begin{tikzpicture}[-,level/.style={sibling distance = 5cm/#1,level distance = 1cm}] 
%\tikzstyle{every node}=[draw,ellipse]
%\tikzstyle{fixlab}=[draw=none]
%\node (eps) {$\rnpair$}
%	child{ node (y1) {$y_{1}$}}
%    child{ node (enc) {$\rnenc$} 
%            child{ node (p2) {$\rnpair$}
%            		child{ node (y2) {$y_{2}$}} 							child{ node (enc2) {$\rnenc$}
%					child{ node (y3) {$y_{3}$}}
%					child{ node (k2) {$k$}}
%				}
%            }
%            child{ node (k) {$k$}}                            
%    }
%;
%
%  \node (bm1) [inner sep=5pt, fit=(enc) (y2) (k2) (k),draw,rectangle,dashed,line width=0.5mm] {};
%  %\node [fixlab,right of=bm1, node distance=1.8cm] {$\fixedname$};
%  \node (bm2) [fit=(enc2)(y3)(k2),draw,rectangle,dashed] {};
%  %\node [fixlab,right of=bm2, node distance=0.7cm] {$\fixedname$};
%
%\end{tikzpicture}
%\caption{Tree for $y$ for Example~\ref{ex:zap}.}
%\label{fig:exzap}
%\end{figure}

\begin{example}
Let $t = \intsub(x)$ for a minimal $x$. So $\varsof(t) = \emptyset$, and $\vintsub(t) = t$. However, $\zap{\intsub(t)} = \zap{t} = \fixedname$, since $t$ is zappable. Thus, it is not true that $\vlambda(t) = \zap{\lambda(t)}$ for all $t$. 
\end{example}

However, the next lemma shows that this holds for all $t \in \constst$.
\begin{restatable}{lemma}{zaptaunut}\label{lem:zap-taunut}
    For $i \leq n$ and $t \in \constst$, $\vintsub\viwsub_{i}(t) = \zap{\intsub\iwsub_{i}(t)}$. 
\end{restatable}
%\begin{proof}
%The following cases arise.
%    \begin{itemize}
%        \item {$t = x \in \dom(\intsub\lambda)$:} By Definition~\ref{def:vlambda}, $\vintsub\vlambda(x) = \zap{\intsub\lambda(x)}$. 
%        \item {$t = x \notin \dom(\intsub\lambda)$:} $\vintsub\vlambda(x) = x = \zap{x} = \zap{\intsub\lambda(x)}$. 
%        \item {$t \in \names$:} Since $t \in \constst$, $t \in \stnonvars$ and hence is not zappable. Thus $\zap{t} = t$ and $\intsub\lambda(t) = t = \vintsub\vlambda(t)$. Therefore $\vintsub\vlambda(t) = \zap{\intsub\lambda(t)}$.
%        \item {$t = \func(t_{1}, \ldots, t_{k})$:} $t \in \stnonvars$, so $t_{j} \in \constst$ for $j \leq k$, and we get $\vintsub\vlambda(t_{j}) = \zap{\intsub\lambda(t_{j})}$ by IH. We claim that $u = \intsub\lambda(t)$ is not zappable, since for any $x$ such that $\bigsub(u) = \bigsub(x)$, $x$ is not minimal (since $\bigsub(x) = \bigsub(t)$ as well, and $t \in \stnonvars$). Therefore, we have 
%        \begin{align*}
%            \zap{\intsub\lambda(t)} &= \zap{\func(\intsub\lambda(t_{1}), \ldots, \intsub\lambda(t_{k}))} \\
%            &= \func(\zap{\intsub\lambda(t_{1})}, \ldots, \zap{\intsub\lambda(t_{k})}) \\
%            &= \func(\vintsub\vlambda(t_{1}), \ldots, \vintsub\vlambda(t_{k})) = \vintsub\vlambda(t). \qedhere
%        \end{align*}
%    \end{itemize}
%\end{proof}

%\subsection{Simulating $\lambda$ using $\vlambda$}

We now show that we can simulate any $\lambda$ using $\vlambda$ for deriving both terms and equalities. 

\begin{obs}\label{obs:stnu}
$T_{i} \subseteq \constst$, so by Lemma~\ref{lem:zap-taunut}, $\vintsub(T_{i}) = \zap{\intsub(T_{i})}$, and similarly $\vintsub(E_{i}) = \zap{\intsub(E_{i})}$, for all $i \le k$.
\end{obs}

\begin{lemma}\label{lem:dyzetazap}
	For $i \leq n$ and any term $t$, 
	if $\intsub(T_{i}) \DYderives t$ then $\vintsub(T_{i}) \DYderives \zap{t}$.
%	\[\intsub(T_{i}) \DYderives t\ \Longrightarrow\ \vintsub(T_{i}) \DYderives \zap{t}.\]
\end{lemma}
\begin{proof}
    Let $X$ and $Y$ stand for $\intsub(T_{i})$ and $\vintsub(T_{i})$. By Observation~\ref{obs:stnu}, $\zap{X} = Y$. Let $\pi$ be a typed normal $\DYderives$ proof of $X \vdash t$ (ensured by Theorem~\ref{thm:rustur}). We prove that $Y \DYderives \zap{t}$. Consider the last rule $\rnrule$ of $\pi$. The following cases arise.
    \begin{itemize}[leftmargin=*]
		\item {$\rnrule = \textsf{ax}$:} $t \in X$, and therefore $\zap{t} \in Y$. Thus $Y \DYderives \zap{t}$ by $\textsf{ax}$. 
		\item {$\rnrule$ is a constructor:} Let $t = \func(t_{1}, t_{2})$ and let the immediate subproofs of $\pi$ be $\pi_{1}, \pi_{2}$, with $\concof(\pi_{i}) = t_{i}$ for $i \le 2$. By IH, there is a proof $\varpi_{i}$ of $Y \vdash \zap{t_{i}}$ for each $i \leq 2$. If $t$ is zappable, then $\zap{t} = \fixedname \in Y$ ($\fixedname \in T_{i}$ for all $i$, so $\fixedname \in X$ and $\fixedname \in Y$), and we have $Y \DYderives \zap{t}$ using $\textsf{ax}$. If $t$ is not zappable, then $\zap{t} = \zap{\func(t_{1}, t_{2})} = \func(\zap{t_{1}}, \zap{t_{2}})$, and we can apply $\rnrule$ on the $\varpi_{i}$s to get $Y \DYderives \zap{t}$.
		\item {$\rnrule$ is a destructor:} Let the immediate subproofs of $\pi$ be $\pi_{1}, \pi_{2}$, deriving $t_{1}, t_{2}$ respectively, with $t_{1}$ being the major premise, and $t$ an immediate subterm of $t_{1}$. Since $\pi$ is typed normal, $\pi_{1}$ is also typed and ends in a destructor, so by Definition~\ref{def:welltyped}, $t_{1} \in \intsub(\stnonvars) \cup \qvar$. Since we applied a destructor on $t_{1}$, it is not in $\qvar$. Thus, there is some $u_{1} \in \stnonvars$, with the same outermost operator as $t_{1}$, such that $t_{1} = \intsub(u_{1})$. Hence, $\bigsub(t_{1}) = \bigsub(u_{1})$. 
		
		If $t_{1}$ were zappable, there would be a minimal $x$ such that $\bigsub(x) = \bigsub(t_{1}) = \bigsub(u_{1}) \in \bigsub(\stnonvars)$, which contradicts the minimality of $x$. Thus, $t_{1}$ is not zappable, and $\zap{t_{1}}$ has the same outermost structure as $t_{1}$. By IH, there is a proof $\varpi_{i}$ of $Y \vdash \zap{t_{i}}$ for each $i \leq 2$. Since $\zap{t_{1}}$ is not atomic, we can apply the destructor $\rnrule$ on the $\varpi_{i}$s to get $Y \DYderives \zap{t}$. 
 \qedhere
	\end{itemize}
\end{proof}

\begin{restatable}{lemma}{eqzetazap}\label{lem:eqzetazap}
    For $i \leq n$ and terms $t,u$, 
    if $\intsub(T_{i};E_{i}) \eqderives \equals{t}{u}$ then $\vintsub(T_{i};E_{i}) \eqderives \equals{\zap{t}}{\zap{u}}$.
%	\[\intsub(T_{i};E_{i}) \eqderives \equals{t}{u}\ \Longrightarrow\ \vintsub(T_{i};E_{i}) \eqderives \equals{\zap{t}}{\zap{u}}.\]
\end{restatable}
\begin{proof}
    Let $(X;A)$ and $(Y;B)$ denote $\intsub(T_{i};E_{i})$ and $\vintsub(T_{i};E_{i})$ respectively. By Observation~\ref{obs:stnu}, $\zap{X} = Y$ and $\zap{A} = B$. Let $\pi$ be a typed normal $\eqderives$ proof of $X;A \vdash \equals{t}{u}$ (guaranteed by Theorem~\ref{thm:rustur-eq}). We prove that $Y;B \eqderives \equals{\zap{t}}{\zap{u}}$. Most of the cases are straightforward, so here we only consider the cases when $\pi$ ends in $\rnproj$ or $\rncons$. The full proof is in Appendix~\ref{app:insecurity}.
	\begin{itemize}[leftmargin=*]		
	\item {$\pi$ ends in $\rnproj$}: Let the immediate subproof of $\pi$ be $\pi'$ deriving $X;A \vdash \equals{a}{b}$ where $a = \func(a_{0},a_{1})$, $b = \func(b_{0},b_{1})$, and $t = a_{0}$ and $u = b_{0}$. By IH, there is a proof $\varpi'$ of $Y;B \vdash \equals{\zap{a}}{\zap{b}}$. 
        For $\rnproj$, we need $X \DYderives \{a_{0},a_{1},b_{0},b_{1}\}$. By Lemma~\ref{lem:dyzetazap}, $Y \DYderives \{\zap{a_{0}},\zap{a_{1}},\zap{b_{0}},\zap{b_{1}}\}$. By Lemma~\ref{lem:prooftotruth} and Observation~\ref{obs:bigsublambda}, $\bigsub(a) = \bigsub(b)$. By normality, $\rncons$ cannot occur in $\pi$. $\pi$ is also typed, so either $a = b$ or $a$ and $b$ are typed. If $a = b$, then $t = u$, and we have a proof of $Y;B \vdash \equals{\zap{t}}{\zap{u}}$ ending in $\rneq$. If $a$ and $b$ are typed, we apply Lemma~\ref{lem:projcase} and consider two cases.
        \begin{itemize}
            \item {$a$ and $b$ not zappable:} Then $\zap{a}$ and $\zap{b}$ have the same outermost structure as $a$ and $b$, and $\zap{t} = \zap{a_{0}}$ and $\zap{u} = \zap{b_{0}}$. So we can apply $\rnproj$ on $\varpi'$ to get $Y;B \eqderives \equals{\zap{t}}{\zap{u}}$. 
            \item {$a = b$:} Then $t = u$ as well, and hence $\zap{t} = \zap{u}$. Since $Y \DYderives \{\zap{t}, \zap{u}\}$, $Y;B \eqderives \equals{\zap{t}}{\zap{u}}$ with last rule $\rneq$. 
        \end{itemize}
        \item {$\pi$ ends in $\rncons$}: Let $t = \func(t_{0},t_{1})$ and $u = \func(u_{0},u_{1})$. Let $\pi$ have immediate subproofs $\pi_{0}$ and $\pi_{1}$, each $\pi_{i}$ proving $X;A \vdash \equals{t_{i}}{u_{i}}$. By IH, there are proofs $\varpi_{1}, \varpi_{2}$, each $\varpi_{i}$ proving $Y;B \vdash \equals{\zap{t_{i}}}{\zap{u_{i}}}$. By Lemma~\ref{lem:projcase}, two cases arise.
		\begin{itemize}
            \item {$t$ and $u$ not zappable:} Then $\zap{t} = \func(\zap{t_{1}}, \zap{t_{2}})$ and $\zap{u} = \func(\zap{u_{1}}, \zap{u_{2}})$. So $Y;B \eqderives \equals{\zap{t}}{\zap{u}}$ using $\rncons$ on the $\varpi_{i}$s. 
            \item {$t$ and $u$ zappable:} Then, $\zap{t} = \zap{u} = \fixedname \in Y$, so we have a proof of $Y;B \vdash \equals{\zap{t}}{\zap{u}}$ ending in $\rneq$.  \qedhere
        \end{itemize} 
	\end{itemize}
\end{proof}  

The following theorem is an immediate consequence of Lemmas~\ref{lem:zap-taunut},~\ref{lem:dyzetazap} and~\ref{lem:eqzetazap}.
\begin{restatable}{theorem}{stot}\label{thm:stot}
    Let $t, u \in \constst$ and $i \le n$. 
    \begin{itemize}[leftmargin=*]
        \item If $\intsub(T_{i-1}) \DYderives \intsub\iwsub_{i}(t)$ then $\vintsub(T_{i-1}) \DYderives \vintsub\viwsub_{i}(t)$.
        \item If $\intsub(T_{i-1}; E_{i-1}) \eqderives \intsub\iwsub_{i}(\equals{t}{u})$ then $\vintsub(T_{i-1}; E_{i-1}) \eqderives \vintsub\viwsub_{i}(\equals{t}{u})$.
    \end{itemize}
\end{restatable}
%\begin{proof} 
%    From Lemmas~\ref{lem:dyzetazap} and~\ref{lem:eqzetazap}, $\vintsub(T_{i-1}) \DYderives \zap{\intsub\iwsub_{i}(t)}$ and $\vintsub(T_{i-1}; E_{i-1}) \eqderives \zap{\intsub\iwsub_{i}(\equals{t}{u})}$. By Lemma~\ref{lem:zap-taunut}, $\zap{\intsub\iwsub_{i}(t)} = \vintsub\viwsub_{i}(t)$ and $\zap{\intsub\iwsub_{i}(\equals{t}{u})} = \vintsub\viwsub_{i}(\equals{t}{u})$, and we are done.
%\end{proof}

Having shown that the $\vlambda$s simulate the $\lambda$s, we next show that they are all bounded.
\begin{restatable}{theorem}{vlambdasmall}\label{thm:vlambda-small}
    For $\lambda \in \{\intsub, \bigsub, \iwsub_{i} \mid i \leq n\}$, $\vlambda$ is $|\stnonvars|$-bounded. 
\end{restatable}
\begin{proof}
For each $\lambda$ and any $x$, $\vbigsub(\vlambda(x)) = \vbigsub(x) = \zap{\bigsub(x)}$ (by Definition~\ref{def:vlambda}) and thus, $\dagsize{\vlambda(x)} \leq \dagsize{\vbigsub(x)}$. So it suffices to prove a bound on $\dagsize{\vbigsub(x)}$. We show that for $t \in \constst$, $\st(\vbigsub(t)) \subseteq \vbigsub(\stnonvars)$. Note that if $t = x$ is non-minimal, there is an $r \in \stnonvars$ s.t.\ $\vbigsub(t) = \vbigsub(r)$. Thus it suffices to prove the statement for $t$ which is either a minimal variable or in $\stnonvars$.

The proof is by induction on $|\vbigsub(t)|$. 
\begin{itemize}[leftmargin=*]
\item $|\vbigsub(t)| = 1:$ $\vbigsub(t) \in \names$. So $t \in \names$ or a minimal variable. If $t \in \names$, $\vbigsub(t) = t\in \names$. Otherwise, $\vbigsub(t) = \fixedname$. In both these cases, $\st(\vbigsub(t)) \subseteq \vbigsub(\stnonvars)$.
\item $|\vbigsub(t)| > 1:$ Let $a \in \st(\vbigsub(t))$. If $a = \vbigsub(u)$ for some $u \in \st(t) \setminus \varsof(t)$, then $a \in \vbigsub(\stnonvars)$. If $a = \vbigsub(x)$ for some minimal $x \in \varsof(t)$, then $a = \fixedname = \vbigsub(\fixedname) \in \vbigsub(\stnonvars)$. If $a \in \st(\vbigsub(x))$ for non-minimal $x \in \varsof(t)$, then $x \neq t$, and there is an $r \in \stnonvars$ s.t.\ $\vbigsub(x) = \vbigsub(r)$, and $a \in \st(\vbigsub(r))$. Since $|\vbigsub(r)| < |\vbigsub(t)|$, by IH, $\st(\vbigsub(r)) \subseteq \vbigsub(\stnonvars)$. Thus $a \in \vbigsub(\stnonvars)$. 
%\begin{itemize}
%\item There exists $u \in \st(t)\setminus\varsof(t)$ s.t.\ $a = \vbigsub(u)$. But then $a \in \vbigsub(\stnonvars)$. 
%\item $a \in \st(\vbigsub(x))$ for a minimal $x \in \varsof(t)$. But then $a = \fixedname = \vbigsub(\fixedname) \in \vbigsub(\stnonvars)$. 
%\item $a \in \st(\vbigsub(x))$ for non-minimal $x \in \varsof(t)$. In thus case, there is some $r \in \stnonvars$ s.t.\ $\vbigsub(x) = \vbigsub(r)$, and $a \in \st(\vbigsub(r))$. Since $|\vbigsub(r)| < |\vbigsub(t)|$, we have $\st(\vbigsub(r)) \subseteq \vbigsub(\stnonvars)$ by IH. Thus $a \in \vbigsub(\stnonvars)$, as desired. 
%\end{itemize}
\end{itemize}
Hence, $\dagsize{\vbigsub(t)} \leq |\vbigsub(\stnonvars)| \leq |\stnonvars|$, for $t \in \constst$.
\end{proof}

%For this, we follow the template of~\cite{RT03}, and build a finite sequence of pairs $(E_{0}, V_{0}), \ldots, (E_{N}, V_{N})$, such that:
%    \begin{itemize}
%        \item each $E_{i} \subseteq \constst$,
%        \item $V_{N} = \emptyset$, and
%        \item for all $i \leq N$: $\dagsize{\zap{\bigsub(t)}} \leq \dagsize{E_{i}} + {\sum_{z\in V_{i}}}~\dagsize{\zap{\bigsub(z)}}$.
%    \end{itemize}
%It follows that $\dagsize{\zap{\bigsub(t)}} \leq \dagsize{E_{N}}$, but $\st(E_{N}) \subseteq \constst$, so $\dagsize{\zap{\bigsub(t)}} \leq |\constst|$. We build the sequence starting with $(E_{0},V_{0}) \coloneqq (\{t\}, \varsof(t))$, and for all $i \geq 0$:  
%\begin{itemize}
%\item If $V_{i} = \emptyset$, stop (and let $N = i$).
%\item If there is a minimal variable in $V_{i}$, pick one such (say $y$), and take $(E_{i+1}, V_{i+1}) \coloneqq (E_{i}\cup\{\fixedname\}, V_{i}\setminus\{y\})$.
%\item If every variable in $V_{i}$ is non-minimal, pick one such (say $y$), find $r \in \stnonvars$ s.t.\ $\bigsub(y) = \bigsub(r)$, and take $(E_{i+1},V_{i+1}) \coloneqq (E_{i}\cup\{r\}, V_{i}\setminus\{y\}\cup \varsof(r))$.
%\end{itemize}
%
%Essentially, whenever a non-minimal variable $y$ occurs in a $V_{i}$, we represent it by a term $r$ from $\stnonvars$ that it unifies with, and replace $y$ by other variables unifying with strict subterms of $r$. One can see that this sequence satisfies the properties listed earlier, thus leading to the following theorem, whose full proof is worked out in Appendix~\ref{app:insecurity}. 

\subsection{NP algorithm for Insecurity: Sketch}
After guessing a coherent set of sessions and an interleaving of these sessions of length $n$, we guess bounded substitutions $\vintsub$ and $\viwsub_{1}, \ldots, \viwsub_{n}$, as well as a sequence of knowledge functions such that the relevant equalities and terms (communicated in the $\vintsub(\beta_{i})$s) are derivable from $\vintsub(\dc(\knowfunc_{i-1}(I)))$. All these derivability checks can be carried out in time polynomial in the size of the protocol description. 
%There are a few checks of the form $\posof{x}{t} \subseteq \abstractable(T,t)$, which can also be reduced to checking if the subterms of $t$ at appropriate positions are derivable.

For the honest agent sends, we only require derivabilities of the form $\knowfunc_{i}(u_{i}) \assderives \alpha_{i}$, without applying any substitution. This is the derivability problem, which we proved to be solvable in NP. Finally, we check that $\vintsub(\knowfunc_{n}(I)) \assderives \vintsub(\gamma)$, which can also be solved in NP. Thus the $K$-bounded insecurity problem for assertions is in NP.

%We thus obtain an NP algorithm for $K$-bounded insecurity with assertions, as outlined at the beginning of this section. 

\section{Applications and Extended Syntax}\label{sec:fullsyntax}
The Dolev-Yao model with formulas can be used for modelling real-world protocols, as in~\cite{MPR13} and \cite{RSS17}. We now introduce a variant of the model from~\cite{RSS17}, and show how to encode the well-known FOO e-voting protocol~\cite{FOO92} in that system. Finally, we discuss how our results can be extended to this system.

\subsection{Modelling the FOO protocol {\`a} la \texorpdfstring{\cite{RSS17}}{RSS17}}
In~\cite{RSS17}, the authors allow principals to communicate a richer class of assertions over an extended syntax. The syntax considered there  includes atomic predicates, equality, conjunction, existentially quantified assertions, disjunction, and a$\says$connective. We will consider the same syntax, but instead of full disjunction, we will consider list membership (denoted by $\listmemb$), as this suffices for most examples. In the following, $t,u \in \Tterms$, $P$ is an $m$-ary predicate, $u_{1}, \ldots, u_{m}, t_{0} \in \names \cup \vars$, and $t_{1}, \ldots, t_{n} \in \names$, $x \in \qvar$, and $\pk_{A}$ denotes the public key of the agent $A \in \ag$. 
{
  \begin{align*}
      \alpha &:= \equals{t}{u} \mid P(u_{1}, \ldots, u_{m}) \mid t_{0} \listmemb [t_{1}, \ldots, t_{n}] \\
      & \hspace{7mm} \mid \alpha_{0} \conj \alpha_{1} \mid \exists x.~\alpha(x) \mid \pk_{A} \says \alpha
  \end{align*}
}
The extra rules required for this syntax are shown in Table~\ref{tab:fullasstheory}. 

One can model interesting protocols in this language, including the FOO e-voting protocol~\cite{FOO92, KR05}. In this protocol, there is a voter $V$, an authority $A$ who verifies the eligibility of voters but should not know their votes, and a collector $C$ who should not know voters' identities but counts all votes. 

To model this using only terms, a new operator called \emph{blinding} is used. Signing a blinded object allows the signature to percolate through to the object inside the blind. Formally, one can use $t$ and $b$ to make a blind pair $\blind(t, b)$, and get $\sign(t, k)$ from $\sign(\blind(t, b), k) $ and $b$. The voter uses the blinding operation to hide their vote from the authority but still get it certified as coming from an authorized voter (as identified by their signing key $\sk_{V}$). The authority's signature $\sign(\cdot, \sk_{A})$ percolates through to the vote when the voter removes the blind, and the voter can then anonymously send (denoted by $\looparrowright$) this signed vote to the collector for inclusion into the final tally. This specification is shown below.
\begin{align*}
    V \rightarrow A &: \sign(\blind(\{v\}_{r}, b), \sk_{V}) \\
    A \rightarrow V &: \sign(\blind(\{v\}_{r}, b), \sk_{A}) \\
    V \looparrowright C &: \sign(\{v\}_{r}, \sk_{A}) 
\end{align*}

%The voter blinds their vote, signs it to indicate origin, and sends it to the authority. The authority verifies that the signatory is eligible to vote, and as a certificate of this eligibility, signs the blinded vote and sends it back to the voter. The voter can now unblind this term, and anonymously send (indicated by $\looparrowright$) it to the collector for inclusion into the final tally. Informally, thus, the authority does not get to know the voter's vote, and the collector does not get to know the voter's identity. The voting phase of the protocol is modelled as follows. $k_{A}$ and $k_{V}$ are the signing keys of $A$ and $V$ respectively, $v$ is the chosen vote, and $r$ and $b$ are freshly picked values for the encrypting key and blinding factor respectively.

One can use assertions to model the voting phase of FOO as below, following~\cite{RSS17} (We use $\{ \alpha \}^{A}$ to serve as shorthand for $\pk_{A}\says\alpha$). In fact, the use of assertions allows one to also specify an eligibility check for voters via a conditional action $\assertact$, which allows the protocol to proceed only if the specified assertion is derivable by the agent in question. Further, using list membership, voters can also include a certificate that their vote is for an allowable candidate from the list $\ell$. These are left implicit in the terms-only modelling. 
\begin{center}
{\small
\vspace{-1.8em}
\begin{align*}
        V \rightarrow A &: \{v\}_{p}, \bigl\{\exists xr.\equals{\{x\}_{r}}{\{v\}_{p}} \conj x \listmemb \ell \bigr\}^{V} \\
        A &: \assertact~\elg(V) \\
        A \rightarrow V &: \bigl\{ \elg(V) \conj \bigl\{\exists xr.\equals{\{x\}_{r}}{\{v\}_{p}} \conj x \listmemb \ell \bigr\}^{V} \bigr\}^{A} \\
        V \looparrowright C &: \{v\}_{q}, \exists Uys.\bigl\{ \elg(U) \conj \bigl\{\exists xr.\equals{\{x\}_{r}}{\{y\}_{s}} \conj x \listmemb \ell \bigr\}^{U} \bigr\}^{A} \\
        & \hspace{5mm} \conj \bigl\{\exists w.\equals{\{y\}_{w}}{\{v\}_{q}}\bigr\}
\end{align*}
}
\end{center}

$V$ first sends to $A$ their encrypted vote along with an assertion claiming that it is for a valid candidate from the list $\ell$. The authority checks the voter's eligibility via the $\assertact$ action on the $\elg$ predicate. If the check passes, the authority issues a certificate stating that the voter is allowed to vote, crucially, without modifying the term containing the vote. $V$ then existentially quantifies out their name from this certificate, and anonymously sends to $C$ a re-encryption of the vote authorized by $A$ along with a certificate to that effect. Thus, the intent behind the various communications is made more transparent than in the model with blind signatures. One can show that this satisfies anonymity~\cite{RSS17}.

We can also specify security properties in a more natural manner (as compared to in the terms-only model). For instance, we say that \emph{vote secrecy} is ensured in the above protocol if there is no run where the intruder can derive $\exists{xy}:[\{v\}_{p} = \{x\}_{y} \wedge x = v]$. Note that this means that while anyone can derive the value of $v$, which is public, they should not be able to identify the value inside the encrypted vote $\{v\}_{p}$ as being a particular public name. To express this in the terms-only formulation, one has to check whether two runs that only differ in the vote $v$ can be distinguished by the intruder~\cite{CK14}. It can be seen from~\cite{RSS17} that the proving such properties might involve considering multiple runs simultaneously, but the specification itself does not refer to a notion of equivalence. 

\begin{example}\label{ex:disje}
Consider a protocol where $V$ sends to $A$ the vote encrypted in a fresh key, and an assertion that the vote belongs to an allowable list $\ell$ of candidates. This looks as follows. $V \rightarrow A : \{v\}_{k}, \exists xr.\bigl\{ \equals{\{x\}_{r}}{\{v\}_{k}} \conj x \listmemb \ell \bigr\}$.

Suppose this same protocol is used for two elections that $V$ participates in simultaneously, where the first election has candidates $0$ and $1$ (so $\ell_{1} = [0, 1]$) and the second has candidates $0$ and $2$ (so $\ell_{2} = [0, 2]$). 

$V$ wants to vote for $0$ in both elections. Since the vote is for the same candidate, $V$ (unwisely) decides to reuse the same term, instead of re-encrypting in a fresh key. So we have a run where $V$ sends both $\exists xr.\bigl\{ \equals{\{x\}_{r}}{\{v\}_{k}} \conj x \listmemb [0, 1] \bigr\}$ and $\exists ys.\bigl\{ \equals{\{y\}_{s}}{\{v\}_{k}} \conj y \listmemb [0, 2] \bigr\}$. Now, since the same term $\{v\}_{k}$ is involved in both assertions, an observer ought to be able to deduce that the vote is actually for $0$. 

Let $S$ be the set $\{ \{v\}_{k}, 0, 1, 2 \}$ and let $A$ consist of the above two assertions sent out by $V$. $\dc(S; A) = (T; E)$ where $T = \{ \{v\}_{k}, 0, 1, 2, x, r, y, s \}$ and $E = \bigl\{ \equals{\{x\}_{r}}{\{v\}_{k}}, x \listmemb [0, 1], \equals{\{y\}_{s}}{\{v\}_{k}}, y \listmemb [0, 2] \bigr\}$. We present a proof of $T; E \vdash \exists zw.\bigl\{\equals{\{z\}_{w}}{\{v\}_{k}} \conj \equals{z}{0}\bigr\}$ in Figure~\ref{fig:exdisje}. We omit the LHS as well as some $\DYderives$ proofs for readability.
\end{example}

\begin{figure*}[!t]
	\centering {\footnotesize
    %   \setlength{\tabcolsep}{2.1em}
	% \begin{tabu}{l|r}
		\begin{math}
              \begin{prooftree}
              \[
                \[
                	\justifies \equals{\{x\}_{r}}{\{v\}_{k}} \using \rnax
                \]
                \[
                	\[
              		\[
              			\justifies x \listmemb [0, 1] \using \rnax
              		\]
              		\[
                        \[ 
              		        \justifies y \listmemb [0, 2] \using \rnax
                        \]
                        \[
              	            \[
              		            \[
              				        \justifies \equals{\{y\}_{s}}{\{v\}_{k}} \using \rnax
              		            \]
              		            \[
						            \[
							            \justifies \equals{\{x\}_{r}}{\{v\}_{k}} \using \rnax
						            \]
              				        \justifies \equals{\{v\}_{k}}{\{x\}_{r}} \using \rnsymm
              		            \]
              		            \justifies \equals{\{y\}_{s}}{\{x\}_{r}} \using \rntrans
              	            \]
              	            \justifies \equals{y}{x} \using \rnproj
                        \]
                        \justifies x \listmemb [0, 2] \using \rnsubst
              		\]
              		\justifies x \listmemb [0] \using \rnlint
              	\]
                	\justifies \equals{x}{0} \using \rnprom 
                \]
                 \justifies \equals{\{x\}_{r}}{\{v\}_{k}} \conj \equals{x}{0} \using \rnconji
              \]
              \justifies \exists zw.\bigl\{\equals{\{z\}_{w}}{\{v\}_{k}} \conj \equals{z}{0}\bigr\} \using \rnexintro^{2}
              \end{prooftree}
           \end{math}
% 		&
% 		\begin{math}
%               \begin{prooftree}
%               \[ 
%               		\justifies y \listmemb [0, 2] \using \rnax
%               \]
%               \[
%               	\[
%               		\[
%               				\justifies \equals{\{y\}_{s}}{\{v\}_{k}} \using \rnax
%               		\]
%               		\[
% 						\[
% 							\justifies \equals{\{x\}_{r}}{\{v\}_{k}} \using \rnax
% 						\]
%               				\justifies \equals{\{v\}_{k}}{\{x\}_{r}} \using \rnsymm
%               		\]
%               		\justifies \equals{\{y\}_{s}}{\{x\}_{r}} \using \rntrans
%               	\]
%               	\justifies \equals{y}{x} \using \rnproj
%               \]
%               \justifies x \listmemb [0, 2] \using \rnsubst
%               \end{prooftree}
%           \end{math}
%          \end{tabu}
 }
% \caption{Proofs for Example~\ref{ex:disje}. The proof on the right is $\pi_{1}$. $\rnexintro^{2}$ indicates two applications of $\rnexintro$.} 
\caption{Proof for Example~\ref{ex:disje}. $\rnexintro^{2}$ indicates two applications of $\rnexintro$.} 
\label{fig:exdisje}
\end{figure*}

\subsection{Adapting the results of Section~\ref{sec:derivability} and Section~\ref{sec:insecurity}}
%We now summarise how to adapt our definitions and results to the extended syntax. Details can be found in Appendix~~\ref{app:fullsyntax}.
We can extend the definition of $\positions{\alpha}$ in a straightforward manner. 
%as follows.
%	\begin{itemize} 
%	%[leftmargin=*,labelsep=1pt]
%        \item $\positions{\equals{t}{t'}} = \{0\cdot p \mid p \in \positions{t}\} \cup \{1\cdot p \mid p \in \positions{t'}\}$
%        \item $\positions{P(u_{1}, \ldots, u_{m})} = \{1, \ldots, m\}$
%        \item $\positions{t \listmemb [t_{1}, \ldots, t_{n}]} = \{0, 1, \ldots, n\}$
%        \item $\positions{\alpha \conj \beta} = \{0\cdot p \mid p \in \positions{\alpha}\} \cup \{1\cdot p \mid p \in \positions{\beta}\}$
%	    \item $\positions{\exists{x}.\alpha} = \{0\cdot p \mid p \in \positions{\alpha}\}$				
%		\item $\positions{\pk(k) \says \alpha} = \{0, 00\} \cup \{1\cdot p \mid p \in \positions{\alpha}\}$
%	\end{itemize}    
%
The definition of $\abstractable(S,\alpha)$ is given as follows.
	\begin{itemize}
	%[leftmargin=*,labelsep=1pt]
		\item $\abstractable(S, \equals{t_{0}}{t_{1}}) = \{i\cdot p \mid i \in \{0, 1\}, \ p \in \abstractable(S,t_{i})\}$
		\item $\abstractable(S, P(u_{1}, \ldots, u_{m})) = \{i \mid 1 \leq i \leq m, S \DYderives u_{i}\}$
		\item $\abstractable(S, t \listmemb [t_{1}, \ldots, t_{n}]) = \{0\}$
		\item $\abstractable(S, \alpha_{0}\conj\alpha_{1}) = \{i\cdot p \mid i \in \{0, 1\}, \ p \in \abstractable(S,\alpha_{i})\} $
		\item $\abstractable(S, \exists{x}.\alpha) = \{0\cdot p \mid p \in \abstractable(S \cup \{x\},\alpha) \}$
		\item $\abstractable(S, \pk_{a} \says \alpha) =  \{0\} \cup \{1\cdot p \mid p \in \abstractable(S, \alpha)\}$
	\end{itemize}

%To the $\assderives$ system, we now add rules for introducing and eliminating conjunctions, rules pertaining to $t\listmemb{l}$, and a rule to derive $\pk_{a}\says\alpha$. These are presented in Table~\ref{tab:fullasstheory}. 

We now use $\eqderives$ to mean all the rules in Tables~\ref{tab:asstheory} and \ref{tab:fullasstheory}, except $\rnconji, \rnconje, \rnexintro,$ and $\rnexelim$.

\begin{table}[!t]
	\centering {\small
        \tabulinesep=2mm
        \setlength{\tabcolsep}{0.6em}
        \begin{tabu}{|c|c|}
            \hline
			\begin{math}
                {
                    \begin{prooftree}
                        S;A \vdash t \listmemb [n]
                        \justifies S;A \vdash \equals{t}{n} \using \rnprom
					\end{prooftree}
                }
			\end{math}
			&
			\begin{math}
                {
                    \begin{prooftree}
                        S;A \vdash \equals{t}{n_{i}}
                        \justifies S;A \vdash t \listmemb [n_{1}, \ldots, n_{k}] \using \rnlwk
					\end{prooftree}
                }
			\end{math}
			\\
			\hline
			\multicolumn{2}{|c|}{
			\begin{math}
                {
                    \begin{prooftree}
                        S;A \vdash t \listmemb l_{1} \quad \cdots \quad S;A \vdash t \listmemb l_{m}
                        \justifies S;A \vdash t \listmemb (l_{1} \cap \ldots \cap l_{m}) \using \rnlint
					\end{prooftree}
                }
			\end{math}	
			}
			\\
			\hline 
			\begin{math}
                {
                    \begin{prooftree}
                        S;A \vdash t\listmemb\ell \quad S;A \vdash \equals{t}{u}
                        \justifies S;A \vdash u\listmemb\ell \using \rnsubst
					\end{prooftree}
                }
			\end{math}
			&
			\begin{math}
                {
                    \begin{prooftree}
                        S;A \vdash \alpha \quad S \DYderives \sk_{a}
                        \justifies S;A \vdash \pk_{a} \says \alpha \using \rnsays
					\end{prooftree}
                }
			\end{math}
			\\
			\hline
			\begin{math}
                {
					\begin{prooftree}
						S;A \vdash \alpha_{0} \quad S;A \vdash \alpha_{1}
						\justifies S;A \vdash \alpha_{0} \conj \alpha_{1} \using \rnconji
					\end{prooftree}
                }
			\end{math}
			 &
			\begin{math}
                {
					\begin{prooftree}
						S;A \vdash \alpha_{0} \conj \alpha_{1}
						\justifies S;A \vdash \alpha_{i} \using \rnconje_{i}
					\end{prooftree}
                }
			\end{math}
			\\
			\hline
		\end{tabu}
	}
	\caption{
        Extra rules for $\conj$, \textit{says} and lists. For the $\rnsays$ rule, $(\pk_{a},\sk_{a})$ is the public-private key pair of $a \in \ag$.}
	\label{tab:fullasstheory}
\end{table}

The notion of kernel and the conditions in Theorem~\ref{thm:eqs-to-gamma} also change. Earlier, an assertion had the form $\exists{x_{1}\dots{}x_{k}}.\equals{t}{u}$, and we reduced all reasoning to assertions of the form $\equals{t}{u}$, which were maximal subformulas without a logical connective for the outermost operator. We can think of those as the ``atoms'' of a formula, and we use this idea to adapt the definition of $\dc$. $(T; E) = \dc(S; A)$ iff: 
\begin{itemize}
    \item $T = S \cup \boundvars(A)$
    \item $A = \{\alpha \in \SF(A) \mid \alpha$ is of the form $\equals{t}{u}$ or $t\listmemb{l}$ or $\pk(k)\says\beta$ or $P(u_{1}, \ldots, u_{m})\}$.
\end{itemize}

For Theorem~\ref{thm:eqs-to-gamma}, one would expect that deriving $\alpha$ reduces to deriving substitution instances of all atoms of $\alpha$. But consider subformulas of the form $\pk_{a}\says\beta$. We can derive those in two ways -- either by using $\rnax$ (if the formula is already in the LHS) or by using the $\rnsays$ rule on $\beta$ and $\sk_{a}$. In the latter case, one would look to derive the atoms of $\beta$. We thus formalize the atoms of a formula as below.
\[
\atomsof(\gamma) = \begin{cases}
\atomsof(\alpha) \cup \atomsof(\beta) & \text{if $\gamma = \alpha\conj\beta$} \\
\atomsof(\alpha) & \text{if $\gamma = \exists{x}.\alpha$} \\
\left\{\pk(k)\says\alpha\right\} \cup \atomsof(\alpha) & \text{if $\gamma = \pk(k)\says\alpha$} \\
\{\gamma\} & \text{otherwise}
\end{cases}
\]

%\begin{itemize}
%    \item $\atomsof(\equals{t}{u}) = \{\equals{t}{u}\}$ 
%    \item $\atomsof(P(u_{1}, \ldots, u_{m})) = \{P(u_{1}, \ldots, u_{m})\}$ 
%    \item $\atomsof(t\listmemb{l}) = \{t\listmemb{l}\}$ 
%    \item $\atomsof(\alpha\conj\beta) = \atomsof(\alpha) \cup \atomsof(\beta)$ 
%    \item $\atomsof(\exists{x}.\alpha) = \atomsof(\alpha)$ 
%    \item $\atomsof(\pk(k)\says\alpha) = \{\pk(k)\says\alpha\} \cup \atomsof(\alpha)$ 
%\end{itemize}

Theorem~\ref{thm:eqs-to-gamma} is modified as follows.
%, to account for the extended syntax and the choice in deriving \!\!$\says$\!\! assertions mentioned above. 
\begin{theorem}
For a formula $\alpha$ s.t.\ $\boundvars(\alpha) \cap \varsof(S;A) = \emptyset$, and $(T;E) = \dc(S;A)$, $(S;A) \assderives \alpha$ iff there is a $\mu$ with $\dom(\mu) = \boundvars(\alpha)$ and $X \subseteq \atomsof(\alpha)$ s.t.: 
\begin{enumerate}[label={[\arabic*]}]
    \item $\forall{}x \in \boundvars(\alpha): T \DYderives \mu(x)$.
    \item $\forall{}x \in \boundvars(\alpha), t \in \st(\alpha)$: $\posof{x}{t} \subseteq \abstractable(T \cup \boundvars(\alpha), t)$.
    \item For all $\beta \in X$, $(T;E) \eqderives \mu(\beta)$.  
    \item $(T;\mu(X)) \assderives \alpha$ via a proof with only introduction rules. 
\end{enumerate} 
\end{theorem}

We can extend the normalization and subterm properties appropriately (details provided in Appendix~\ref{app:normalization}). 

Very few changes are required to adapt the results of Section~\ref{sec:insecurity} to this new system to show that insecurity for this extended system continues to be in NP. In fact, one can obtain bounded substitutions for the passive and active intruder problems exactly as for equality assertions. Only atomic terms occur in assertions of the form $P(\cdots)$ and $t\listmemb{l}$, thus no variable standing for a term occurring in these assertions is zappable. Similarly, for assertions of the form $\pk_{a}\says\alpha$, $\pk_{a}$ is not zappable. However, there might be terms inside $\alpha$ which are zappable, but these will be at the level of equality subformulas of $\alpha$, and the \textit{says} connective itself does not influence the zapping procedure in any way.

So the extended syntax interferes very little with our proof strategy for finding bounded substitutions. Thus, we have a practical language that can be used to model protocols and properties, and which also enjoys nice decidability properties. 

A final note on atomic predicates of the form $P(\cdots)$: There are no rules governing these, other than $\rnax$. We can assume that depending on the context, ground assertions of this kind are added and removed from agents' knowledge states. For example, in a voting protocol, one can add to the authority's initial state all predicates of the form $\elg(V)$ for eligible voters $V$. As and when the authority receives a message from $V$ and checks that $\elg(V)$ is true, it removes $\elg(V)$ from its knowledge state, recording the fact that $V$ has voted and is no longer eligible. Our decidability proofs are not affected by these modifications.

\section{Future Work}\label{sec:disc}
An interesting feature of the language in~\cite{RSS17} is the use of disjunction. While our extended syntax uses list membership to express a limited form of disjunction that seems to suffice for many protocols, it would be worthwhile to explore the utility of full disjunction and its effect on the active intruder problem. 

In fact, with disjunction, we know that even the derivability problem becomes more involved. To check if $(S;A) \assderives \gamma$, one can no longer work with a single kernel of $(S;A)$. We can define a notion of ``down-closure''. For each disjunctive formula $\alpha\disj\beta$, one obtains two down-closures -- one containing $\alpha$, and the other $\beta$. In general, many disjunctions could occur in $A$ and there are exponentially many down-closures for any $(S; A)$. Using the standard left disjunction property ($\alpha\disj\beta$ derives $\gamma$ iff $\gamma$ is derivable from $\alpha$ and from $\beta$), we check if the kernels of all down-closures of $(S;A)$ derive $\gamma$. Thus the derivability problem is in $\Pi_{2}$. Some of these down-closures need not even be consistent, and hence our solution for the insecurity problem is not directly adaptable to full disjunction. Exploring these issues is an interesting direction of research and is left for future work. 

It is also useful to add communicable assertions to the widely-used applied pi calculus~\cite{ABF17}. It would be especially interesting to see how this impacts the notion of static equivalence, and then study expressibility and decidability. These would also help us to extend tools like Proverif~\cite{Bla16} with assertions.

As mentioned earlier, one can express certain ``equivalence'' properties in a more natural manner with assertions as compared to the terms-only model. It is another promising extension to study which equivalence properties can be expressed as reachability properties in this manner, like the work in~\cite{GMV22}.

% \newpage
\bibliographystyle{plain}
% \bibliography{ref}

\newpage
\appendix
\section{Proofs for Sections~\ref{sec:dolevyao} and~\ref{sec:derivability}}
\label{app:derivability}

\absderiv*
\begin{proof}
	For any term $a$ and any set $Q \subseteq \positions{a}$, we let $\subtermat{a}{Q}$ denote $\{\subtermat{a}{q} \mid q \in Q\}$. We now observe some general properties of abstractability. 
	
	For any $T, a$ and $q \in \abstractable(T,a)$ s.t.\ $\subtermat{a}{q}$ is non-atomic, either $\{q0, q1\} \subseteq \abstractable(T,a)$ and $\subtermat{a}{\{q0,q1\}} \DYderives \subtermat{a}{q}$ via a constructor rule, or $q$ is a maximal position in $\abstractable(T,a)$ (it is not the prefix of any other position in the set). We have the following two properties.
\begin{enumerate}
\item Let $M = \{q \in \positions{a} \mid q$ is a maximal position in $\abstractable(T,a)\}$. Then for every $p \in \abstractable(T,a)$, $\subtermat{a}{M} \DYderives \subtermat{a}{p}$ via a proof consisting only of constructor rules. 
\item\label{item:prefsib} Suppose $Q \subseteq \positions{a}$ is prefix-closed (if $q \in Q$ and $p$ is a prefix of $q$, then $p \in Q$) and sibling-closed (if $qi \in Q$ and $qj \in \positions{a}$, then $qj \in Q$). If $T \DYderives \subtermat{a}{q}$ for every maximal $q \in Q$, then $Q \subseteq \abstractable(T,a)$. 
\end{enumerate}
		 
    We now prove the statement of the lemma. Let $u = \replsubtermat{t}{P}{r}$, and let $A$ and $B$ denote $\abstractable(S\cup\{x\},t)$ and $\abstractable(S,u) \cap \positions{t}$ respectively. Note that $A$ and $B$ are both prefix-closed and sibling-closed. Let $M$ (resp.\ $N$) be the set of maximal positions in $A$ (resp.\ $B$). 
    
    Since $P \subseteq A$ is the set of $x$-positions in $t$, $P \subseteq M$ and no $q \in M$ is a prefix of a position in $P$. Thus, for every $q \in M$, either $\subtermat{t}{q} = x$, or $x \notin \varsof(\subtermat{t}{q})$. If $\subtermat{t}{q} = x$, $\subtermat{u}{q} = r$, and $S \DYderives \subtermat{u}{q}$ (since $S \DYderives r$). If $x \notin \varsof(\subtermat{t}{q})$, then $\subtermat{u}{q} = \subtermat{t}{q}$ and $S \DYderives \subtermat{u}{q}$. This is because $q \in \abstractable(S\cup\{x\},t)$, so $S \cup \{x\} \DYderives \subtermat{t}{q}$, but $x$ does not occur in the conclusion. Thus we have $S \DYderives \subtermat{u}{q}$ for every $q \in M$. Since $A$ is prefix-closed and sibling-closed, by~\ref{item:prefsib}, we get $A \subseteq \abstractable(S,u)$. Since $A \subseteq \positions{t}$ as well, we get $A \subseteq B$. 
    
	By similar reasoning as above, we can see that $S \cup \{x\} \DYderives \subtermat{t}{q}$ for each $q \in N$. (For some of these positions $q$, $x$ does not occur at all in the subterm at that position, and $\subtermat{t}{q} = \subtermat{u}{q}$ is derivable from $S$. For other positions $q$, $\subtermat{t}{q} = x$ and is derivable from $S \cup \{x\}$.) Therefore $B \subseteq A$. 
\end{proof}

\begin{lemma}
The $\rnsubst$ rule is admissible in $\assderives$.
\end{lemma}

\begin{proof}
    We first show that $\rnsubst$ can be simulated by a series of applications of $\rnsubst_{1}$, defined below.
    \[
        \begin{prooftree}
            S; A \vdash \replsubtermat{(\equals{t}{u})}{p}{r} \quad S; A \vdash \equals{r}{s}
            \justifies S; A \vdash \replsubtermat{(\equals{t}{u})}{p}{s} \using \rnsubst_{1}
        \end{prooftree}
    \]
    The rule is enabled only if $p \in \posof{x}{\equals{t}{u}} \cap \abstractable(S \cup \{x\}, \equals{t}{u})$ and $S \DYderives \{r,s\}$. $\rnsubst_{1}$ replaces the $r$ \defemph{occurring at $p$} by $s$. 
    
    Let $\alpha$ denote $\equals{t}{u}$ and suppose we have an instance of $\rnsubst$ with $(T;E) \vdash \replsubtermat{\alpha}{P}{r}$ and $(T;E) \vdash \equals{r}{s}$ as premises and $(T;E) \vdash \replsubtermat{\alpha}{P}{s}$ as conclusion, with $P = \{p_{1}, \ldots, p_{\ell}\} \subseteq \posof{x}{\alpha} \cap \abstractable(T \cup \{x\}, \alpha)$, and $T \DYderives \{r,s\}$. The $p_{i}$s are $x$-positions, so none of them is a prefix of another. Therefore, even after replacing the $x$s occurring in the set of positions $P \setminus \{p_{i}\}$ with some terms, $p_{i}$ remains an $x$-position. 
    
    For $0 \leq i \leq \ell$, we define the following. $P_{i} \coloneqq \{p_{1}, \ldots, p_{i}\}$ and $Q_{i} \coloneqq P \setminus P_{i} = \{p_{i+1}, \ldots, p_{\ell}\}$ and $\alpha_{i} \coloneqq  \replsubtermat{(\replsubtermat{\alpha}{Q_{i}}{r})}{P_{i}}{s}$. Note that $\alpha_{0} = \replsubtermat{\alpha}{P}{r}$ and $\alpha_{\ell} = \replsubtermat{\alpha}{P}{s}$. We see that $p_{i}$ is an $x$-position of $\beta_{i} = \replsubtermat{(\replsubtermat{\alpha}{Q_{i}}{r})}{P_{i-1}}{s}$. By the equivalent of Lemma~\ref{lem:absderiv} for assertions, $p_{i} \in \abstractable(T \cup \{x\}, \beta_{i})$, since $P \subseteq \abstractable(T \cup \{x\}, \alpha)$,
    We also see that $\alpha_{i-1} = \replsubtermat{\beta_{i}}{p_{i}}{r}$ and $\alpha_{i} = \replsubtermat{\beta_{i}}{p_{i}}{s}$. Thus we can get from $\alpha_{i-1}$ to $\alpha_{i}$ using the $\rnsubst_{1}$ rule, and from $\alpha_{0}$ to $\alpha_{\ell}$ using a series of $\rnsubst_{1}$ rules.   
    
    Now we show that $\rnsubst_{1}$ can be simulated in the $\eqderives$ system. Suppose $\pi$ is a proof of $(T;E) \eqderives \equals{r}{s}$, and let $T \DYderives \{r,s\}$. For all $p$, and for all $t, u$ s.t.\ $p \in \posof{x}{\equals{t}{u}} \cap \abstractable(T \cup \{x\}, \equals{t}{u})$, we show that if $(T;E) \eqderives \replsubtermat{(\equals{t}{u})}{p}{r}$, then $(T; E) \eqderives \replsubtermat{(\equals{t}{u})}{p}{s}$. The proof proceeds by induction on the length of $p$.
    \begin{itemize}
        \item $p = 0$: We have a proof of $\equals{r}{u}$. By $\rnsymm$, we get a proof of $\equals{u}{r}$. Combining this with $\pi$ using $\rntrans$, we get a proof of $\equals{u}{s}$, from which we can get a proof of $\equals{s}{u}$ by applying $\rnsymm$ again. 
        
        \item $p = 1$: We have a proof of $\equals{t}{r}$. Combining this with $\pi$ using $\rntrans$, we get a proof of $\equals{t}{s}$, as desired.
        
        \item $p = 0q$ for some $q \neq \epsilon$: Note that $\replsubtermat{(\equals{t}{u})}{p}{r}$ is the same as $\equals{\replsubtermat{t}{q}{r}}{u}$. Suppose $q = 1q'$, w.l.o.g. Then, $t = \func(t_{0},t_{1})$ and $\replsubtermat{t}{q}{r} = \func(t_{0},\replsubtermat{t_{1}}{q'}{r})$. Since $q \in \abstractable(T \cup \{x\}, t)$, we have that $T \DYderives \{t_{0},t_{1}\}$, and by $\rneq$, $(T; E) \eqderives \equals{t_{j}}{t_{j}}$ for $j \in \{0,1\}$. From $T \DYderives r$, and $q' \in \abstractable(T \cup \{x\}, t_{1})$, we have that $T \DYderives \replsubtermat{t_{1}}{q'}{r}$. Thus, we have a proof of $(T; E) \eqderives \equals{\replsubtermat{t_{1}}{q'}{r}}{\replsubtermat{t_{1}}{q'}{r}}$ using $\rneq$. Applying IH to the position $0q'$, we have $(T; E) \eqderives \equals{\replsubtermat{t_{1}}{q'}{s}}{\replsubtermat{t_{1}}{q'}{r}}$. By applying $\rncons$ to this and $(T;E) \eqderives \equals{t_{0}}{t_{0}}$, we get $(T; E) \eqderives \equals{\replsubtermat{t}{q}{s}}{\replsubtermat{t}{q}{r}}$. Applying $\rntrans$ to this and $\equals{\replsubtermat{t}{q}{r}}{u}$, we get a proof of $\equals{\replsubtermat{t}{q}{s}}{u}$, i.e. $\replsubtermat{(\equals{t}{u})}{p}{s}$.
        
        \item $p = 1q$ for some $q \neq \epsilon$: This is similar to the above. \qedhere
    \end{itemize}
\end{proof}

\leftex*
\begin{proof}
    For the left to right direction, let $\pi$ be a proof of $S; A, \exists{x}.\alpha \vdash \gamma$. Note that we have a proof $\pi_{1}$ of $\exists{x}.\alpha$ from $(S, x; A, \alpha)$, where the $\rnexintro$ rule is justified because the abstractability side condition $\posof{x}{\alpha} \subseteq \abstractable(S \cup \{x\}, \alpha)$ is assumed. We can then use the $\rncut$ rule (which is admissible in $\assderives$) on this proof along with the proof $\pi$ to get $(S, x; A, \alpha) \assderives \gamma$. 
    \[
        \begin{prooftree}
        \[
            \[
                \justifies S,x; A,\alpha \vdash \alpha \using \rnax
            \]
            \justifies S,x; A,\alpha \vdash \exists{x}.\alpha \using \rnexintro
        \]
        \[
				\pi \leadsto
        	\justifies S; A, \exists{x}.\alpha \vdash \gamma
        \]
        \justifies S, x; A, \alpha \vdash \gamma \using \rncut
        \end{prooftree}
    \]   

    For the other direction, let $\pi$ be a proof of $S, x; A,\alpha \vdash \gamma$. We obtain a proof of $S; A,\exists{x}.\alpha \vdash \gamma$ as follows.
    \[
        \begin{prooftree}
            \[
                \justifies S; A,\exists{x}.\alpha \vdash \exists{x}.\alpha \using \rnax
            \]
            \[
                \pi \leadsto S,x; A,\alpha \vdash \gamma
            \]
            \justifies S; A,\exists{x}.\alpha \vdash \gamma \using \rnexelim
        \end{prooftree}
    \]
\end{proof}

 \dcpure* 
\begin{proof}
Suppose $(T;E) = \dc(S;A)$, for sanitized $(S;A)$. So $\freevars(S;A) \cap \qvar = \emptyset$, $\publics(\beta) \in S$ for all $\beta \in A$, and $T = S \cup \boundvars(A)$ and $E = \{\equals{t}{u} \mid \exists{\vec{x}}.\equals{t}{u} \in A\}$. Thus $\publics(\gamma) \in T$ for every $\gamma \in E$, and $\varsof(E) \cap \qvar \subseteq T$. 

Let $\pi$ be a proof of $(T;E) \assderives \alpha$. Note that $\pi$ has no occurrence of $\rnexelim$. We assume that all premises of $\rneq$ are normal $\DYderives$ proofs ending in a destructor (by repeatedly turning all $\text{constructor}+\rneq$ patterns into $\rneq+\rncons$). We show by induction that $T \DYderives \publics(\alpha)$. Let $\rnrule$ denote the last rule of $\pi$.
\begin{itemize}[leftmargin=*]

\item $\rnrule = \rnax$: $\alpha \in E$. So $T \DYderives \publics(\alpha)$, by purity.

\item $\rnrule = \rneq$: $\alpha$ is $\equals{t}{t}$ with $T \DYderives t$ via a proof ending in destructor. Since any term in $T$ is either in $\qvar$ or contains no variables from $\qvar$, and since $t \in \st(T)$, we see that $\publics(\alpha)$ is $\{t\}$ or $\emptyset$, and $T \DYderives \publics(\alpha)$ in both cases.

\item $\rnrule \in \{\rnsymm, \rntrans\}$: Any $t \in \publics(\alpha)$ is in $\publics(\beta)$ for one of the premises $\beta$, and the result follows.

\item $\rnrule = \rncons$: $\alpha$ is of the form $\equals{t}{u}$, where $t = \func(t_{0},t_{1})$ and $u = \func(u_{0},u_{1})$, and the immediate subproofs of $\pi$ derive $\equals{t_{0}}{u_{0}}$ and $\equals{t_{1}}{u_{1}}$. Now, any term in $\publics(\alpha)$ is a public term of one of the premises (and we can apply IH), unless it is $t$ or $u$. Say it is $t$. Then, $t$ is a maximal subterm of $\alpha$ which avoid $\qvar$, and thus it must be that $t_{0}$ and $t_{1}$ are also public terms of the premises. Thus $T \DYderives \{t_{0},t_{1}\}$ by IH, and hence $T \DYderives t$. Similarly for $u$. 

\item $\rnrule = \rnproj$: $\alpha$ is $\equals{t}{u}$, and any public term of $\alpha$ is a public term of the premise (and we can apply IH), unless it is $t$ or $u$. But by abstractability, $T \DYderives \{t,u\}$, and we are done. 

\item $\rnrule = \rnexintro$: $\alpha$ is of the form $\exists{x}.\beta$, with premise $\gamma = \replsubtermat{\beta}{P}{r}$, where $P = \posof{x}{\beta}$. We also have, by the other requirements for the rule, $T \DYderives r$ and $P \subseteq \abstractable(T\cup\{x\},\beta)$. By Lemma~\ref{lem:absderiv}, $P \subseteq \abstractable(T,\gamma)$. Consider any $a = \subtermat{\alpha}{q} \in \publics(\alpha)$. If $a\in\publics(\gamma)$, then we can apply IH. Otherwise, $q$ has to be a sibling of some position in $p \in P$. In other words, $a$ is public in $\alpha$ because its sibling is $x$, but in $\gamma$, the $x$ is replaced by $r$ (and $\varsof(r) \cap \qvar = \emptyset$), so $a$ is no longer a \emph{maximal} subterm avoiding $\qvar$. Since the set of abstractable positions is sibling-closed, $q\in\abstractable(T, \alpha)$, and since subterms at abstractable positions are derivable, $T\DYderives a$.  
\end{itemize} 

Now consider an $\eqderives$ proof of $(T;E) \vdash \equals{t}{u}$. It has been shown above that $T \DYderives \publics(\equals{t}{u})$. Consider $t$. Either $t \in \publics(\equals{t}{u})$, in which case we are done. Otherwise, every maximal subterm of $t$ which avoids $\qvar$ is derivable from $T$, and every $x \in \varsof(t) \cap \qvar$ is in $T$. From these, we can ``build up'' $t$ using constructor rules only, thereby proving that $T \DYderives t$. Similarly we can show that $T \DYderives u$. 
%$a$ is a proper subterm of some $b \in \publics(\beta)\setminus\publics(\alpha)$. (Thus, there exist $q, i$ such that $p = q\cdot{i}$.) Since $b \in \publics(\beta)$, $\varsof(b) \cap \boundvars(\beta) = \emptyset$. Since $b \not\in \publics(\alpha)$, either $\varsof(b) \cap \boundvars(\alpha) \neq \emptyset$, or there is some larger ``ancestor'' of $b$ which avoids $\boundvars(\alpha)$. We analyze these two cases below.
%    \begin{itemize}[leftmargin=*]
%    \item Let $c = \subtermat{\alpha}{q}$ be the ``parent'' of $a$ in $\alpha$. Thus, $c$ is a subterm of $b$, and $\varsof(c) \cap \boundvars(\beta) = \emptyset$. If $x \notin \varsof(c)$, the maximality of $a$ is contradicted. So $x \in \varsof(c)$. Consider an $x \in \varsof(c)$ at position $r$. We can think of $r$ as $q\cdot{j}\cdot{q'}$. Since $r \in P \subseteq \abstractable(T \cup \{x\}, \beta)$, and $a$ appears at $q\cdot{i}$, by abstractability, $T \cup \{x\} \DYderives a$. Since $x \notin \varsof(T\cup\{a\})$, $T \DYderives a$.
%    \item $\varsof(c) \cap \boundvars(\alpha) = \emptyset$. But $c$, being an ancestor of $b$, is also an ancestor of $a$, which violates the maximality of $a$.
%    \qedhere
%    \end{itemize}
    \end{proof}

\eqstogamma* 
\begin{proof}
\phantom{a}
%
%We will use the following claim, provable easily.
%\begin{claim}\label{claim:absaftersubst}
%	Let $T \cup \{t, r\} \subseteq \Tterms$ s.t. $T \DYderives r$. Let $P = \posof{x}{t}$ and $p \in \positions{t}$. Then, $p \in \abstractable(T, \replsubtermat{t}{P}{r})$ iff $p \in \abstractable(T \cup \{x\}, t)$.
%\end{claim}
%
[$1 \Rightarrow 2$:] Suppose $(S;A) \assderives \alpha$. Then $(T;E) \assderives \alpha$  via a proof which has no occurrence of the $\rnexelim$ rule. Let $\pi$ be a normal proof of $(T;E) \vdash \alpha$. Since $E$ only has equalities, $\pi$ ends with a series of applications of $\rnexintro$. For each $i \leq k$, define $\mu(x_{i}) \coloneqq w_{i}$, where $w_{i}$ is the witness used by the $\rnexintro$ rule introducing $\exists{x_{i}}$ in $\alpha$. The side condition for $\rnexintro$ guarantees that $T \DYderives \mu(x)$ for each $x \in dom(\mu)$, thus satisfying condition \ref{item:dycond} of the theorem. It is easy to see that $\dom(\mu) \cap \varsof(\rng(\mu)) = \emptyset$, since each $\mu(x)$ is derivable from $T$, which has no occurrence of any $y \in \dom(\mu)$.
        
        There is a maximal subproof $\pi'$ of $\pi$ with no occurrence of $\rnexintro$. The conclusion of this is $\equals{t}{u}$, but with $w_{i}$ replacing $x_{i}$ for every $i \leq k$. In other words, $\concof(\pi')$ is $\mu(\equals{t}{u})$. So \ref{item:eqcond} is also satisfied. 
        
        We now prove \ref{item:abscond}. It suffices to consider $k = 3$ to convey the essential idea. For $i \leq 3$, let $\alpha_{i}$ be the formula $\exists{x_{i+1}\ldots{}x_{3}}.(\equals{t}{u})$. Note that $\mu(\alpha_{0}) = \alpha_{0} = \alpha$ and $\alpha_{3} = \equals{t}{u}$. The proof $\pi$ can be viewed as successively proving $\mu(\alpha_{3}), \mu(\alpha_{2}), \mu(\alpha_{1})$ and $\mu(\alpha_{0})$. 
        
        We now assume that $\posof{x_{3}}{t} \subseteq \abstractable(T \cup \dom(\mu), t)$, and show the same for $x_{2}$. 
        %         in $\mu(\alpha_{2})$, and it can be seen that neither $x_{1}$ nor $x_{2}$ occur in $\mu(\alpha_{2})$\footnote{$x_{1}$ and $x_{2}$ are replaced by $\mu(x_{1})$ and $\mu(x_{2})$, both of which are derivable from $T$. Since $\{x_{1},x_{2}\} \cap \varsof(T) = \emptyset$, it follows that $\{x_{1},x_{2}\} \cap \varsof(\mu(x_{1}),\mu(x_{2})) = \emptyset$.}\!\!. 
         Letting $\mu' = {\mu\restr\{x_{1}\}}$, we see that $\mu(\alpha_{2}) = \replsubtermat{(\mu'(\alpha_{3}))}{P}{w_{2}}$, where $P = \posof{x_{2}}{\mu'(\alpha_{3})}$. The side condition for $\rnexintro$ obtaining $\mu(\alpha_{1})$ from $\replsubtermat{(\mu'(\alpha_{3}))}{P}{w_{2}}$ says that $\posof{x_{2}}{\mu'(\alpha_{3})} \subseteq \abstractable(T \cup \{x_{2}\}, \mu'(\alpha_{3}))$. We see that $p \in \posof{x_{2}}{\mu'(t)}$ iff $00p \in \posof{x_{2}}{\mu'(\alpha_{3})}$, and $p \in \abstractable(T \cup \{x_{2},x_{3}\}, \mu'(t))$ iff $00p \in \abstractable(T \cup \{x_{2}\}, \mu'(\alpha_{3}))$. So the side condition is equivalent to $\posof{x_{2}}{\mu'(t)} \subseteq \abstractable(T \cup \{x_{2},x_{3}\}, \mu'(t))$.
         
        Since $x_{2} \notin \varsof(\rng(\mu'))$, $\posof{x_{2}}{\mu'(t)} = \posof{x_{2}}{t}$. Further, $T \DYderives \mu'(x_{1})$, so by Lemma~\ref{lem:absderiv}, $\abstractable(T \cup \{x_{2},x_{3}\}, \mu'(t)) = \abstractable(T \cup \{x_{1},x_{2},x_{3}\}, t)$. Thus we see that $\posof{x_{2}}{t} \subseteq \abstractable(T \cup \dom(\mu), t)$. We can similarly prove the same for $x_{1}$, and obtain a proof of~\ref{item:abscond} for $t$. The proof for $u$ is similar.
        
        \defemph{[$2 \Rightarrow 1$:]} Given the proof of $(T;E) \vdash \mu(\equals{t}{u})$ we can apply $\rnexintro$ successively to get a proof of $\alpha$. We just need to show that the abstractability side conditions are satisfied. Using the same notation as in the proof of [$1 \Rightarrow 2$] above, we consider the passage from $\replsubtermat{(\mu'(\alpha_{3}))}{P}{w_{2}}$ to $\mu(\alpha_{2})$. As above, we need to show that $\posof{x_{2}}{\mu'(t)} \subseteq \abstractable(T \cup \{x_{2},x_{3}\}, \mu'(t))$ and similarly for $u$. But we saw above that this is equivalent to $\posof{x_{2}}{t} \subseteq \abstractable(T \cup \dom(\mu), t)$, which is guaranteed by \ref{item:abscond}. \qedhere

\end{proof}

\section{Proof of Theorem~\ref{thm:deriv-witness-bounds}}\label{app:decidable-assderives}

We recommend reading this appendix after reading the main paper, since we reuse some notions from Section~\ref{sec:insecurity} here.

Let $(S;A)$ be sanitized, and $\alpha = \exists{x_{1}\ldots{}x_{k}}.(\equals{t^{0}}{u^{0}})$ be an assertion with $\boundvars(\alpha) \cap \varsof(S;A) = \emptyset$. Let $(T;E) = \dc(S;A)$. By Theorem~\ref{thm:eqs-to-gamma}, $(S;A) \assderives \alpha$ iff there is a substitution $\mu$ with $\dom(\mu) = \boundvars(\alpha)$ s.t.: 
\begin{enumerate}[label={[\arabic*]}]
    \item\label{appitem:dycond} $\forall{}x \in \dom(\mu): T \DYderives \mu(x)$.
    \item\label{appitem:abscond} $\forall{}x \in \dom(\mu), r \in \{t^{0},u^{0}\}$: $\posof{x}{r} \subseteq \abstractable(T \cup \boundvars(\alpha), r)$.
    \item\label{appitem:asscond} $(T;E) \eqderives \mu(\equals{t^{0}}{u^{0}})$.  
\end{enumerate} 

In this section, we prove that if there is a $\mu$ satisfying the above conditions, there is also a small $\nu$ satisfying the same. We use $\dommu$ to refer to $\dom(\mu)$. Let $\constst = \st(T) \cup \st(E \cup \{\alpha\})$, and let $\stnonvars = \constst \setminus \dommu$. For all $x \in Z$, $T \DYderives \mu(x)$. So all variables occurring in $\mu(x)$ must also be in $\varsof(T)$. But $\varsof(T;E) \cap \dommu = \emptyset$, so $\varsof(\mu(x)) \cap \dommu = \emptyset$ for any $x \in \dommu$. 

Define $t \approx u$ iff $T;E \eqderives \mu(\equals{t}{u})$. It is easy to see that $\approx$ is a partial equivalence relation (on the subset of terms $t$ such that $T \DYderives \mu(t)$).   

We say that $x \in \dommu$ is \emph{minimal} if there is no $t \in \stnonvars$ with $x \approx t$. Let $\vars_{m}$ denote the set of all minimal variables. Our strategy for finding a small $\nu$ is to ``zap'' minimal variables, and propagate the change to (interpretations of) non-minimal variables. For this, it is convenient to translate every term to an ``equivalent'' one with only minimal variables. This notion of equivalence is based on unifiability under $\mu$. The set of terms which are equivalent to terms in $\constst$ is defined as follows. 

\begin{definition}
    $\hatst \coloneqq \{t \mid \varsof(t) \cap \dommu \subseteq \vars_{m}$, either $t \in \vars_{m}$ or $\exists{u} \in \stnonvars:t \approx u\}$. 
\end{definition}

\begin{lemma}\label{lem:st-to-hatst} 
    For every $t \in \constst$ with $T \DYderives \mu(t)$, there is $\expplus{t} \in \hatst$ such that: $T \DYderives \mu(\expplus{t})$; $t \approx \expplus{t}$; and for all $x \in \vars_{m}$, $\posof{x}{\expplus{t}} \subseteq \abstractable(T \cup \dommu, \expplus{t})$. 
\end{lemma}
\begin{proof}
    For $x, y \in \dommu$, $x \prec y$ iff $\exists{r}\in \stnonvars:x \in \st(r) \text{ and } r \approx y$. 

    We now show that $\prec$ is acyclic. Towards this, we claim that if $x \prec y$ and $y \prec z$, then there is some term $a$ (not necessarily in $\constst$) s.t.\ $\mu(x)$ is a proper subterm of $\mu(a)$ and $a \approx z$. Extending this reasoning, we see that if $x \prec^{+} x$, we have some term $a$ such that $\mu(x)$ is a proper subterm of $\mu(a)$ and $(T;E) \eqderives \equals{\mu(a)}{\mu(x)}$. But $E$ is consistent, which means that there is some $\lambda$ s.t.\ $\lambda(\mu(a)) = \lambda(\mu(x))$. But this is incompatible with $\mu(x)$ being a proper subterm of $\mu(a)$. Thus $\prec$ is acyclic. 

    We now prove the claim. Suppose $x \prec y$ and $y \prec z$. Then there exist $r,s \in \stnonvars$ such that the following hold:
    \begin{itemize}
    \item $x \in \st(r)$ and $(T;E) \eqderives \equals{\mu(r)}{\mu(y)}$
    \item $y \in \st(s)$ and $(T;E) \eqderives \equals{\mu(s)}{\mu(z)}$
    \end{itemize}
    Let $a = \replsubtermat{s}{\posof{y}{s}}{r}$, i.e. the term obtained by replacing each $y$ in $s$ by $r$. Since $\mu(x) \in \st(\mu(r))$, $\mu(x)$ is a proper subterm of $\mu(a)$. One can derive $(T;E) \vdash \equals{\mu(a)}{\mu(z)}$ via $\rnsubst$ on $\equals{\mu(s)}{\mu(z)}$ and $\equals{\mu(r)}{\mu(y)}$. Thus $a \approx z$. 

    Since $\prec$ is acyclic, we can define a notion of \emph{rank} for variables: $\rank(x) = \max\{\rank(y) \mid y \prec^{+} x\} + 1$. For a term $u \in \stnonvars$, we define $\rank(u) = \max\{\rank(x) \mid x \in \varsof(u) \cap \dommu\}$. It is easy to see that if $u \in \stnonvars$ and $x \approx u$, then $\rank(x) > \rank(u)$. 
    It is also easy to see that every $x \in Z$ has rank $\geq 1$, with $x \in \vars_{m}$ having rank $1$. 
%    It is possible for some $x \in Z\setminus\vars_{m}$ to have rank $1$ (if $\varsof(u) \cap Z = \emptyset$ for every $u$ s.t.\ $x\approx u$). 

    We now prove the lemma by induction on $\delta(t) = (\rank(t), |t|)$, where $|t|$ is the number of operators in $t$. Fix an ordering on $\hatst$. For $\delta(t) = (0,0)$, $t \in \names\cup\vars\setminus{}Z$, and we take $\expplus{t} = t$.       
%        \item $t = x \in \vars\setminus\vars_{m}$: This means that there is some $u \in \stnonvars$ s.t.\ $x \approx u$. But since $\rank(x) = 0$, $\varsof(u) \cap \dommu = \emptyset$ for each such $u$. Choose $\expplus{t}$ to be the earliest such $u$ (according to the ordering on $\hatst$). Clearly $(T;E) \vdash \equals{\mu(x)}{\mu(\expplus{t})}$, and by purity, $T \DYderives \mu(\expplus{t})$. Finally $\varsof(\expplus{t}) \cap \dommu = \emptyset$, so $\posof{y}{\expplus{t}} = \emptyset \subseteq \abstractable(T \cup \dommu, \expplus{t})$ for all $y \in \vars_{m}$. 
    Suppose $\delta(t) > (0,0)$ and the lemma is true for all $u$ such that $\delta(u) < \delta(t)$. There are three cases to consider:
    \begin{itemize}
        \item {$t = x \in \vars_{m}$:} We choose $\expplus{t} = x$.
        \item $t = x \in Z\setminus\vars_{m}$: There is a $u \in \stnonvars$ s.t.\ $x \approx u$, and $\rank(u) < \rank(x)$. Pick the earliest such $u \in \hatst$. By IH there is a $\expplus{u}$, and we define $\expplus{t} = \expplus{u}$. Since $t \approx u$ and $u \approx \expplus{u}$, we have $t \approx \expplus{t}$, by transitivity.
        \item {$t \notin Z$:} For each $y \in \varsof(t) \cap \dommu$, since $\delta(y) < \delta(t)$, there is a $\expplus{y}$. We obtain $\expplus{t}$ by replacing each $y \in \varsof(t)$ by $\expplus{y}$. Now, $\varsof(\expplus{t}) \cap \dommu \subseteq \vars_{m}$. Also since all variables appear in abstractable positions in $t$, we can justify the relevant applications of $\rnsubst$ to show that $t \approx \expplus{t}$. Finally, if $z$ appears in an abstractable position in $r$ and $y$ appears in an abstractable position in $s$, then $z$ appears in an abstractable position in $\replsubtermat{s}{\posof{y}{s}}{r}$, i.e. the term obtained by replacing every $y$ occurring in $s$ by $r$. Thus the abstractability part of the statement is also fulfilled.  \qedhere
    \end{itemize}
\end{proof} 

We now define the substitution $\nu$ as follows. Assume that there is some $\fixedname \in T \cap \names$ such that $\fixedname \notin \st(E \cup \{\alpha\}) \cup \st(\rng(\mu))$.\footnote{Thus $\fixedname$ is a ``spare name'' that does not occur in any of the derivations under consideration.} Let $\nu_{m}$ be the substitution that maps each $x \in \vars_{m}$ to $\fixedname$. For all $x \in \dommu: \nu(x) = \nu_{m}(\expplus{x})$. Notice that for all $x \in \dom(\nu)$, either $\nu(x) = \fixedname$ or there is $u \in \stnonvars$ s.t.\ $\nu(x) = \nu(u)$. Thus we can follow the proof of Theorem~\ref{thm:vlambda-small} to show that $\nu$ is $|\constst|$-bounded. To complete the proof of Theorem~\ref{thm:deriv-witness-bounds}, we just need to show that $\nu$ preserves derivability. This is proved in Theorem~\ref{thm:nu-simulate-mu}, the main result of this section. But first we state a useful observation. 
\begin{obs}\label{obs:constst-not-mvars}
    \phantom{a}
    \begin{enumerate}
        \item For $x \in \dommu$, if $\mu(x) \in \constst$ then $x \notin \vars_{m}$. 
        \item If $t \in \hatst$ and $\mu(t) \in \constst$, then $\varsof(t) \cap \dommu = \emptyset$ and $\mu(t) = t$.  
    \end{enumerate}
\end{obs}
\begin{proof}
    \phantom{a} 
    \begin{enumerate}
        \item Let $\mu(x) = t \in \constst$. Since $\varsof(t) \cap \dommu = \emptyset$, we have that $t \notin \dommu$ and $\mu(t) = t$. Thus $t \in \stnonvars$, and $\equals{\mu(x)}{\mu(t)}$ is derivable using the $\rneq$ rule, i.e., $x \approx t$. Therefore $x \notin \vars_{m}$.
        \item For every $x \in \varsof(t) \cap \dommu$, $\mu(x) \in \constst$. Thus we have $x \notin \vars_{m}$, by the previous part. Since $t \in \hatst$, we have that $\varsof(t) \cap \dommu \subseteq \vars_{m}$. Thus $\varsof(t) \cap \dommu = \emptyset$ and  $\mu(t) = t$. 
        \qedhere
    \end{enumerate}
\end{proof}

\begin{theorem}\label{thm:nu-simulate-mu} 
    \phantom{a}
    \begin{enumerate}
        \item For any $t \in \constst$, if $T \DYderives \mu(t)$ then $T \DYderives \nu(t)$. 
        \item For any $t, u \in \constst$, if $T; E \eqderives \equals{\mu(t)}{\mu(u)}$ then $T; E \eqderives \equals{\nu(t)}{\nu(u)}$. 
    \end{enumerate}
\end{theorem}
\begin{proof}
    Lemma~\ref{lem:st-to-hatst} guarantees appropriate $\expplus{t}, \expplus{u} \in \hatst$, which we will refer to as $r$ and $s$ respectively. It suffices to prove that if $T \DYderives \mu(r)$ then $T \DYderives \nu_{m}(r)$, and that if $T;E \eqderives \equals{\mu(r)}{\mu(s)}$ then $T;E \derives \equals{\nu_{m}(r)}{\nu_{m}(s)}$. 

    \begin{enumerate}
        \item Suppose $T \DYderives r$. Since variables from $\dommu$ only occur in $r$ at abstractable positions (w.r.t.\ $T \cup \dommu$), and since $T \DYderives \fixedname$, we can easily prove by induction on the size of terms that $T \DYderives \nu_{m}(r)$. 
        \item Let $\pi$ be a normal proof of $T;E \vdash \equals{\mu(r)}{\mu(s)}$ with last rule $\rnrule$. We prove the desired statement by induction on the structure of $\pi$. The following cases arise. 
        \begin{itemize}[leftmargin=*]
            \item {$\rnrule  = \rnax$:} Then $\equals{\mu(r)}{\mu(s)} \in E$, and thus $\mu(r), \mu(s) \in \constst$. By Observation~\ref{obs:constst-not-mvars}, $\varsof(r,s) \cap \dommu = \emptyset$. Thus $\nu_{m}(r) = r = \mu(r)$ and $\nu_{m}(s) = s = \mu(s)$. Therefore $\pi$ itself is a proof of $\equals{\nu_{m}(r)}{\nu_{m}(s)}$.
             \item {$\rnrule = \rneq$:} $\mu(r) = \mu(s)$ and $T \DYderives \mu(r)$ via a proof ending in ${\sf ax}$ or a destructor rule, and thus $\mu(r), \mu(s) \in \st(T) \subseteq \constst$. Repeating the above argument, $\nu_{m}(r) = \mu(r)$ and $\nu_{m}(s) = \mu(s)$, and $\pi$ is the desired proof of $\equals{\nu_{m}(r)}{\nu_{m}(s)}$.  
            \item {$\rnrule = \rnproj$:} By subterm property for normal $\eqderives$-proofs $\mu(r), \mu(s) \in \st(T;E) \subseteq \constst$. $\equals{\nu_{m}(r)}{\nu_{m}(s)}$ is provable using $\pi$, as in the $\rnax$ case.            	
            \item {$\rnrule = \rnsymm$:} The immediate premise is $\equals{\mu(s)}{\mu(r)}$. By IH, we have a proof of $\equals{\nu_{m}(s)}{\nu_{m}(r)}$, to which we apply $\rnsymm$ to obtain $\equals{\nu_{m}(r)}{\nu_{m}(s)}$.  
            \item {$\rnrule = \rntrans$:} Suppose the immediate subproofs are $\pi_{1}, \ldots, \pi_{n}$, with each $\pi_{i}$ deriving $\equals{v_{i-1}}{v_{i}}$. Let $\mu(r) = v_{0}$ and $\mu(s) = v_{n}$. Since no $\pi_{i}$ ends in $\rntrans$ and no two adjacent $\pi_{i}$s end in $\rncons$, each $v_{i}$ (for $0 < i < n$) appears in at least one proof ending in $\rnax$, $\rneq$, $\rnsymm$ or $\rnproj$. Thus, by the subterm property, $v_{i} \in \st(T;E) \subseteq \constst$ for $0 < i < n$. Since $\varsof(T;E) \cap \dommu = \emptyset$, it follows that $v_{i} \in \hatst$ and $\mu(v_{i}) = v_{i}$. Thus we can view each $\pi_{i}$ as deriving $\equals{\mu(r_{i-1})}{\mu(r_{i})}$, where $r_{i-1}, r_{i} \in \hatst$ (taking $r_{0}$ and $r_{n}$ to be $r$ and $s$). By IH, there are proofs $\varpi_{1}, \ldots, \varpi_{n}$, with each $\varpi_{i}$ deriving $\equals{\nu_{m}(r_{i-1})}{\nu_{m}(r_{i})}$. By composing them using $\rntrans$, we get a proof of $T;E \vdash \equals{\nu_{m}(r)}{\nu_{m}(s)}$, as desired.

            \item {$\rnrule = \rncons$:} Suppose $r = \func(r_{1}, r_{2})$ and $s = \func(s_{1},s_{2})$. Each $r_{i}, s_{i} \in \hatst$, and the immediate subproofs are $\pi_{1}$ and $\pi_{2}$, deriving $\equals{\mu(r_{1})}{\mu(s_{1})}$ and $\equals{\mu(r_{2})}{\mu(s_{2})}$, respectively. By IH we have proofs $\varpi_{1}$ and $\varpi_{2}$, with $\varpi_{i}$ proving $\equals{\nu_{m}(r_{i})}{\nu_{m}(s_{i})}$. We can compose them with the $\rncons$ rule to get the desired proof of $\equals{\nu_{m}(r)}{\nu_{m}(s)}$. 

            Suppose, on the other hand, that $r$ is a variable. Since $r \in \hatst$, $r \in \vars_{m}$. Now $s \in \hatst$, so either $s \in \vars_{m}$ or there is $a \in \stnonvars$ with $s \approx a$. But in the second case, $r \approx a$ (by symmetry and transitivity), which cannot happen for a minimal variable $r$. Therefore $s \in \vars_{m}$. And we have $\nu_{m}(r) = \nu_{m}(s) = \fixedname \in T$, so there is a proof of $T, E \eqderives \equals{\nu_{m}(r)}{\nu_{m}(s)}$ ending in $\rneq$. 

            We have a similar argument in case $s$ is a variable, thereby proving the theorem. \qedhere
        \end{itemize}
    \end{enumerate}
\end{proof}

\section{Algorithm to decide \texorpdfstring{$\eqderives$}{equality derivations}}
\label{app:conseq}
We present a saturation-based procedure in Algorithm~\ref{alg:eqderive} for deciding whether $T; E \eqderives \equals{t}{u}$. The procedure first computes the set $\conseq{T}{E}{t}{u}$ defined below. 
\[
    \bigl\{\equals{r}{s} \mid r,s \in \st(T\cup\{t,u\}) \cup \st(E), (T; E) \eqderives \equals{r}{s}\bigr\}.
\]
It then checks whether $\equals{t}{u} \in \conseq{T}{E}{t}{u}$.

We start out with a set $C$, which contains all the equalities in $E$, and trivial equalities over all terms $t$ such that $T \DYderives t$. We assume that a non-atomic $r$ is of the form $\func(r_{0},r_{1})$, and similarly for $s$.

$C$ initially contains all equalities obtained using $\rnax$ and $\rneq$. $C_{1}$ is all equalities obtained by one application of the $\rntrans$ rule to the formulas in $C$, $C_{2}$ all those obtained using $\rncons$, and $C_{3}$ all those obtained using $\rnproj_{i}$. These sets are added to $C$ and the procedure iterated till nothing new can be added.

Letting $Z = \st(T \cup \{t,u\}) \cup \st(E)$, and $M = |Z|$, it can be seen that the algorithm runs in time polynomial in $M$. There are at most $M^{2}$ equalities that can be added to $C$, and hence the \defemph{while} loop runs for at most $M^{2}$ iterations. In each iteration, the amount of work to be done is polynomial in $M$. (Recall that $\DYderives$ can be decided in PTIME.) Thus the algorithm works in time polynomial in $M$.

\begin{algorithm}[!t]
	\caption{Algorithm to compute $\conseq{T}{E}{t}{u}$, given $(T; E), t,u$}
	\label{alg:eqderive}
	\begin{algorithmic}[1]
		\State $Z \gets \st(S \cup \{t,u\}) \cup \st(E)$;
		\State $B \gets \emptyset$;
		\State $C \gets E \cup \big\{\equals{t}{t} \mid t \in Z, T \DYderives t\big\}$;
		\While{$(B \neq C)$}
		\State $B \gets C$;
		\State $C_{1} \gets \big\{\equals{r}{s} \mid r,s \in Z$, and there is a $v$ s.t.\ $\{\equals{r}{v}, \equals{v}{s}\} \subseteq B\big\}$;
		\State $C_{2} \gets \big\{\equals{r}{s} \mid r, s \in Z, \{\equals{r_{0}}{s_{0}}, \equals{r_{1}}{s_{1}}\} \subseteq B\big\}$;
		\State $C_{3} \gets \big\{\equals{r_{i}}{s_{i}} \mid \equals{r}{s} \in B, S \DYderives \{r_{0},r_{1},s_{0},s_{1}\}\big\}$;
		\State $C \gets B \cup C_{1} \cup C_{2} \cup C_{3}$;
		\EndWhile
		\State \Return $B$.
	\end{algorithmic}
\end{algorithm}

\section{Proofs for Section~\ref{sec:insecurity}}
\label{app:insecurity}

\equndertheta*
\begin{proof}
	Suppose $T; E \derives \equals{t}{u}$ via a proof $\pi$ with last rule $\rnrule$. The proof is by induction on the structure of $\pi$. The following cases arise. 
	\begin{itemize}[leftmargin=*]
		\item $\rnrule = \rnax$: $\equals{t}{u} \in E$, so by assumption, $\lambda(t) = \lambda(u)$.
		\item $\rnrule = \rneq$: $t = u$, so $\lambda(t) = \lambda(u)$ as well.
		\item $\rnrule = \rnsymm$: The premise is $\equals{u}{t}$, and by IH, $\lambda(u) = \lambda(t)$.
		\item $\rnrule = \rntrans$: Suppose $\equals{t_{0}}{t_{1}}, \ldots, \equals{t_{n-1}}{t_{n}}$ are the premises of $\rnrule$, with $t = t_{0}$ and $u = t_{n}$. By IH, $\lambda(t_{i-1}) =  \lambda(t_{i})$ for all $i \leq n$. It follows that $\lambda(t) = \lambda(u)$.
		\item $\rnrule = \rncons$: Let $t = \func(t_{1},t_{2})$ and $u = \func(u_{1},u_{2})$ and let $\equals{t_{1}}{u_{1}}, \equals{t_{2}}{u_{2}}$ be the premises of $\rnrule$. By IH, $\lambda(t_{1}) = \lambda(u_{1})$ and $\lambda(t_{2}) = \lambda(u_{2})$. Thus we have the following:
		\begin{tabbing}
		    \hspace{0.25cm}\=
            $\lambda(t)$ \=$= \lambda(\func(t_{1},t_{2})) = \func(\lambda(t_{1}), \lambda(t_{2})) = \func(\lambda(u_{1}), \lambda(u_{2}))$ \\ \>\> 
            $=\lambda(\func(u_{1}, u_{2})) = \lambda(u)$.
		\end{tabbing}
		\item $\rnrule = \rnproj$: Let $\equals{\func(t_{1},t_{2})}{\func(u_{1},u_{2})}$ be the premise of the last rule with $t = t_{1}$ and $u = u_{1}$, w.l.o.g. By IH, $\lambda(\func(t_{1}, t_{2})) = \lambda(\func(u_{1}, u_{2}))$. So, $\lambda(t) = \lambda(u)$. \qedhere
	\end{itemize}
\end{proof}

\prooftotruth*
\begin{proof}
	Claim 1 is vacuously true for $E_{0} = \emptyset$. We prove the claims simultaneously by induction on $i > 0$. Assume that they hold for all $j < i$ via IH1, IH2, and IH3. 

	\begin{enumerate}[leftmargin=*]
		\item Suppose $\equals{t}{u} \in E_{i}$. Then, $\exists{}j < i: \equals{t}{u} \in \SF(\alpha_{j})$, and $\intsub(U_{j};F_{j}) \eqderives \intsub\hwsub_{j}(\equals{t}{u})$. By IH3, $\bigsub(t) = \bigsub(u)$. If $\equals{t}{u} \in F_{i}$, then $\exists{}j \leq i: \equals{t}{u} \in \SF(\beta_{j})$, and $\intsub(T_{j-1};E_{j-1}) \eqderives \intsub\iwsub_{j}(\equals{t}{u})$. If $j < i$, by IH2, $\bigsub(t) = \bigsub(u)$. If $j = i$, by IH1, $\bigsub({r}) = \bigsub({s})$ for every $\equals{{r}}{{s}} \in E_{i-1}$. Any $\equals{a}{b} \in \intsub(E_{i-1})$ is of the form $\intsub(\equals{r}{s})$ for some $\equals{r}{s} \in E_{i-1}$. Thus, $\bigsub(a) = \bigsub(\intsub(r)) = \bigsub(r) = \bigsub(s) = \bigsub(\intsub(s)) = \bigsub(b)$. By Lemma~\ref{lem:equndertheta}, $\bigsub(\intsub\iwsub_{j}(t)) = \bigsub(\intsub\iwsub_{j}(u))$, i.e. $\bigsub(t) = \bigsub(u)$. 
		
        \item Suppose $\intsub(T_{i-1}; E_{i-1}) \eqderives \intsub\iwsub_{i}(\equals{t}{u})$. As above, for each $\equals{a}{b} \in \intsub(E_{i-1})$, $\bigsub(a) = \bigsub(b)$. Using Lemma~\ref{lem:equndertheta}, $\bigsub(\intsub\iwsub_{i}(t)) = \bigsub(\intsub\iwsub_{i}(u))$, i.e. $\bigsub(t) = \bigsub(u)$. 
        
        \item The proof is similar to the above. \qedhere
	\end{enumerate}
\end{proof}

\sigmaxinsigmaTp*
\begin{proof}
    Consider $t \in \st(\intsub(u)) \setminus (\intsub(\stnonvars) \cup \qvar)$ for some $u \in T_{i}$. Then, $t \in \st(\intsub(y))$ for some $y \in \varsof(u)$. Since $u \in T_{i}$, there is a $j < i$ such that $u \in \honterms_{j} \cup \qvar$. If $u \in \qvar$, then $u = y = \sigma(y)$ and $t = y$, but we know that $t \not\in \qvar$. Thus $u \not\in \qvar$ and $u \in \honterms_{j}$, i.e. $y \in \varsof(\honterms_{j})$. Now $\xi$ is an interleaving of sessions of $\prot$, and $y \in \varsof(u)$ where $u$ occurs in an honest agent send in a session. Thus by Observation~\ref{obs:earlierbetaj}, there is an earlier intruder send in the same session in which $y$ occurs. This send occurs before $\alpha_{j}$ in $\xi$. Thus there is a $k \leq j$ such that $y \in \varsof(\publics(\beta_{k})) = \varsof(\intterms_{k})$. Thus, $t \in \st(\intsub(\intterms_{k}))$.
\end{proof}

\earlierproof*
\begin{proof}
  	Since $\pi$ ends in a destructor rule, $t \in \st(\intsub(T_{i}))$. By Lemma~\ref{lem:sigmax-in-sigmaTp}, there is an $i' < i$ such that $t \in \st(\intsub(\intterms_{i'}))$. 
    Let $j$ be the earliest such index, and let $a \in \intterms_{j}$ such that $t \in \st(\intsub(a))$. Since $\intsub(T_{j-1};E_{j-1}) \assderives \intsub\iwsub_{j}(\beta_{j})$, and $a \in \intterms_{j} = \publics(\beta_{j})$, it follows by Observation~\ref{obs:dc-pure} that $\intsub(T_{j-1}) \DYderives \intsub\iwsub_{j}(a)$. But $\varsof(a) \cap \dom(\iwsub_{j}) = \emptyset$, so $\intsub(T_{j-1}) \DYderives \intsub(a)$, via a normal proof $\rho$. Consider a minimal subproof $\chi$ of $\rho$ such that $t \in \st(\concof(\chi))$. (There is at least one such subproof, namely $\rho$.) If $\chi$ ends in a destructor, then $\concof(\chi) \in \st(\intsub(T_{j-1}))$, and hence $t \in \st(\intsub(T_{j-1}))$. But by Lemma~\ref{lem:sigmax-in-sigmaTp}, there must be a $k < j-1$ such that $t \in \st(\intsub(\intterms_{k}))$, contradicting the fact that $j$ is the earliest such index. So $\chi$ ends in a constructor rule. If $t \neq \concof(\chi)$, then $t \in \st(\concof(\chi'))$, for some proper subproof of $\chi$. But this cannot be, since $\chi$ is a minimal proof with this property. Thus, $t = \concof(\chi)$ and $\chi$ is a proof of $\intsub(T_{j-1}) \vdash t$ (and we choose our $\ell$ to be $j-1$).           
\end{proof}

\projcase*
\begin{proof} 
    For any $a \in \intsub(\constst)$, one of the following cases arises: 
    \begin{itemize}[leftmargin=*]
    \item $a \in \intsub(\stnonvars)$.
    \item $a = \intsub(x)$ for some $x \in \qvar$, so $a = x \in \qvar$.
    \item $a = \intsub(x)$ for $x \in \dom(\intsub)$, and so $a = \bigsub(x) \in \bigsub(\constst)$. 
    \end{itemize}
    Thus $\intsub(\constst) \subseteq \intsub(\stnonvars) \cup \bigsub(\constst) \cup \qvar$. 
    
    We know that $t,u$ are typed, but $t,u\notin \qvar$ (since they are non-atomic). So $t, u \in \intsub(\stnonvars) \cup \bigsub(\constst)$. 
    
    To prove the Lemma, we consider two cases.
    \begin{itemize}
        \item {\bf Neither $t$ nor $u$ is zappable:} Consider $t$. If $t \in \intsub(\stnonvars)$, each $t_{i} \in \intsub(\constst) \subseteq \intsub(\stnonvars) \cup \bigsub(\constst) \cup \qvar$. If $t \in \bigsub(\constst)$, then since $t$ is not zappable, $t = \bigsub(a)$ for some $a \in \stnonvars$. Then $a$ has to be of the form $\func(a_{0},a_{1})$, with $a_{0},a_{1} \in \constst$, $t_{0} = \bigsub(a_{0})$ and $t_{1} = \bigsub(a_{1})$. Thus $t_{0}, t_{1} \in \bigsub(\constst) \subseteq \intsub(\stnonvars) \cup \bigsub(\constst) \cup \qvar$. Reasoning about $u$ in a similar manner, we see that  $u_{0},u_{1} \in \intsub(\stnonvars) \cup \bigsub(\constst) \cup \qvar$. So $t_{0}, t_{1}, u_{0}$ and $u_{1}$ are typed.
        \item {\bf One of $t$ and $u$ is zappable:} Say $t$ is zappable. Then $u$ is zappable as well, by Observation~\ref{obs:tzap-iff-uzap}. Therefore $t, u \notin \intsub(\stnonvars)$, which implies that $t, u \in \bigsub(\constst)$. Therefore both $t$ and $u$ are ground terms, so $t = \bigsub(t) = \bigsub(u) = u$. \qedhere
    \end{itemize}
\end{proof}

\rustureq*
\begin{proof}
    Let $\pi$ be a normal $\eqderives$ proof of $\intsub(T_{i}; E_{i}) \vdash \equals{t}{u}$ ending in $\rnrule$. Assume all proper subproofs of $\pi$ are typed.
    \begin{itemize}[leftmargin=*]
        \item {$\rnrule = \rnax$:} $\equals{t}{u} \in \intsub(E_{i})$. So $t, u \in \intsub(\constst)$ and $\pi$ is typed.
        \item {$\rnrule = \rnsymm$:} By normality, the premise of $\rnrule$, i.e. $\equals{u}{t}$, is the conclusion of $\rnax$. Thus $t,u \in \intsub(\constst)$ as above. 
        \item {$\rnrule = \rneq${\bf :}} $t = u$, so $\pi$ is typed. 
        \item {$\rnrule = \rntrans${\bf :}} Let $\pi_{1}, \ldots, \pi_{k}$ be the immediate (typed) subproofs of $\pi$, with each $\pi_{i}$ deriving $\equals{t_{i}}{t_{i+1}}$, where $t = t_{1}$ and $u = t_{k+1}$. By normality of $\pi$, for all $i \leq k$, $t_{i} \neq t_{i+1}$. By Definition~\ref{def:welltypedeq}, the following cases arise.
        \begin{itemize}
            \item $\rncons$ occurs in some $\pi_{i}$, hence in $\pi$, and $\pi$ is typed. 
            \item $\rncons$ does not occur in any $\pi_{i}$, so every $t_{i}$ (including $t = t_{1}$ and $u = t_{k+1}$) is a typed term, and $\pi$ is typed. 
        \end{itemize}
        \item {$\rnrule = \rnproj${\bf :}} This case is already presented in the main text.
%         Let $a = \func(a_{0},a_{1})$, $b = \func(b_{0},b_{1})$ and $\pi'$ proving $\equals{a}{b}$ be the immediate (typed) subproof of $\pi$. W.l.o.g, let $t = a_{0}$ and $u = b_{0}$. By Lemma~\ref{lem:prooftotruth}, $\bigsub(a) = \bigsub(b)$. By normality, neither $\pi$ nor $\pi'$ contains $\rncons$. The following cases arise.
%        \begin{itemize}
%            \item {$a = b$:} In this case, $t$ and $u$ are immediate subterms of $a$ and $b$, so $t = u$ and $\pi$ is typed.
%            \item {$a$ and $b$ are typed:} We have $\bigsub(a) = \bigsub(b)$. Hence, by Lemma~\ref{lem:projcase}, either $a = b$ (so $t = u$ as above) or $t$ and $u$ are typed. In both cases, $\pi$ is typed.
%        \end{itemize} 
        \item {$\rnrule = \rncons${\bf :}} All subproofs are typed by IH and $\pi$ has an occurrence of $\rncons$, so $\pi$ is typed. \qedhere
    \end{itemize}
\end{proof}

\zaptaunut*
\begin{proof}
Consider $\lambda = \iwsub_{i}$ for some $i$. The following cases arise.
    \begin{itemize}
        \item {$t = x \in \dom(\intsub\lambda)$:} By Definition~\ref{def:vlambda}, $\vintsub\vlambda(x) = \zap{\intsub\lambda(x)}$. 
        \item {$t = x \notin \dom(\intsub\lambda)$:} $\vintsub\vlambda(x) = x = \zap{x} = \zap{\intsub\lambda(x)}$. 
        \item {$t \in \names$:} Since $t \in \constst$, $t \in \stnonvars$ and hence is not zappable. Thus $\intsub\lambda(t) = t = \vintsub\vlambda(t)$. Since $\zap{t} = t$, $\vintsub\vlambda(t) = \zap{\intsub\lambda(t)}$.
        \item {$t = \func(t_{0}, t_{1})$:} $t \in \stnonvars$, so $t_{0},t_{1} \in \constst$, and for $j \leq 1$ we get $\vintsub\vlambda(t_{j}) = \zap{\intsub\lambda(t_{j})}$ by IH. We claim that $u = \intsub\lambda(t)$ is not zappable, since for any $x$ such that $\bigsub(u) = \bigsub(x)$, $x$ is not minimal (since $\bigsub(x) = \bigsub(t)$ as well, and $t \in \stnonvars$). Therefore, we have 
        \begin{align*}
            \zap{\intsub\lambda(t)} &= \zap{\func(\intsub\lambda(t_{1}), \intsub\lambda(t_{2}))} \\
            &= \func(\zap{\intsub\lambda(t_{1})}, \zap{\intsub\lambda(t_{2})}) \\
            &= \func(\vintsub\vlambda(t_{1}),  \vintsub\vlambda(t_{2})) = \vintsub\vlambda(t). \qedhere
        \end{align*}
    \end{itemize}
\end{proof}

\eqzetazap*
\begin{proof}
    Let $(X;A)$ and $(Y;B)$ denote $\intsub(T_{i};E_{i})$ and $\vintsub(T_{i};E_{i})$ respectively. By Observation~\ref{obs:stnu}, $\zap{X} = Y$ and $\zap{A} = B$. Let $\pi$ be a typed normal $\eqderives$ proof of $X;A \vdash \equals{t}{u}$ (guaranteed by Theorem~\ref{thm:rustur-eq}). We prove that $Y;B \eqderives \equals{\zap{t}}{\zap{u}}$. Consider the last rule $\rnrule$ of $\pi$. The following cases arise.
	\begin{itemize}[leftmargin=*]
		\item {$\rnrule = \rnax$}:  In this case, $\equals{t}{u} \in A$, and so $\equals{\zap{t}}{\zap{u}} \in B$, and there is a proof of $Y;B \vdash \equals{\zap{t}}{\zap{u}}$ ending in $\rnax$.
		\item {$\rnrule = \rnsymm$}:  The premise of the rule is $\equals{u}{t}$. By IH, $Y;B \eqderives \equals{\zap{u}}{\zap{t}}$. We get $Y;B \eqderives \equals{\zap{t}}{\zap{u}}$ using $\rnsymm$.
		\item {$\rnrule = \rneq$}: $t = u$, so $\zap{t} = \zap{u}$. Since $X \DYderives t$, by Lemma~\ref{lem:dyzetazap}, $Y \DYderives \zap{t}$. So $Y;B \eqderives \equals{\zap{t}}{\zap{u}}$ with last rule $\rneq$.
		\item {$\rnrule = \rntrans$}: Let $\pi_{1}, \ldots, \pi_{k}$ be the immediate subproofs of $\pi$, each $\pi_{i}$ deriving $X; A \vdash \equals{t_{i}}{t_{i+1}}$, with $t = t_{1}$ and $u = t_{k+1}$. Thus, $\equals{\zap{t}}{\zap{u}} = \equals{\zap{t_{1}}}{\zap{t_{k+1}}}$. By IH, there are $\varpi_{1}, \ldots, \varpi_{k}$, with each $\varpi_{i}$ deriving $Y;B \vdash \equals{\zap{t_{i}}}{\zap{t_{i+1}}}$. We can apply $\rntrans$ to get a proof of $Y;B \vdash \equals{\zap{t_{1}}}{\zap{t_{k+1}}}$. 

		\item {$\rnrule = \rnproj$ or $\rnrule = \rncons$}: This is in the main text.
\qedhere
	\end{itemize}
\end{proof}

%\section{Definitions for the extended syntax}
%\label{app:fullsyntax}
\section{Normalization and subterm property for \texorpdfstring{$\eqderives$}{equality derivations}}
\label{app:normalization}

%Here we prove the normalization theorem and the subterm property for the $\eqderives$ system. We first recall the relevant definition of normal derivations. 
A proof $\pi$ of $T; E \eqderives \equals{t}{u}$ is \emph{normal} if the following hold.
	\begin{enumerate}
		\item All $\DYderives$ subproofs are normal. 
		\item The premise of $\rnsymm$ can only be the conclusion of $\rnax$ or $\rnprom$.
		\item The premise of $\rneq$ can only be the conclusion of a destructor rule.
		\item No premise of a $\rntrans$ is of the form $\equals{a}{a}$, or the conclusion of a $\rntrans$.
		\item Adjacent premises of a $\rntrans$ are not conclusions of $\rncons$.
		\item No premise of $\rnlint$ is the conclusion of $\rnlint$ or $\rnlwk$.
%    	\item No premise of $\rnsubst$ is the conclusion of $\rnsubst$.
		\item No subproof ending in $\rnproj$ contains $\rncons$. 
	\end{enumerate}

A set $E$ of atomic formulas is said to be \emph{consistent} if there is a $\lambda$ s.t.\ $\lambda(t) = \lambda(u)$ for each $\equals{t}{u} \in E$, and $\lambda(t) \in \{t_{1},\ldots,t_{n}\}$ for each $t\listmemb{[t_{1},\ldots,t_{n}]} \in E$. 

\begin{table*}[!t]
\centering {\small
\tabulinesep=0.8mm
\setlength{\tabcolsep}{0.4em}
\begin{tabu}{|l|c|}
    \hline
    \multirow{2}{*}{R1} & $\rneq(\func(\pi_{1},\pi_{2}))$ \\ & $\rncons_{\func}(\rneq(\pi_{1}), \rneq(\pi_{2}))$ \\ \hline 
    \multirow{2}{*}{R2} & $\rnsymm(\rneq(\pi))$ \\ & $\rneq(\pi)$ \\ \hline
    \multirow{2}{*}{R3} & $\rnsymm(\rnsymm(\pi))$ \\ & $\pi$ \\ \hline
    \multirow{2}{*}{R4} & $\rnsymm(\rnrule(\pi_{1}, \ldots, \pi_{k}))$ \\ & $\rnrule(\rnsymm(\pi_{1}), \ldots, \rnsymm(\pi_{k}))$ \\ \hline
    \multirow{2}{*}{R5} & $\rntrans(\pi_{1}, \ldots, \pi_{i-1}, \varpi, \pi_{i}, \ldots, \pi_{r-1})$ \\ & $\rntrans(\pi_{1}, \ldots, \pi_{i-1}, \pi_{i}, \ldots, \pi_{r-1})$ \\ \hline 
    \multirow{2}{*}{R6} & $\rntrans(\pi_{1}, \ldots, \rntrans(\pi^{1}_{i}, \ldots, \pi^{k}_{i}), \ldots, \pi_{r-1})$ \\ & $\rntrans(\pi_{1}, \ldots, \pi^{1}_{i}, \ldots, \pi^{k}_{i}, \ldots, \pi_{r-1})$ \\ \hline 
    \multirow{2}{*}{R7} & $\rntrans(\pi_{1}, \ldots, \rncons(\pi^{1}_{i-1}, \pi^{2}_{i-1}), \rncons(\pi^{1}_{i}, \pi^{2}_{i}), \ldots, \pi_{r-1})$ \\ & $\rntrans(\pi_{1}, \ldots, \rncons(\rntrans(\pi^{1}_{i-1}, \pi^{1}_{i}), \rntrans(\pi^{2}_{i-1}, \pi^{2}_{i})), \ldots, \pi_{r-1})$ \\ \hline 
    \multirow{2}{*}{R8} & $\rnproj_{j}(\rncons(\pi_{1},\pi_{2}))$ \\ & $\pi_{j}$ \\ \hline 
    \multirow{2}{*}{R9} & $\rnproj_{j}(\rntrans(\pi_{1}, \ldots,\pi_{i-1}, \rncons_{\func}(\pi^{1}_{i},\pi^{2}_{i}), \pi_{i+1}, \ldots, \pi_{r-1}))$ \\ & $\rntrans(\rnproj_{j}(\rntrans(\pi_{1}, \ldots,\pi_{i-1})), \pi^{j}_{i}, \rnproj_{j}(\rntrans(\pi_{i+1}, \ldots, \pi_{r-1})))$ \\ \hline 
    \multirow{2}{*}{R10} & $\rnlint(\pi_{1}, \ldots, \pi_{k-1}, \rnlint(\pi_{k}, \ldots, \pi_{m}), \pi_{m+1}, \ldots, \pi_{n})$ \\ & $\rnlint(\pi_{1}, \ldots, \pi_{k-1}, \pi_{k}, \ldots, \pi_{m}, \pi_{m+1}, \ldots, \pi_{n})$ \\ \hline 
    \multirow{2}{*}{R11} & $\rnlint(\pi_{1}, \ldots, \rnweak(\pi_{i}), \ldots, \pi_{n})$ \\ & $\rnweak(\pi_{i})$ \\ \hline
\end{tabu}
}
\caption{Proof transformation rules. The LHS is above and the RHS is below. In R4, $\rnrule \in \{\rntrans, \rnproj, \rncons\}$. In R5, $\concof(\varpi)$ is assumed to be of the form $\equals{a}{a}$.}        
\label{tab:rwrules}
\end{table*}

   We next prove normalization for $\eqderives$ proofs (with a consistent LHS). We present proof transformation rules in Table~\ref{tab:rwrules}. To save space, we use \emph{proof terms} -- $\rnrule (\pi_{1}, \ldots, \pi_{n})$ denotes a proof $\pi$ with last rule $\rnrule$ and immediate subproofs $\pi_{1}, \ldots, \pi_{n}$. It is assumed that the derivations are from a consistent $(T;E)$. R1 is applicable when $\func$ is a constructor rule, and ensures that $\DYderives$ subproofs do not end in a constructor rule. R2 and R3 eliminate some occurrences of $\rnsymm$, while R4 pushes $\rnsymm$ up towards the axioms. R5 and R6 ensure that no premise of $\rntrans$ is the conclusion of $\rneq$ or $\rntrans$. R7 ensures that adjacent premises of $\rntrans$ are not the result of $\rncons$. R8 simplifies proofs where $\rnproj$ follows $\rncons$. We will discuss R9 later. R10 ensures that the conclusion of $\rnlint$ is not a premise of $\rnlint$. In R11, $\pi_{i}$ proves an equality $\equals{v}{n}$, and it is weakened to a list membership of the form $v \listmemb \ell'$, but by consistency, even after intersection, the conclusion must be of the form $v\listmemb{\ell}$ where $\lambda(v)$ is an element of $\ell$ for some $\lambda$. Thus we can directly apply weakening to $\pi_{i}$ to get the same conclusion. 
        
    R9 requires some explanation. Let $\pi_{i}$ be the proof $\rncons_{\func}(\pi^{1}_{i}, \pi^{2}_{i})$, and let $\concof(\pi_{j})$ be $\equals{t_{j}}{t_{j+1}}$, for $1 \leq j < r$. We see that $\concof(\rntrans(\pi_{1}, \ldots, \pi_{r-1}))$ is $\equals{t_{1}}{t_{r}}$. Since $\rnproj$ is applied on this, there is some constructor $\gunc$ such that $t_{e} = \gunc(t^{1}_{e},  t^{2}_{e})$ for $e \in \{1, r\}$. Since $\pi_{i}$ ends in $\rncons_{\func}$, we see that $t_{e} = \func(t^{1}_{e}, t^{2}_{e})$ for $e \in \{i, i+1\}$. But $\equals{t_{1}}{t_{i}}$ is provable from $(T;E)$, which is consistent. Therefore it has to be the case that $\func = \gunc$. Thus we see that for all $e \in \{1, i, i+1, r\}$, $t_{e} = \func(t^{1}_{e}, t^{2}_{e})$. So we can rewrite the LHS of R9 to the RHS to get a valid proof. Note that we can apply $\rnproj$ on $\equals{t_{1}}{t_{i}}$ in the transformed proof since all components of $t_{1}$ and $t_{i}$ are abstractable -- for $t_{1}$ this is true because the $\rnproj$ rule was applied to $\equals{t_{1}}{t_{r}}$ in the proof on the LHS; and for $t_{i}$ this follows from the fact that $\pi^{1}_{i}$ (resp.\ $\pi^{2}_{i}$) derives $\equals{t^{1}_{i}}{t^{1}_{i+1}}$ (resp.\ $\equals{t^{2}_{i}}{t^{2}_{i+1}}$), and so by purity, $T \DYderives \{t^{1}_{i}, t^{2}_{i}\}$. For a similar reason, we can apply $\rnproj$ on $\equals{t_{i+1}}{t_{r}}$.

\begin{theorem} If $(T;E) \eqderives \alpha$ then there is a normal proof of $(T;E) \vdash \alpha$ in the $\eqderives$ system.
\end{theorem}
\begin{proof}
Let $\pi$ be any proof of $(T; E) \vdash \alpha$ such that all DY subproofs of $\pi$ are normal. 
    Suppose we repeatedly apply the transformations of Table~\ref{tab:rwrules} starting with $\pi$ and reach a proof $\varpi$ on which we can no longer apply any of the rules. Then $\varpi$ satisfies clauses 1 to 6 in the definition of normal proofs (since none of the rewrite rules, in particular R1--R7 and R10--R11, apply to $\varpi$). 
            
    Clause 7 is also satisfied by $\varpi$, for the following reason. Suppose a subproof $\varpi_{1}$ ends in $\rnproj$ and $\varpi_{2}$ is a maximal subproof of $\varpi_{1}$ ending in $\rncons$. $\varpi_{2}$ is a proper subproof of $\varpi_{1}$, so there has to be a subproof of $\varpi_{1}$ of the form $\rho = \rnrule(\cdots\varpi_{2}\cdots)$. Since $\rncons$ appears as the rule above $\rnrule$, a priori, $\rnrule$ can only be one of $\{\rnsymm, \rntrans, \rnproj, \rncons\}$. But since $\varpi_{2}$ is a \emph{maximal subproof} of $\varpi_{1}$ ending in $\rncons$, $\rnrule \neq \rncons$. Since R4 and R8 cannot be applied on $\varpi$, $\rnrule \notin \{\rnsymm, \rnproj\}$. But if $\rnrule = \rntrans$, then $\rho$ is a proper subproof of $\varpi_{1}$. In particular, it is the immediate subproof of some $\rho' = \rnrule'(\cdots\rho\cdots)$. Now $\rnrule'$ cannot be $\rnsubst$, since then $\concof(\rho')$ is a list membership assertion, which cannot occur in a proof ending in $\rnproj$. $\rnrule' \neq \rncons$, as that would violate the maximality of $\varpi_{2}$. $\rnrule' \notin \{\rnsymm, \rntrans, \rnproj\}$, since then one of the rewrite rules R4, R6, R8 would apply to $\varpi$. We have ruled out all possible cases for $\rnrule'$, and thus we are forced to conclude that $\varpi_{2}$ cannot be a subproof of $\varpi_{1}$. Thus, $\rncons$ does not occur in any subproof of $\varpi$ ending in $\rnproj$, and $\varpi$ satisfies all the clauses in the definition of normal proofs. 
    
    We next show that we can always reach a stage where no transformation is enabled. To begin with, apply the rules R2--R4 until the premise of each occurrence of $\rnsymm$ is the conclusion of an $\rnax$ or a $\rnprom$. None of the other rules converts a proof ending in $\rnax$ or $\rnprom$ to one which does not, so the above property is preserved even if we apply the other rules in any order. 
    
    Associate three sizes to an $\eqderives$-proof $\pi$: 
    \begin{itemize}
    \item $\measure_{1}(\pi)$ is the sum of the sizes of the $\DYderives$ subproofs of $\pi$, 
    \item $\measure_{2}(\pi)$ is the number of $\rncons$ rules that occur in $\pi$, and 
    \item $\measure_{3}(\pi)$ is the size of the proof $\pi$ (number of nodes in the proof tree). 
    \end{itemize}
    
    We also define $\measure(\pi) \coloneqq (\measure_{1}(\pi), \measure_{2}(\pi), \measure_{3}(\pi))$.

    We now show that if $\pi'$ is obtained from $\pi$ by one application of any of the transformation rules other than R2--R4, $\measure(\pi') < \measure(\pi)$.
    \begin{itemize}
        \item If R1 is applied, $\measure_{1}(\pi') < \measure_{1}(\pi)$ and so $\measure(\pi') < \measure(\pi)$.
        \item If R7 or R9 is applied, we have $\measure_{1}(\pi') \leq \measure_{1}(\pi)$ and $\measure_{2}(\pi') < \measure_{2}(\pi)$. Therefore, $\measure(\pi') < \measure(\pi)$.
        \item If R5, R6, R8, R10 or R11 is applied, we have that $\measure_{i}(\pi') \leq \measure_{i}(\pi)$ for $i \in \{1,2\}$ and $\measure_{3}(\pi') < \measure_{3}(\pi)$. So $\measure(\pi') < \measure(\pi)$.
    \end{itemize}
    Thus, once we apply R2--R4 till they can no longer be applied, we cannot have an infinite sequence of transformations starting from any $\pi$. Hence, every proof $\pi$ can be transformed into a normal proof $\varpi$ with the same conclusion. 
\end{proof}

We state and prove subterm property next. We use the following notation.
\begin{itemize}[leftmargin=*]
\item $\termsof(\pi) \coloneqq \{t \mid$ a subproof of $\pi$ derives $\alpha$ and $t$ is a maximal subterm of $\alpha\}$. 
\item $\listsof(E) \coloneqq \{\ell \mid \exists{t}:t\listmemb{\ell}$ is in $E\}$.
\item $\listsof(\pi) \coloneqq \{\ell \mid$ a subproof of $\pi$ derives $t\listmemb{\ell}\}$. 
\end{itemize}

\begin{theorem}[Subterm property]
For any normal proof $\pi$ of $T; E \eqderives \alpha$, $\termsof(\pi) \subseteq \st(T) \cup \st(E\cup\{\alpha\})$ and $\listsof(\pi) \subseteq \listsof(E \cup \{\alpha\}) \cup \{[n] \mid n \in \st(T) \cup \st(E\cup\{\alpha\})\}$. Further, if $\pi$ does not contain $\rncons$, then $\termsof(\pi) \subseteq \st(T) \cup \st(E)$ . Also, if $\pi$ does not end in $\rnweak$ and does not end in $\rnlint$, then $\listsof(\pi) \subseteq \listsof(E) \cup \{[n] \mid n \in \st(T) \cup \st(E)\}$.
\end{theorem}
We implicitly use the following easily provable facts. 
\begin{enumerate}[label=(F\arabic*)]
\item \label{item:f1} If a normal proof $\pi$ ends in $\rntrans$ and an immediate subproof $\varpi$ does not end in $\rncons$, then $\rncons$ does not occur in $\varpi$. 
\item \label{item:f2} If a normal proof $\pi$ derives a list membership assertion, $\rncons$ does not occur in $\pi$.
\end{enumerate}

\begin{proof}
    Let $\rnrule$ be the last rule of $\pi$. We have the following cases. We mention $\listsof(\pi)$ only in cases where the rules involve lists. 
    \begin{itemize}
        \item $\rnrule = \rnax$: $\alpha \in E$, so $\termsof(\pi) \subseteq \st(E)$ and $\listsof(\pi) \subseteq \listsof(E)$. 
        \item $\rnrule = \rneq$: $\alpha$ is $\equals{t}{t}$ and $T \DYderives t$. Since $\pi$ is a normal proof whose $\DYderives$ subproofs are also normal, $T \DYderives t$ via a proof ending in a destructor rule, and by subterm property for $\DYderives$, it follows that $t \in \st(T)$. Thus $\termsof(\pi) = \{t\} \subseteq \st(T)$.
        \item $\rnrule = \rnsymm$: $\termsof(\pi) = \termsof(\pi')$, where $\pi'$ is the immediate subproof, and the statement follows by IH.
        \item $\rnrule = \rncons$: $\alpha$ is $\equals{\func(t_{1},t_{2})}{\func(u_{1},u_{2})}$, and for $i \in \{1,2\}$, there is a subproof $\pi_{i}$ with conclusion $\equals{t_{i}}{u_{i}}$. By IH, $\termsof(\pi_{i}) \subseteq \st(T \cup \{t_{i},u_{i}\}) \cup \st(E) \subseteq \st(T) \cup \st(E \cup \{\alpha\})$ for $i \in \{1,2\}$. Thus $\termsof(\pi) \subseteq \st(T) \cup \st(E \cup \{\alpha\})$.
        \item $\rnrule = \rntrans$: Suppose the subproofs of $\pi$ are $\pi_{1}$ through $\pi_{k-1}$ with conclusions $\equals{t_{1}}{t_{2}}$ through $\equals{t_{k-1}}{t_{k}}$ respectively, and $\alpha = \equals{t_{1}}{t_{k}}$. Since $\pi$ is a normal proof, no two adjacent premises of $\rnrule$ are obtained by $\rncons$, and no premise of $\rnrule$ is obtained by $\rntrans$. 
        %Further, by~\ref{item:f1} any immediate subproof of $\pi$ that does not end in $\rncons$ does not have an occurrence of $\rncons$. Consider any $r \in \termsof(\pi)$. 
        The following cases arise.
        \begin{itemize}
            \item $r \in \{t_{1},t_{k}\}$. In this case, $r \in \st(\alpha)$.
            
            \item $r \in \termsof(\pi_{i})$, where $\pi_{i}$ does not end in $\rncons$. By~\ref{item:f1}, $\rncons$ does not occur in $\pi_{i}$. By IH, $r \in \st(T) \cup \st(E)$. 
            
            \item $r \in \termsof(\pi_{i})$, where $\pi_{i}$ ends in $\rncons$, and $1 < i < k-1$. Both $\pi_{i-1}$ and $\pi_{i+1}$ end in a rule other than $\rncons$, by normality of $\pi$. So, by~\ref{item:f1}, $\rncons$ does not occur in $\pi_{i-1}$ and $\pi_{i+1}$, and $t_{i}, t_{i+1} \in \termsof(\pi_{i-1}) \cup \termsof(\pi_{i+1}) \subseteq \st(T) \cup \st(E)$ (by IH on $\pi_{i-1}$ and $\pi_{i+1}$). So, by applying IH on $\pi_{i}$, we get $r \in \st(T) \cup \st(E \cup \{\equals{t_{i}}{t_{i+1}}\}) \subseteq \st(T) \cup \st(E)$.
            
            \item $r \in \termsof(\pi_{1})$, where $\pi_{1}$ ends in $\rncons$. By normality of $\pi$, we see that $\pi_{2}$ ends in a rule other than $\rncons$. So $\rncons$ does not occur in $\pi_{2}$. By IH on $\pi_{2 }$, $t_{2} \in \termsof(\pi_{2}) \subseteq \st(T) \cup \st(E)$. By IH on $\pi_{1}$, $r \in \st(T\cup\{t_{1}, t_{2}\}) \cup \st(E) \subseteq \st(T) \cup \st(E \cup \{\alpha\})$.

            \item $r \in \termsof(\pi_{k-1})$, where $\pi_{k-1}$ ends in $\rncons$. The proof is similar to the above. 
        \end{itemize}        
        \item $\rnrule = \rnproj$: Let $\alpha = \equals{t}{u}$, got from a proof $\pi'$ with conclusion $\equals{a}{b}$. Since $\pi$ is normal, $\rncons$ does not occur in $\pi$ (or in $\pi'$). By IH, $a,b \in \termsof(\pi') \subseteq \st(T) \cup \st(E)$. Since $t, u \in \st(\{a,b\})$, we have $\termsof(\pi) \subseteq \st(T) \cup \st(E)$.
        \item $\rnrule = \rnprom$: $\alpha$ is $\equals{t}{u}$, and the immediate subproof $\pi'$ proves $t\listmemb{[u]}$. $\pi'$ does not contain $\rncons$, and so by IH, $\termsof(\pi) = \termsof(\pi') \subseteq \st(T) \cup \st(E)$. Note that $\listsof(\pi) \subseteq \listsof(\pi') \cup \{[u]\}$, so the statement about lists is also true. 
        \item $\rnrule = \rnweak$: Let $\pi'$ be the immediate subproof. The result follows from IH and the fact that $\listsof(\pi) = \listsof(\pi') \cup \listsof(\alpha)$.
        \item $\rnrule = \rnlint$: All terms in the conclusion appear in some proper subproof, so the statement on terms follows by IH. None of the subproofs ends in $\rnlint$ or $\rnweak$ (and does not contain $\rncons$). Thus $\listsof(\pi') \subseteq \listsof(E) \cup \{[n] \mid n \in \st(T)\cup\st(E)]$, for every subproof $\pi'$. It follows that $\listsof(\pi) \subseteq \listsof(E \cup \{\alpha\}) \cup \{[n] \mid n \in \st(T) \cup \st(E \cup \{\alpha\})\}$. 
        
        \item $\rnrule = \rnsubst$: Let the major premise be $t\listmemb{\ell}$ and the minor premise be $\equals{t}{u}$. Both $t,u$ are from $\vars\cup\names$, and thus are in $\st(T) \cup \st(E)$. The result follows from IH.
        \item $\rnrule = \rnsays$: Let the major premise be $\beta$ and the minor premise be $\sk_{a}$. Since $T \DYderives \sk_{a}$, $\sk_{a} \in \st(T)$. And $\termsof(\pi) \subseteq \st(T) \cup \st(E) \cup \st(\beta) \cup \{\pk_{a}\} \subseteq \st(T) \cup \st(E \cup \{\alpha\})$.
        \qedhere
    \end{itemize}
\end{proof}
\end{document}